\documentclass[3p]{elsarticle}
\makeatletter
\def\ps@pprintTitle{%
 \let\@oddhead\@empty
 \let\@evenhead\@empty
 \def\@oddfoot{\centerline{\thepage}}%
 \let\@evenfoot\@oddfoot}
\makeatother

\date{}
\def\texpsfig#1#2#3{\vbox{\kern #3\hbox{\includegraphics{#1}\kern #2}}\typeout{(#1)}}

\usepackage{url,hyperref}
\usepackage{bbm}
\usepackage{eurosym}                           
\usepackage[latin1]{inputenc}                  
\usepackage{longtable}                         
\usepackage{array}                             
\usepackage{graphicx,color}
\usepackage{amsthm}
\usepackage{amsbsy}
\usepackage{subfig}

\usepackage{amssymb}
\usepackage{amsmath}
\usepackage{enumerate}
\usepackage{graphicx,color}
\usepackage{epsfig}
\usepackage[ruled,vlined]{algorithm2e}
\usepackage{multirow,bigdelim}
\usepackage[table]{xcolor}
\usepackage{natbib}
\usepackage[noabbrev]{cleveref}

\newcommand{\zbox}[1]{
\noindent
\begin{center}
\resizebox{0.98\textwidth}{!}{
\framebox[14.5cm]{
\begin{minipage}{14cm}
#1
\end{minipage}
}
}
\end{center}
}

\theoremstyle{plain}
\newtheorem{thm}{Theorem}[section]

\newtheorem{rem}[thm]{Remark}
\theoremstyle{remark}

\theoremstyle{plain}
\newtheorem{lem}[thm]{Lemma}

\newtheorem{prop}[thm]{Proposition}

\theoremstyle{definition}

\newcommand{\e}{{\rm e}}        
\def\R{\mathbb{ R}}             
\def\E{\mathbb{ E}}             
\def\Q{\mathbb{ Q}}             

\def\P{\mathbb{ P}}             
\def\C{\mathbb{ C}}             

\renewcommand{\d}{{\rm d}}      
\def\dW{{\rm d}W}               
\def\x{\rm x}
\def\dt{{\rm d}t}

\def\N{{\mathcal N}}

\def\T{{\rm T}}

\def\W{\rm W}

\def\g{\Tilde{g}}
\def\xib{\bar{\xi}}
\def\Dt{\Delta t}

\def\V{\mathbf{V}}
\def\v{\mathbf{v}}
\def\M{\mathbf{M}}
\def\T{\mathbf{T}}

\def\a{\mathbf{a}}
\def\bxi{\boldsymbol{\xi}}
\def\bal{\boldsymbol{\alpha}}
\def\p{\mathbf{p}}
\def\S{\hat{S}}
\def\A{\hat{A}}
\def\H{\Tilde{H}}
\def\fx{\texttt{fx}}
\def\fl{\texttt{fl}}
\DeclareMathOperator*{\esssup}{ess\,sup}

\def\1{{\mathbbm{1}}}            

\theoremstyle{plain}

\newtheorem{assumption}{Assumption}[section]

\newtheorem{result}{Result}[section]

\usepackage[margin=1cm]{caption}
\usepackage{mathdots}

\numberwithin{equation}{section}	     

\title{On Pricing of Discrete Asian and Lookback Options\\ under the Heston Model}

\begin{document}

\author[1]{Leonardo Perotti\corref{cor1}}
\ead{L.Perotti@uu.nl}
\author[1,2]{Lech A.~Grzelak}
\ead{L.A.Grzelak@uu.nl}
\cortext[cor1]{Corresponding author.}
\address[1]{Mathematical Institute, Utrecht University, Utrecht, the Netherlands}
\address[2]{Financial Engineering, Rabobank, Utrecht, the Netherlands}

\begin{abstract}
    \noindent We propose a new, data-driven approach for efficient pricing of -- fixed- and floating-strike -- discrete arithmetic Asian and Lookback options when the underlying process is driven by the Heston model dynamics. The method proposed in this article constitutes an extension of \citep*{FastSamplingBridge}, where the problem of sampling from time-integrated stochastic bridges was addressed. The model relies on the Seven-League scheme \citep*{SevenLeague}, where artificial neural networks are employed to ``learn'' the distribution of the random variable of interest utilizing stochastic collocation points \cite{StochasticCollocation}. The method results in a robust procedure for Monte Carlo pricing. Furthermore, semi-analytic formulae for option pricing are provided in a simplified, yet general, framework. The model guarantees high accuracy and a reduction of the computational time up to thousands of times compared to classical Monte Carlo pricing schemes. 
\end{abstract}

\maketitle


\section{Introduction}  
\label{sec:intro}

A non-trivial problem in the financial field is the pricing of  path-dependent derivatives, as for instance \emph{Asian} and \emph{Lookback} options. The payoffs of such derivatives are expressed as functions of the underlying process monitored over the life of the option. The monitoring can either be \emph{continuous} or \emph{discrete}. Depending on different specifications of the underlying dynamics,  only a few case-specific theoretical formulae exist. For example, under the Black-Scholes and Merton log-normal dynamics, closed-form formulae were derived for \emph{continuously-monitored geometric} Asian options (see, for instance, \cite{DEVREESE2010780}). In the same model framework, \cite{goldman1979path,conze1991path} derived an analytic formula for the \emph{continuously-monitored} Lookback option, using probabilistic arguments such as the \emph{reflection principle}. 
However, options whose payoffs are \emph{discretely-monitored} are less tractable analytically, and so approximations are developed, as for the discrete Lookback option under the lognormal dynamics \cite{heynen1995lookback}.
Furthermore, stochastic volatility frameworks are even more challenging for the pricing task, where no applicable closed-form theoretical solutions are known.

Whenever an exact theoretical pricing formula is not available, a rich literature on numerical methods and approximations exists. The three main classes of approaches are Monte Carlo (MC) methods (e.g. \cite{KEMNA1990113}), partial differential equations (PDEs) techniques (see the extensive work in \cite{OptionPricingPDE} and \cite{vecer2002unified}), and Fourier-inversion based techniques (among many relevant works, we report \cite{benhamou2002fast,AsianOptionChfLevy}). 

Monte Carlo methods are by far the most flexible approaches, since they neither require any particular assumption on the underlying dynamics, nor on the targeted payoff. Furthermore, they benefit from a straightforward implementation based on the discretization of the time horizon as, for instance, the well-known \emph{Euler-Maruyama} scheme. The cost to be paid, however, is typically a significant computational time to get accurate results. PDE approaches are more problem-specific since they require the derivation of the partial differential equation which describes the evolution of the option value over time. Then, the PDE is usually solved using finite difference methods. Fourier-inversion-based techniques exploit the relationship between the probability density function (PDF) and the characteristic function (ChF) to recover the underlying transition density by means of Fast Fourier Transform (FFT). Thanks to the swift algorithm for FFT, such methods produce high-speed numerical evaluation, but they are often problem-specific, depending on the underlying dynamics. \textcolor{black}{A relevant example is \cite{corsaro2019general}, where a numerical method is proposed for the pricing of discrete arithmetic Asian options under stochastic volatility model dynamics. In the same group of techniques, we also refer to \cite{kirkby2020efficient} where a unified framework for the pricing of discrete Asian options is described, allowing for regime-switching jump diffusion underlying dynamics. A further example is \cite{leitao2021ctmc}, where a close approximation of the classic Heston model is proposed -- via the CTMC model -- which allows for a general SWIFT-based (see \cite{leitao2018swift}) pricing approach with application to exotic derivatives, as discrete Asian options.}

In this article, we propose a \textcolor{black}{data-driven}\footnote{\textcolor{black}{The meaning of ``data-driven'' here is the one given in \cite{SevenLeague}. The (empirical) distribution of interest is computed for a set of structural parameters and stored. Such synthetic ``data'' are used to ``drive'' the training of a suitable model.}} extension to MC schemes that allows for efficient pricing of \emph{discretely-monitored Asian and Lookback options}, without losing the flexibility typical of MC methods. We develop the methodology in the complex framework of the stochastic volatility model of Heston \cite{heston1993closed}, with an extensive application to the case of \emph{Feller condition} not satisfied. Under this dynamics, we show how to price fixed- or floating-strike discrete Asian and Lookback options. Moreover, the pricing model is applied also for the challenging task of pricing options with both a fixed- and a floating-strike component. We underline that the strengths of the method are its speed and accuracy, coupled with significant flexibility. The procedure is, indeed, independent of the underlying dynamics (it could be applied, for instance, for any stochastic volatility model), and it is not sensitive to the targeted payoff.

Inspired by the works in \cite{SevenLeague,FastSamplingBridge}, the method relies on the technique of Stochastic Collocation (SC) \cite{StochasticCollocation}, which is employed to accurately approximate the targeted distribution by means of piecewise polynomials. Artificial neural networks (ANNs) are used for fast recovery of the coefficients which uniquely determine the piecewise approximation. Given these coefficients, the pricing can be performed in a ``MC fashion'' sampling from the target (approximated) distribution and computing the numerical average of the discounted payoffs. Furthermore, in a simplified setting, we provide a semi-analytic formula that allows to directly price options without the need of sampling from the desired distribution. In both situations (MC and semi-analytic pricing) we report a significant computational speed-up, without affecting the accuracy of the result which remains comparable with the one of highly expensive MC methods.

The remainder of the paper is as follows. In \Cref{sec: GeneralFramework}, we formally define \emph{discrete arithmetic Asian and Lookback options}, as well as the model framework for the underlying process. Then, in \Cref{sec: 7L}, the pricing model is described. Two different cases are considered, in increasing order of complexity, to handle efficiently both unconditional sampling (\Cref{ssec: SpecialPricing}) and conditional sampling (\Cref{ssec: GeneralPricing}) for pricing of discrete arithmetic Asian and Lookback options. \Cref{sec: ErrorAnalysis} provides theoretical support to the given numerical scheme. The quality of the methodology is also inspected empirically with several numerical experiments, reported in \Cref{sec: NumericalExperiments}. \Cref{sec: Conclusions} concludes.

\section{Discrete arithmetic Asian and Lookback options}  
\label{sec: GeneralFramework}

In a generic setting, given the present time $t_0\geq 0$, the payoff at time $T>t_0$ of a discrete arithmetic Asian or Lookback option, with underlying process $S(t)$, can be written as:
\begin{equation}
\label{eqn: PathDependentPayoff}
H_{\omega}(T; S) = \max\Big(\omega\big( A(S)- K_1S(T)-K_2\big), 0\Big),
\end{equation}
where $A(S)\equiv A(S(t); t_0 < t \leq T)$ is a deterministic function of the underlying process $S$, the constants $K_1,K_2\geq 0$ control the floating- and fixed-strikes of the option, and $\omega=\pm 1$. Particularly, discrete arithmetic Asian and Lookback options are obtained by setting the quantity $A(S)$ respectively as follows:
\begin{align}
\label{eqn: ALambdaDefinitions}
    A(S):=\frac{1}{N}\sum_{n\in I} S(t_n), \qquad\qquad A(S):=\omega\max_{n\in I} \omega S(t_n),
\end{align}
with $I=\{1,\dots,N\}$ a set of indexes, $t_1 < t_2 < \dots < t_N = T$, a discrete set of future monitoring dates, and $\omega$ as in (\ref{eqn: PathDependentPayoff}).

Note that in both the cases of discrete arithmetic Asian and Lookback options, $A(S)$ is expressed as a deterministic transformation of the underlying process' path, which is the only requirement to apply the proposed method. Therefore, in the paper, we refer always to the class of \emph{discrete arithmetic Asian options}, often just called \emph{Asian options}, for simplicity. However, the theory holds for both classes of products. Actually, the pricing model applies to any product with a \emph{path-dependent European-type payoff}, requiring only a different definition of $A(S)$.

\subsection{Pricing of arithmetic Asian options and Heston framework}
\label{ssec: PricingAsianOptions}

This section focuses on the risk-neutral pricing of arithmetic -- fixed- and floating-strike -- Asian options, whose payoff is given in (\ref{eqn: PathDependentPayoff}) with $A(S)$ in (\ref{eqn: ALambdaDefinitions}).
By setting $K_1=0$ or $K_2=0$, we get two special cases: the fixed- or the floating-strike arithmetic Asian option.
From \Cref{eqn: PathDependentPayoff}, for $K_1=0$, the simplified payoff of a \emph{fixed-strike arithmetic Asian option} reads:
\begin{equation}
\label{eqn: fixedPayoff}
    H_{\omega}^{\fx}(T;S) = \max\Big(\omega\big( A(S)-K_2\big), 0\Big),
\end{equation}
therefore, with the risk-neutral present value:
\begin{equation}
\label{eqn: FixPresentValue}
    V_{\omega}^{\fx}(t_0)=\frac{M(t_0)}{M(T)} \E_{t_0}^{\Q}\Big[\max\big(\omega\big( A(S)-K_2\big), 0\big)\Big],
\end{equation}
where we assume the money-savings account $M(t)$ to be defined through the deterministic dynamics $\d M(t)=rM(t)\dt$ with constant interest rate $r\geq 0$.
For $K_2=0$, however, the payoff in \Cref{eqn: PathDependentPayoff}, becomes the one of a \emph{floating-strike arithmetic Asian option}:
\begin{equation}
\label{eqn: floatPayoff}
    H_{\omega}^{\fl}(T;S) = \max\Big(\omega\big( A(S)-K_1S(T)\big), 0\Big).
\end{equation}
The payoff in (\ref{eqn: floatPayoff}) is less tractable when compared to the one in (\ref{eqn: fixedPayoff}), because of the presence of two dependent stochastic unknowns, namely $S(T)$ and $A(S)$.
However, a similar representation to the one in \Cref{eqn: FixPresentValue} can be achieved, allowing for a unique pricing approach in both cases. By a change of measure from the risk-neutral measure $\Q$ to the measure $\Q^S$ associated with the num\'eraire $S(t)$, i.e. the \emph{stock measure}, we prove the following proposition.
\begin{prop}[Pricing of floating-strike arithmetic Asian option under the stock measure]
\label{prop: floatAsianStockMeasure}
Under the stock measure $\Q^S$, the value at time $t_0 \geq 0$ of an arithmetic floating-strike Asian option, with maturity $T>t_0$ and future monitoring dates $t_n$, $n\in\{1,\dots,N\}$, reads: 
\begin{equation}
\label{eqn: floatPresentValue}
    V_{\omega}^{\mathtt{fl}}(t_0)=S(t_0)\E_{t_0}^{S}\Big[ \max\big(\omega\big( A^{\mathtt{fl}}(S)-K_1\big), 0\big)\Big],
\end{equation}
with ${A}^{\mathtt{fl}}(S)$, defined as:
\begin{align*}
\label{eqn: ALambdafloatDefinitions}
    A^{\mathtt{fl}}(S):=A\bigg(\frac{S(\cdot)}{S(T)}\bigg)=\frac{A(S)}{S(T)},  
\end{align*}
where $A(S)$ is defined in (\ref{eqn: ALambdaDefinitions}).
\end{prop}
\begin{proof}
For a proof, see \ref{ap: StockMeasureForfloatstrike}.
\end{proof}

Both the representations in \Cref{eqn: FixPresentValue,eqn: floatPresentValue} can be treated in the same way. This means that the present value of both a fixed- and a floating-strike Asian option can be computed similarly, as stated in the following proposition \textcolor{black}{(see, e.g., \cite{henderson2002equivalence})}.
\begin{prop}[Symmetry of fixed- and floating-strike Asian option present value]
\label{prop: FullyFixfloatOption}
Let us consider the process $S(t)$ and the money-savings account $M(t)$, for $t\geq t_0$. Then, the same representation holds for the value at time $t_0$ of both fixed- or floating-strike Asian options, with maturity $T>t_0$, underlying process $S(t)$, and future monitoring dates  $t_n$, $n\in\{1,\dots,N\}$.
The present value is given by:
\begin{equation}
\label{eqn: FixfloatPresentValue}
    V_{\omega}^{\lambda}(t_0)=
    \begin{cases}
    \frac{M(t_0)}{M(T)}\E_{t_0}^{\Q}\big[ \max\big(\omega\big( A(S)-K_2\big), 0\big)\big],\quad\text{for }\: \lambda = \mathtt{fx}, \\[3pt]
    S(t_0)\E_{t_0}^{S}\big[ \max\big(\omega\big( A^{\mathtt{fl}}(S)-K_1\big), 0\big)\big],\quad\text{for }\: \lambda=\mathtt{fl}.
    \end{cases}
\end{equation}
\end{prop}

\begin{proof}
The proof follows by direct comparison between \Cref{eqn: FixPresentValue,eqn: floatPresentValue}.
\end{proof}

When $K_1\neq 0$ and $K_2\neq 0$, \Cref{eqn: PathDependentPayoff} is the payoff of a fixed- \emph{and} floating-strikes arithmetic Asian option. Its present value does not allow any simplified representation, and we write it as the expectation of the discounted payoff under the risk-neutral measure $\Q$:
\begin{equation}
\label{eqn: GeneralPresentValue}
    V_{\omega}(t_0)=\frac{M(t_0)}{M(T)}\E_{t_0}^{\Q}\Big[\max\big(\omega\big({A}(S)-K_1S(T)-{K}_2\big), 0\big)\Big].
\end{equation}

By comparing \Cref{eqn: FixfloatPresentValue,eqn: GeneralPresentValue} a difference in the two settings is unveiled. \Cref{eqn: FixfloatPresentValue} is characterized by a unique unknown stochastic quantity $A(S)$, whereas in (\ref{eqn: GeneralPresentValue}) an additional term appears, namely the stock price at final time, $S(T)$. Furthermore, the two stochastic quantities in \eqref{eqn: GeneralPresentValue} are not independent.
This might suggest that different procedures should be employed for the different payoffs. In particular, in a MC setting, to value (\ref{eqn: FixfloatPresentValue}) we only have to sample from the \emph{unconditional} distribution of $A(S)$; while in (\ref{eqn: GeneralPresentValue}) the MC scheme requires dealing with both the sampling of $S(T)$ and the \emph{conditional} sampling of $A(S)|S(T)$.

Let us define the stochastic volatility dynamics of Heston for the underlying stochastic process $S(t)$, with initial value $S(t_0) = {S}_0$, through the following system of stochastic differential equations (SDEs):
\begin{align}
\label{eqn: StockDynamics}
		\d {S}(t)&=r{S}(t)\dt+\sqrt{v(t)}{S}(t)\d{W}_x(t), & {S}(t_0)&={S}_0,\\
\label{eqn: VolatilityDynamics}
\d v(t)&=\kappa\left(\bar{v}-v(t)\right)\dt+\gamma\sqrt{v(t)}\d{W}_v(t), & v(t_0)&=v_0,
\end{align}
with $r,\kappa, \bar{v}, v_0\geq 0$, $\gamma > 0$ the constant rate, the speed of mean reversion, the long-term mean of the variance process, the initial variance, and the volatility-of-volatility, respectively. $W_x(t)$ and $W_v(t)$ are Brownian Motions (BMs) under the risk-neutral measure $\Q$ with correlation coefficient \textcolor{black}{$\rho\in[-1,1]$}, i.e., $\d W_x(t)\d W_v(t)=\rho \dt$.\footnote{\textcolor{black}{We remark that in most real applications (with a few exceptions such as some commodities and FX rates) the correlation is negative. Furthermore, the phenomenon of ``moment explosion'' for certain choices of Heston parameters involving positive correlation is discussed in \cite{andersen2007moment}.}}
The dynamics in (\ref{eqn: StockDynamics}) and (\ref{eqn: VolatilityDynamics}) are defined in the risk-neutral framework. However, \Cref{prop: floatAsianStockMeasure} entails a different measure framework, whose dynamics still fall in the class of stochastic volatility model of Heston, with adjusted parameters \textcolor{black}{(see \Cref{prop: HestonUnderStockMeasure})}.

\section{Swift numerical pricing using deep learning}  
\label{sec: 7L}

This section focuses on the efficient pricing of discrete arithmetic Asian options in a MC setting. The method uses a Stochastic Collocation \cite{StochasticCollocation} (SC) based approach to approximate the target distribution. Then, artificial neural networks (ANNs) ``learn'' the proxy of the desired distribution, allowing for fast recovery \cite{SevenLeague,FastSamplingBridge}.

\subsection{``Compressing'' distribution with Stochastic Collocation}
\label{ssec: StochasticCollocation}

In the framework of MC methods, the idea of SC -- based on the \emph{probability integral transform}\footnote{Given the two random variables $X,Y$, with CDFs $F_X,F_Y$, it holds: $F_X(X)\overset{\d}{=}F_Y(Y)$.} -- is to approximate the relationship between a ``computationally expensive'' random variable, say $A$, and a ``computationally cheap'' one, say $\xi$. The approximation is then used for sampling. A random variable is ``expensive'' if its inverse CDF is not known in analytic form, and needs to be computed numerically. With SC, the sampling of $A$ is performed at the cost of sampling $\xi$ (see \cite{StochasticCollocation}). Formally, the following mapping is used to generate samples from $A$:
\begin{equation}
\label{eqn: SCMC2}
    A \overset{\d}{=} F_A^{-1}(F_{\xi}(\xi))=: g(\xi)\approx\g(\xi),
\end{equation}
with $F_A$ and $F_{\xi}$ being respectively the CDFs of $A$ and $\xi$, and the function $\g$ a suitable, easily evaluable approximation of $g$. The reason why we prefer $\g$ to $g$ is that, by definition, every evaluation of $g$ requires the numerical inversion of $F_A$, the CDF of $A$. 

Many possible choices of $\g$ exist. In \cite{StochasticCollocation,FastSamplingBridge}, $\g$ is an $(M-1)$-degree polynomial expressed in Lagrange basis, defined on \emph{collocation points} (CPs) $\bxi:=\{\xi_k\}_{k=1}^M$ computed as Gauss-Hermite quadrature nodes, i.e.:
\begin{equation}
\label{eqn: LagrangePolynomial}
    \g(x):=\sum_{k=1}^M a_k \ell_k(x), \qquad \ell_k(x):=\prod_{\substack{1\leq j\leq M\\j \neq k}} \frac{x-\xi_j}{\xi_k-\xi_j}, \quad k=1,\dots, M,
\end{equation}
where the coefficients $\mathbf{a}:=\{a_k\}_{k=1}^M$ of the polynomial in the Lagrange basis representation, called \emph{collocation values} (CVs), are derived by imposing the system of equations:
\begin{equation}
\label{eqn: ConditionsLagrangePoly}
    g(\xi_k)=\g(\xi_k)=:a_k, \quad k=1,\dots, M,
\end{equation}
which requires only $M$ evaluations of $g$.

In this work, we define $\g$ as a piecewise polynomial. Particularly, we divide the domain $\R$ of the random variable $\xi$ (which for us is standard normally distributed\footnote{The two main reasons for $\xi$ being standard normal are the availability of such a distribution in most of the computing tools, and the ``similarity'' between a standard normally r.v. and (the logarithm of) $A(S)$ (see \cite{StochasticCollocation} for more details).}) in three regions. In each region, we define $\g$ as a polynomial. In other words, given the partition of the real line:
\begin{equation}
\label{eqn: DefinitionInterExtrapDomains}
    \Omega_{-}\cup \Omega_{M}\cup \Omega_{+}:=(-\infty, -\bar\xi) \cup [-\bar\xi, \bar\xi] \cup (\bar\xi, +\infty),
\end{equation}
for a certain $\bar\xi>0$, $\g$ is specified as:
\begin{equation}
\label{eqn: DefinitionTildeG}
    \g(\xi):=
    g_-(\xi)\cdot\1_{\Omega_{-}}(\xi)+
    g_M(\xi)\cdot\1_{\Omega_{M}}(\xi)+
    g_+(\xi)\cdot\1_{\Omega_{+}}(\xi),
\end{equation}
where $\1_{(\cdot)}$ is the indicator function, and $g_-, g_M, g_+$ are suitable polynomials.
To ensure high accuracy in the approximation, $g_M$ is defined as a Lagrange polynomial of high-degree $M-1$. The CPs $\bxi$, which identify the Lagrange basis in (\ref{eqn: LagrangePolynomial}), are chosen as Chebyshev nodes in the bounded interval $\Omega_M=[-\xib, \xib]$ \cite{ChebyshevPOP2018}. Instead, the CVs $\mathbf{a}$ are defined as in (\ref{eqn: ConditionsLagrangePoly}). 
The choice of Chebyshev nodes allows for increasing the degree of the interpolation (i.e. the number of CPs and CVs), avoiding the \emph{Runge's phenomenon} within the interval $\Omega_M$ \cite{trefethen2019approximation}. However, the behavior of $g_M$ outside $\Omega_M$ is out of control. We expect the high-degree polynomial $g_M$ to be a poor approximation of $g$ in $\Omega_-$ and $\Omega_+$. Therefore, we define $g_-$ and $g_+$ as linear (or at most quadratic) polynomials, with degree $M_--1$ and $M_+-1$, built on the \emph{extreme} CPs of $\bxi$. 

Summarizing, $g_-$, $g_M$ and $g_+$ are all defined as Lagrange polynomials:
\begin{equation}
    \begin{aligned}
        g_{(\cdot)}(x)&:=\sum_{k\in I_{(\cdot)}} a_k \ell_k(x), \qquad \ell_k(x):=\prod_{\substack{j\in I_{(\cdot)}\\j \neq k}} \frac{x-\xi_j}{\xi_k-\xi_j}, \quad k\in I_{(\cdot)},
    \end{aligned}
\end{equation}
where the sets of indexes for $g_-$, $g_M$ and $g_+$ are $I_-=\{1,2\}$, $I_M=\{1,\dots,M\}$ and $I_+=\{M-1,M\}$, respectively (if a quadratic extrapolation is preferred, we get $I_-=\{1,2,3\}$, and $I_+=\{M-2,M-1,M\}$).

\begin{rem}[``Compressed'' distributions]
\label{rem: Compression}
The SC technique is a tool to ``compress'' the information regarding $A$ in (\ref{eqn: SCMC2}), into a small number of coefficients, the CVs $\mathbf{a}$. Indeed, the relationship between $A$ and $\mathbf{a}$ is bijective, provided the distribution of the random variable $\xi$, and the corresponding CPs $\boldsymbol\xi$ (or, equivalently, the Lagrange basis in (\ref{eqn: LagrangePolynomial})), are specified a priori.
\end{rem}

\subsection{Semi-analytical pricing of fixed- or floating-strike Asian options}
\label{ssec: SpecialPricing}

Let us first consider the pricing of the fixed- or the floating-strike Asian options. Both the products allow for the same representation in which the only unknown stochastic quantity is $A^{\lambda}(S)$, $\lambda\in\{\fx,\fl\}$, as given in Proposition \ref{prop: FullyFixfloatOption}. For the sake of simplicity, in the absence of ambiguity, we call $A^\lambda(S)$ just $A$. 

For pricing purposes, we can benefit from the SC technique presented in the previous section, provided we know the map $\g$ (or, equivalently, the CVs $\mathbf{a}$), i.e.:
\begin{equation}
\label{eqn: ApproxOptionValue}
\begin{aligned}
    V_{\omega}(t_0)&= C\: \E\Big[ \max\big(\omega\big( A-K\big), 0\big)\Big]\\
    &\approx C\: \E\Big[ \max\big(\omega\big( \g(\xi)-K\big), 0\big)\Big]=:\Tilde{V}_{\omega}(t_0),
\end{aligned}
\end{equation}
where $\xi$ is a standard normally distributed random variable, and $C$ is a constant coherent with \Cref{prop: FullyFixfloatOption}.
We note, that $\Tilde{V}_{\omega}(t_0)$ is the expectation of (the positive part of) polynomials of a standard normal distribution. Hence, a semi-analytic formula exists, in a similar fashion as the one given in \citep*{TruncatedMoments}.

\begin{prop}[Semi-analytic pricing formula]
\label{prop: semianalyticPrice}
Let $\Tilde{V}_\omega(t_0)$ (and $C$) be defined as in \Cref{eqn: ApproxOptionValue}, with $\g$ defined in (\ref{eqn: DefinitionTildeG}). Assume further that $\bal_{-}$, $\bal_M$ and $\bal_{+}$ are the coefficients in the canonical basis of monomials for the three polynomials $g_{-}$, $g_M$ and $g_{+}$ respectively of degree $M_{-}-1$, $M-1$ and $M_{+}-1$. Then, using the notation $a\vee b = \max(a,b)$, the following semi-analytic pricing approximation holds:
\begin{equation*}
\begin{aligned}
    \frac{\Tilde{V}_\omega(t_0)}{\omega C} 
    =& \Bigg[\sum_{i=0}^{M_{-\omega}-1}\omega^i\alpha_{-\omega,i}m_i(\omega c_K, -\xib \vee \omega c_K)\Bigg]\cdot \left(F_\xi(-\xib \vee \omega c_K) - F_\xi(\omega c_K)\right)\\
    +& \Bigg[\sum_{i=0}^{M-1}\omega^i\alpha_{M,i}m_i(-\xib \vee \omega c_K, \xib \vee \omega c_K)\Bigg]\cdot \left(F_\xi(\xib \vee \omega c_K) - F_\xi(-\xib \vee \omega c_K)\right)\\
    +& \Bigg[\sum_{i=0}^{M_{\omega}-1}\omega^i\alpha_{\omega,i}m_i(\xib \vee \omega c_K, +\infty)\Bigg]\cdot \left(1-F_\xi(\xib \vee \omega c_K)\right) - K \cdot \left(1-F_\xi(\omega c_K)\right),
\end{aligned}
\end{equation*}
where $F_{\xi}$ is the CDF of a standard normal random variable, $\xi$, $m_i(a,b):=\E[\xi^i|a \leq \xi \leq b]$\footnote{A recursive formula for the computation of $m_i(a,b)$ is given in \ref{ap: analyticPricing}.}, $c_K$ satisfies $K=\g(c_K)$, and $\omega = \pm 1$ according to the call/put case.
\end{prop}

\begin{proof}
For proof of the previous proposition, see \ref{ap: analyticPricing}.
\end{proof}

We note also that \Cref{prop: semianalyticPrice} uses as input in the pricing formula the coefficients in the canonical basis, not the ones in the Lagrange basis, $\a$, in (\ref{eqn: ConditionsLagrangePoly}).
\begin{rem}[Change of basis]
\label{rem: ChangeBasis}
Given a Lagrange basis identified by $M$ collocation points $\bxi$, any $(M-1)$-degree polynomial $g(\xi)$ is uniquely determined by the corresponding $M$ coefficients $\mathbf{a}$. A linear transformation exists that connects the $M$ coefficients in the Lagrange basis with the $M$ coefficients $\boldsymbol{\alpha}$ in the canonical basis of monomials. In particular, it holds:
\begin{equation*}
    \M\boldsymbol{\alpha}=\boldsymbol{a},
\end{equation*}
with $\M\equiv\M(\bxi)$ a $M\times M$ Vandermonde matrix with element $\M_{k,i}:=\xi_k^{i-1}$ in position $(k,i)$. The matrix $\M$ admits an inverse; thus, the coefficients $\bal$ in the canonical basis are the result of matrix-vector multiplication, provided the coefficients $\a$ in the Lagrange basis are known. Moreover, since the matrix $\M$ only depends on $\bxi$, its inverse can be computed a priori once the  CPs $\bxi$ are fixed.
\end{rem}

Proposition \ref{prop: semianalyticPrice} provides a semi-analytic formula for the pricing of fixed- or floating-strike Asian options. Indeed, it requires the inversion of the map $\g$ which typically is not available in analytic form. On the other hand, since both the CPs $\bxi$ and the CVs $\a$ are known, a proxy of $\g^{-1}$ is easily achievable by interpolation on the pairs of values $(a_k,\xi_k)$, $k=1,\dots,M$.

The last problem is to recover the CVs $\a$ (which identify $\g$) in an accurate and fast way.
We recall that, for $k=1,\dots,M$, each CV $a_k$ is defined in terms of the exact map $g$ and the CP $\xi_k$ by the relationship:
\begin{equation*}
    a_k:=g(\xi_k)=F_A^{-1}\big(F_{\xi}(\xi_k)\big).
\end{equation*}
The presence of $F_A^{-1}$ makes it impossible to directly compute $\a$ efficiently. On the other hand, by definition, the CVs $\a$ are quantiles of the random variable $A\equiv A(S)$, which depends on the parameters $\p$ of the underlying process $S$. As a consequence, there must exist some unknown mapping $H$ which links $\p$ to the corresponding $\a$. We approximate such a mapping from synthetic data setting a regression problem, which is solved with an ANN $\Tilde{H}$ (in the same fashion as in \citep*{SevenLeague,FastSamplingBridge}). We have the following mapping:
\begin{equation*}
    \begin{aligned}
        \p \mapsto \a:=H(\p)\approx\Tilde{H}(\p), \qquad \p\in\Omega_{\p},\:\: \a\in\Omega_{\a},
    \end{aligned}
\end{equation*}
with $\Omega_{\p}$ and $\Omega_{\a}$ being the spaces of the underlying model parameters and of the CVs, respectively, while the ANN $\H$ is the result of an optimization process on a synthetic training set\footnote{The synthetic data are generated via MC simulation, as explained in \Cref{sec: NumericalExperiments}.}:
\begin{equation}
\label{eqn: TrainingSet}
    \mathbf{T}=\Big\{(\p_i,\a_i): i\in\{1,\dots,N_{\texttt{pairs}}\}\Big\}.
\end{equation}
The pricing procedure is summarized in the following algorithm.
\zbox{
{\bf Algorithm: Semi-analytic pricing}
\label{alg: semiAnaliticPricing}
\begin{enumerate}
\itemsep0.0em
    \item Fix the $M$ collocation points $\bxi$.
    \item Given the parameters $\p$, approximate the $M$ collocation values, i.e. $\a\approx\Tilde{H}(\p)$.
    \item Given $\a$, compute the coefficients $\bal_-$, $\bal_M$ and $\bal_+$ for $g_-$, $g_M$ and $g_+$ (see \Cref{rem: ChangeBasis}).
    \item Given $K$ and $\a$, compute $c_K$ of \Cref{prop: semianalyticPrice} interpolating $\g^{-1}$ on $(\a,\bxi)$.
    \item Given the coefficients $\bal_-$, $\bal_M$, $\bal_+$, and $c_K$, use \Cref{prop: semianalyticPrice} to compute $\Tilde{V}_\omega(t_0)$.
\end{enumerate}
}

\subsection{Swift Monte Carlo pricing of fixed- and floating-strikes Asian options}
\label{ssec: GeneralPricing}

Let us consider the case of an option whose payoff has both a fixed- and a floating-strike. The present value of such a derivative is given by:
\begin{equation*}
    V_{\omega}(t_0)=\frac{M(t_0)}{M(T)}\E_{t_0}^{\Q}\Big[\max\big(\omega\big({A}(S)-K_1S(T)-{K}_2\big), 0\big)\Big],
\end{equation*}
hence the price of the option is a function of the two dependent quantities $A(S)$ and $S(T)$. This means that, even in a MC setting, the dependency between $A(S)$ and $S(T)$ has to be fulfilled.
Therefore, a different methodology with respect to the one proposed in the previous section needs to be developed.

Due to the availability of efficient and accurate sampling techniques for the underlying process $S$ at a given future time $T$ (we use the COS method \cite{TheBook} enhanced with SC \cite{StochasticCollocation}, and we call it COS-SC), the main issue is the sampling of the conditional random variable $A(S)|S(T)$.
This task is addressed in the same fashion as it is done in \citep*{FastSamplingBridge}, where ANNs and stochastic collocation are applied for the efficient sampling from {time-integral} of {stochastic bridges}\footnote{By \emph{stochastic bridge}, we mean any stochastic process conditional to both its initial and final values.}, namely $\int_{t_0}^T S(t) \dt$ given the value of $S(T)$.
The underlying idea here is the same since the random variable $A(S)$ is conditional to $S(T)$. 
Especially, in the previous sections we pointed out that the distribution of $A(S)$ has an unknown parametric form which depends on the set of Heston parameters $\p$. Similarly, we expect the distribution of $A|S(T)=\S$ to be parametric into the ``augmented'' set of parameters $\p_{\S}:=\p\cup \{\S\}$. Hence, there exists a mapping $H$ which links $\p_{\S}$ with the CVs, $\a_{\S}$, corresponding to the \emph{conditional} distribution $A(S)|S(T)=\S$. We approximate $H$ by means of a suitable ANN $\Tilde{H}$, getting the the following mapping scheme:
\begin{equation*}
    \begin{aligned}
        \p_{\S} \mapsto \a_{\S}:=H(\p_{\S})\approx\Tilde{H}(\p_{\S}), \qquad \p_{\S}\in\Omega_{\p_{S}},\:\: \a_{\S}\in\Omega_{\a_{S}},
    \end{aligned}
\end{equation*}
where $\Omega_{\p_{S}}$ and $\Omega_{\a_{S}}$ are respectively the spaces of the underlying model parameters (augmented with $\S$) and of the CVs (corresponding to the conditional distribution $A(S)|S(T)$), and $\H$ is the result of a regression problem on a suitable training set $\T$ (see \Cref{eqn: TrainingSet}).
We propose a first \emph{brute force} sampling scheme. 
\zbox{
{\bf Algorithm: Brute force conditional sampling and pricing}
\label{alg: NaiveConditionalSampling}
\begin{enumerate}
\itemsep0.0em
    \item Fix the $M$  collocation points $\bxi$.
    \item Given the parameters $\p$, for $j=1,\dots,N_{\texttt{paths}}$, repeat:
\begin{enumerate}
\itemsep0.0em
    \item generate the sample $\S_j$ from $S(T)$ (e.g. with COS-SC method \cite{TheBook} and \cite{StochasticCollocation});
    \item given $\p_{\S_j}$, approximate the $M$ conditional CVs, i.e. $\a_{\S_j}\approx\Tilde{H}(\p_{\S_j})$;
    \item given $\a_{\S_j}$, use SC to generate the conditional sample $\A_j$.
\end{enumerate}
\item Given the pairs $(\S_j,\A_j)$, for $j=1,\dots,N_{\texttt{paths}}$, and any desired $(K_1,K_2)$, evaluate:
\begin{equation*}
    V_{\omega}(t_0)\approx\frac{1}{N_{\texttt{paths}}}\frac{M(t_0)}{M(T)}\sum_{j=1}^{N_{\texttt{paths}}}\max\big(\omega\big({\A_j}-K_1\S_j-{K}_2\big), 0\big).
\end{equation*}
\end{enumerate}
}

Nonetheless, the brute force sampling proposed above requires $N_{\texttt{paths}}$ evaluations of $\H$ (see 2(b) in the previous algorithm).
This is a massive computational cost, even if a single evaluation of an ANN is high-speed. We can, however, benefit from a further approximation. We compute the CVs using $\H$ only at specific \emph{reference} values for $S(T)$. Then, the intermediate cases are derived utilizing (linear) interpolation. We choose a set of $Q$ equally-spaced values $\{S^1,S^2,\dots,S^Q\}$ for $S(T)$, defined as:
\begin{equation}
\label{eqn: ReferenceQuantiles}
    S^q:=S_{\min}+\frac{q-1}{Q-1}(S_{\max}-S_{\min}), \qquad q=1,\dots,Q,
\end{equation}
where the boundaries are quantiles corresponding to the probabilities $p_{\min},p_{\max}\in(0,1)$, i.e.  $S_{\min}:=F_{S(T)}^{-1}(p_{\min})$ and $S_{\max}:=F_{S(T)}^{-1}(p_{\max})$.

Calling $\p^q = \p_{S^q}$, and $\a^q = \a_{S^q}$, $q=1,\dots,Q$, we compute the grid $\mathbf{G}$ of \emph{reference} CPs, with only $Q$ ANN evaluations, namely:
\begin{equation}
\label{eqn: ReferenceGridCPs}
    \begin{bmatrix}
    \a^1\\
    \cdots \\
    \a^q\\
    \cdots \\
    \a^Q\\
    \end{bmatrix}=
    \begin{bmatrix}
    H(\p^1)\\
    \cdots \\
    H(\p^q)\\
    \cdots \\
    H(\p^Q)\\
    \end{bmatrix}\approx
    \begin{bmatrix}
    \H(\p^1)\\
    \cdots \\
    \H(\p^q)\\
    \cdots \\
    \H(\p^Q)\\
    \end{bmatrix}=:\mathbf{G},
\end{equation}
where $\a^q$, $H(\p^q)$ and $\H(\p^q)$, $q=1,\dots,Q$, are row vectors.
The interpolation on $\mathbf{G}$ is much faster than the evaluation of $\H$. Therefore, the grid-based conditional sampling results more efficient than the brute force one, particularly when sampling a huge number of MC samples.

The algorithm for the grid-based sampling procedure, to be used instead of point 2. in the previous algorithm, is reported here.

\zbox{
{\bf Algorithm: Grid-based conditional sampling}
\label{alg: NaiveConditionalSampling}
\begin{enumerate}
\itemsep0.0em
    \item[2.1.] Fix the \emph{boundary} probabilities $p_{\min},p_{\max}\in(0,1)$ and compute the \emph{boundary} quantiles $S_{\min}:=F_{S(T)}^{-1}(p_{\min})$ and $S_{\max}:=F_{S(T)}^{-1}(p_{\max})$ (e.g. with the COS method \cite{TheBook}).
    \item[2.2.] Compute the \emph{reference} values $S^q:=S_{\min}+\frac{q-1}{Q-1}(S_{\max}-S_{\min})$, $q=1,\dots,Q$.
    \item[2.3.] Given the ``augmented'' parameters $\p^q$, evaluate $Q$ times $\H$ to compute $\mathbf{G}$ (see (\ref{eqn: ReferenceGridCPs})).
    \item[2.4.] Given the parameters $\p$ and the grid $\mathbf{G}$, for $j=1,\dots,N_{\texttt{paths}}$, repeat:
\begin{enumerate}
\itemsep0.0em
    \item generate the sample $\S_j$ from $S(T)$ (e.g. with COS-SC method \cite{TheBook} and \cite{StochasticCollocation});
    \item given $\S_j$, approximate the $M$ conditional CVs, i.e. $\a_{\S_j}$, by interpolation in $\mathbf{G}$;
    \item given $\a_{\S_j}$, use SC to generate the conditional sample $\A_j$.
\end{enumerate}
\end{enumerate}
}

\section{Error analysis}
\label{sec: ErrorAnalysis}
This section is dedicated to the assessment and discussion of the error introduced by the main approximations used in the proposed pricing method.
Two primary sources of error are identifiable. The first error is due to the SC technique: in \Cref{ssec: StochasticCollocation} the exact map $g$ is approximated by means of the piecewise polynomial $\g$. 
The second one is a regression error, which is present in both \Cref{ssec: SpecialPricing,ssec: GeneralPricing}. ANNs $\H$ are used instead of the exact mappings $H$. For the error introduced by the SC technique, we bound the ``$L^2$-distance'', $\epsilon_{SC}$, between the exact distribution and its SC proxy showing that $\g\1_{\Omega_M}=g_M$ in $\Omega_M$ is an analytic function. $\epsilon_{SC}$ is used to provide a direct bound on the option price error, $\epsilon_P$. On the other hand, regarding the approximation of $H$ via $\H$ we provide a general convergence result for \texttt{ReLU}-architecture ANN, i.e. ANN with Rectified Linear Units as activation functions.

\subsection{Stochastic collocation error using Chebyshev polynomials}
Let us consider the error introduced in the methodology using the SC technique (\Cref{ssec: StochasticCollocation}), and investigate how this affects the option price. We restrict the analysis to the case of fixed- or floating-strike discrete arithmetic Asian and Lookback options (\Cref{ssec: SpecialPricing}). We define the error $\epsilon_P$ as the ``$L^1$-distance'' between the real price $V_{\omega}(t_0)$ and its approximation $\Tilde{V}_{\omega}(t_0)$, i.e.:
\begin{equation}
\label{eqn: PricingError}
    \epsilon_P:=\big|\Tilde{V}_\omega(t_0)-V_\omega(t_0)\big|.
\end{equation}

Given the standard normal kernel $\xi\sim\N(0,1)$, we define the SC error as the (squared) $L^2$-norm of $g-\g$, i.e.:
\begin{equation}
\label{eqn: L2error}
    \epsilon_{SC}:=\E\big[(g - \g)^2(\xi)\big].
\end{equation}
We decompose $\epsilon_{SC}$ accordingly to the piecewise definition of $\g$, namely:
\begin{equation*}
\begin{alignedat}{6}
    \epsilon_{SC} &= \E\big[(g\1_{\Omega_{-}}-g_-)^2(\xi)\big]
    &+\E\big[(g\1_{\Omega_{M}}-g_M)^2(\xi)\big]
    &+\E\big[(g\1_{\Omega_{+}}-g_+)^2(\xi)\big]\\
    &=:\epsilon_- +\epsilon_M +\epsilon_+.
\end{alignedat}
\end{equation*}
with the domains $\Omega_{(\cdot)}$ defined in \Cref{eqn: DefinitionInterExtrapDomains}, i.e., for $\xib>0$:
\begin{equation*}
    \Omega_{-}= (-\infty,-\xib),\qquad\Omega_{M}= [-\xib, \xib],\qquad\Omega_{+}= (\xib,+\infty). 
\end{equation*}

To deal with the ``extrapolation'' errors $\epsilon_-$ and $\epsilon_+$, we formulate the following assumption.
\begin{assumption}
\label{ass: ExtrapDecay}
The functions $(g\1_{\Omega_-} - g_{-})^2$ and $(g\1_{\Omega_+} - g_{+})^2$ are $O(\exp{x^2/2})$. Equivalently, $g^2\1_{\Omega_-}$ and $g^2\1_{\Omega_+}$ are $O(\exp{x^2/2})$ (since $g_{_-}$ and $g_{_+}$ are polynomials). \end{assumption}
Given \Cref{ass: ExtrapDecay} and the fact that $\xi\sim\N(0,1)$\footnote{Th PDF $f_{\xi}$ works as a dumping factor in $\epsilon_-=\E[(g\1_{\Omega_{-}}-g_-)^2(\xi)]$ and $\epsilon_+=\E[(g\1_{\Omega_{+}}-g_+)^2(\xi)]$.}, then the ``extrapolation'' errors $\epsilon_-$ and $\epsilon_+$ vanish, with exponential rate, as $\xib$ tends to infinity, i.e. $\epsilon_-=\epsilon_-(\xib)$, $\epsilon_+=\epsilon_+(\xib)$, and:
\begin{equation}
\label{eqn: ExtrapErrors}
\begin{cases}
\epsilon_-(\xib)\rightarrow 0,\\
\epsilon_+(\xib)\rightarrow 0,
\end{cases}
    \quad \text{for}\quad \xib\rightarrow +\infty.
\end{equation}
An illustration of the speed of convergence is reported in \Cref{fig: ExtrapError}. \Cref{fig: ExtrapError}a shows that the growth of $g\1_{\Omega_+}$ is (much) less than exponential (consistently with \Cref{ass: ExtrapDecay}), whereas \Cref{fig: ExtrapError}b illustrates the exponential decay of $\epsilon_-$ and $\epsilon_+$ when $\xib$ increases.

Therefore, if $\xib$ is taken sufficiently big, the error $\xi_{SC}$ in (\ref{eqn: L2error}) is mainly driven by the ``interpolation'' error $\epsilon_M$, whose estimate is connected to error bounds for Chebyshev polynomial interpolation, and it is the focus of the next part.

\begin{figure}[t]%
    \centering
    \subfloat[\centering]{{\includegraphics[width=6.8cm]{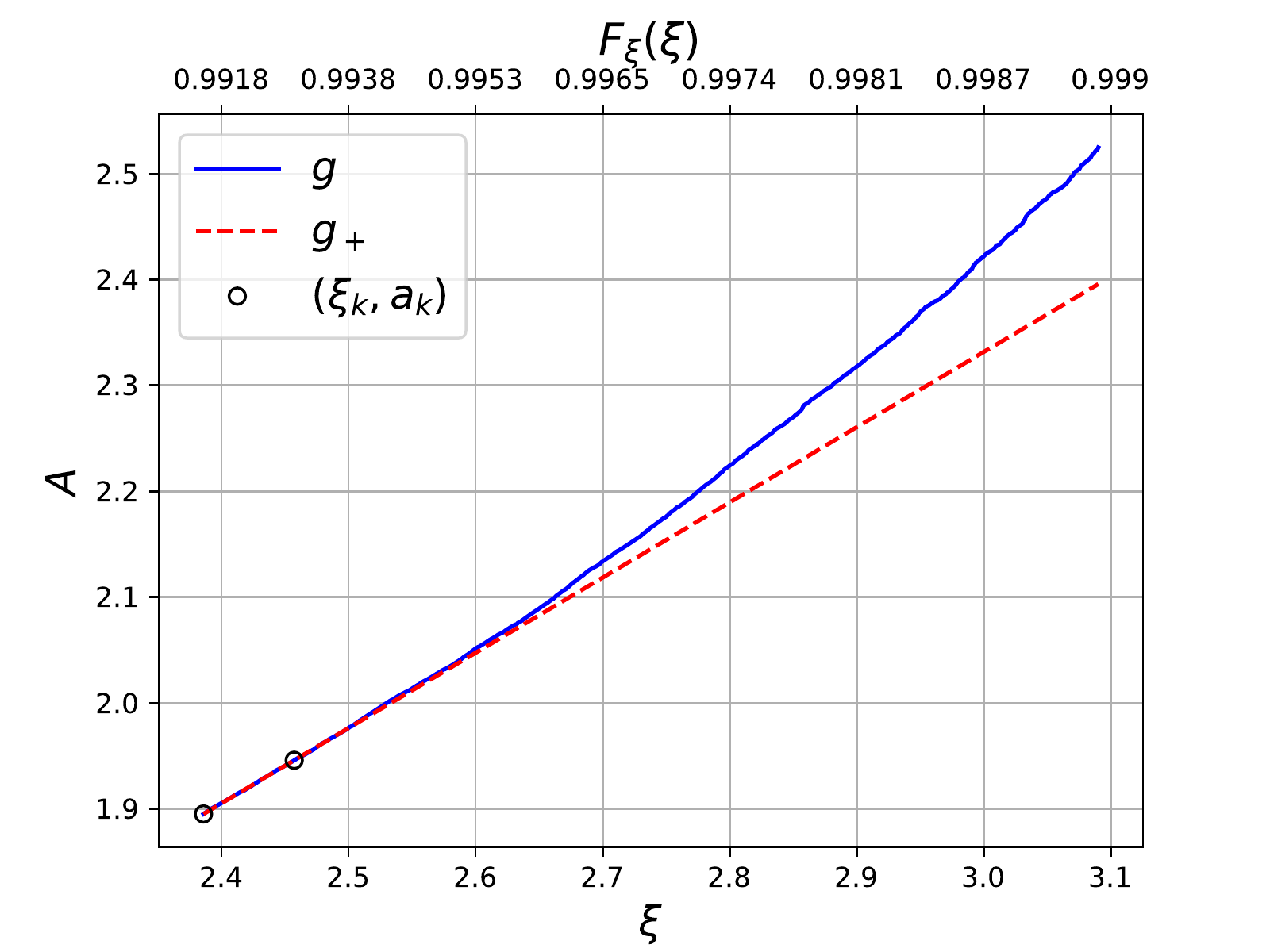} }}%
    ~
    \subfloat[\centering]{{\includegraphics[width=6.8cm]{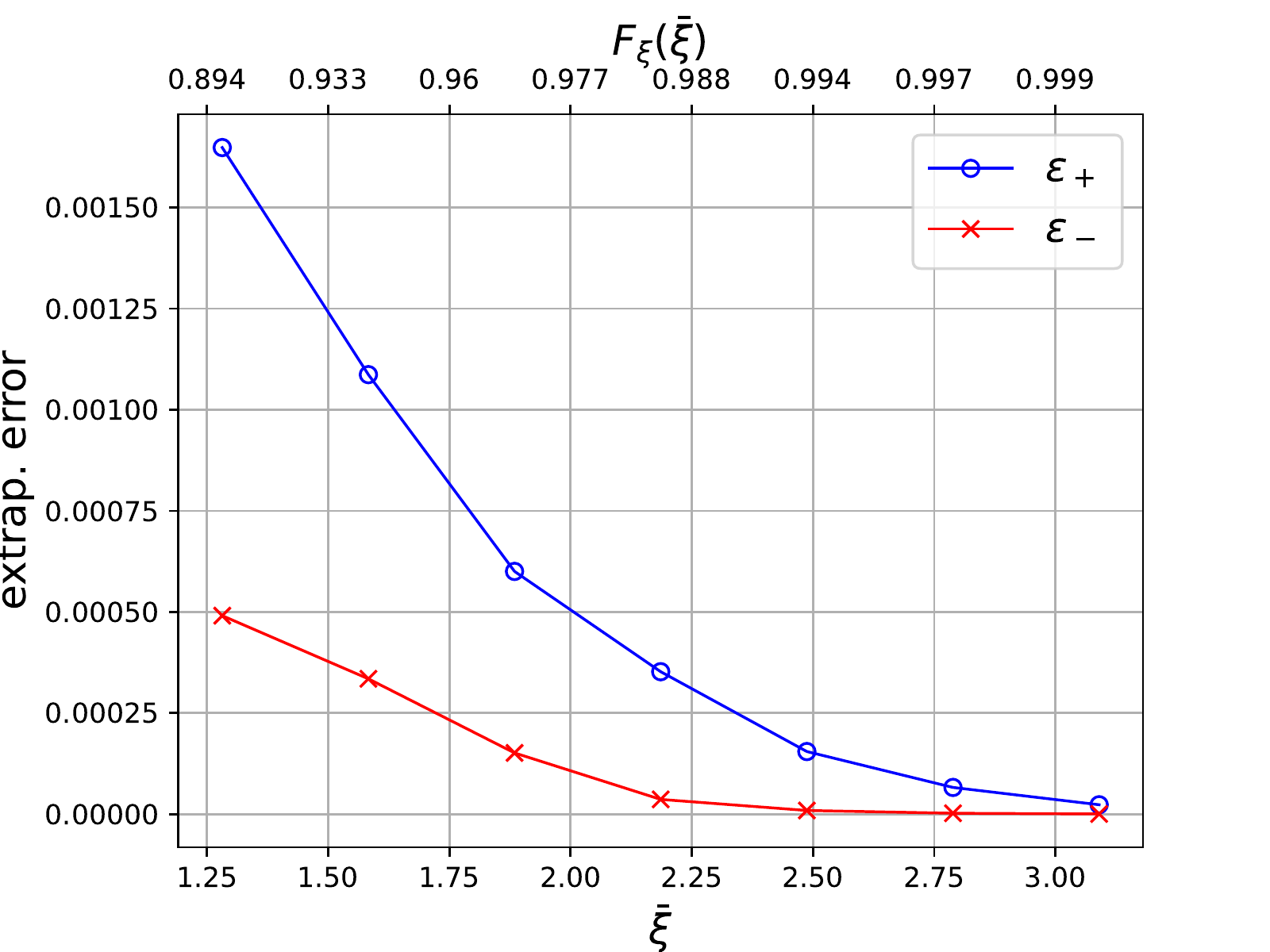} }}%
    \caption{\footnotesize Left: (slow) growth of $g\1_{\Omega_+}$ compared to the linear extrapolation $g_+$. Right: exponential decays of $\epsilon_-(\xib)$ and $\epsilon_+(\xib)$ in (\ref{eqn: ExtrapErrors}) when the $\xib$ increases. The upper x-axis represents the probability $F_\xi(\xib)$.}%
    \label{fig: ExtrapError}%
\end{figure}

\begin{thm}[Error bound for analytic function \citep*{ChebyshevPOP2018,ErrorBoundsCheb}]
\label{thm: BoundsAnalyticFunctions}
Let $f$ be a real function on $[-1, 1]$ and $f_M$ be its $(M-1)$-degree polynomial interpolation built on Chebyshev nodes $\xi_k:=\cos{\frac{k - 1}{M - 1}\pi}$, $k=1,\dots,M$. If $f$ has an analytic extension in a Bernstein ellipse $\mathcal{B}$ with foci $\pm 1$ and major and minor semiaxis lengths summing up to $\varrho > 1$ such that $\sup_{\mathcal{B}} |f| \leq \frac{\varrho - 1}{4}\overline{C}$ for some constant $\overline{C}>0$, then, for each $M \geq 1$, the following bound holds: 
\begin{equation*}
    ||f - f_M||_{L^{\infty}([-1, 1])}\leq \overline{C}\varrho^{1-M}.
\end{equation*}
\end{thm}

Since $g\1_{\Omega_M}$ is approximated by means of the $(M-1)$-degree polynomials $g_M$, built on Chebyshev nodes, to apply \Cref{thm: BoundsAnalyticFunctions}, we verify the required assumptions, namely the boundedness of $g\1_{\Omega_M}$ in $\Omega_M$ and its analyticity. 

We recall that:
\begin{equation}
\label{eqn: gDefinition}
    g=F_A^{-1}\circ F_{\xi},
\end{equation}
with $F_A$ and $F_{\xi}$ the CDFs of $A(S)$ and $\xi$, respectively.
Hence, the boundedness on the compact domain $\Omega_M$ is satisfied because the map $g$ is monotone increasing (as a composition of monotone increasing functions), and defined everywhere in $\Omega_M$. 

Furthermore, since the CDF of a standard normal, $F_{\xi}$, is analytic, from (\ref{eqn: gDefinition}) it follows that $g$ is analytic if $F^{-1}_A$ is analytic. The analyticity of $F^{-1}_A$ is fulfilled if $F_A$ is analytic and $F_A'=f_A$ does not vanish in the domain $\Omega_{M}$. Observe that, by restricting the domain to $\Omega_M$, the latter condition is trivially satisfied because we are ``far'' from the tails of $A(S)$ (corresponding to the extrapolation domains $\Omega_-$ and $\Omega_+$), and $F_A'$ do not vanish in other regions than the tails.

On the contrary, proving that $F_A$ is analytic is not trivial because of the lack of an explicit formula for $F_A$. However, it is beneficial to represent $F_A$ through the \emph{characteristc function} (ChF) of $A(S)$, $\phi_A$. For that purpose, we use a well-known inversion result.
\begin{thm}[ChF Inversion Theorem]
\label{res: ChFInversionTheorem}
Let us denote by $F$ and $\phi$ the CDF and the ChF of a given real-valued random variable defined on $\R$. Then, it is possible to retrieve $F$ from $\phi$ according to the inversion formula:
\begin{equation*}
\label{eqn: ChFInversionTheorem}
F(x) - F(0) = \frac{1}{2\pi} \int_{-\infty}^{+\infty} \phi(u)\frac{1-\e^{-i ux}}{i u} \d u,
\end{equation*}
with the integral being understood as a principal value.
\end{thm}
\begin{proof}
For detailed proof, we refer to \citep*{KendallAdvancedStatistics}.
\end{proof}
Thanks to \Cref{res: ChFInversionTheorem}, we have that if $\phi_A$ is analytic, so it is $F_A$ (as long as the integral in (\ref{eqn: ChFInversionTheorem}) is well-defined).
Thus, the problem becomes to determine if $\phi_A$ is analytic. We rely on a characterization of \emph{entire}\footnote{\emph{Entire} functions are \emph{complex analytic} functions in the whole complex plane $\C$.} ChFs, which can be used in this framework to show that in the cases of -- fixed- or floating-strike -- discrete arithmetic Asian and Lookback options, the (complex extension of the) function $\phi_A$ is analytic in a certain domain.

\begin{thm}[Characterization of entire ChFs \citep*{EntireChF}]
\label{thm: CharacterizationEntireChF}
Let $Y$ be a real random variable. Then, the complex function $\phi(z):=\E[e^{izY}]$, $z\in \C$, is entire if and only if the absolute moments of $Y$ exist for any order, i.e. $\E[|Y|^k]< +\infty$ for any $k\in \mathbb{N}$, and the following limit holds:
\begin{equation}
\label{eqn: LimitMomentsChF}
    \lim_{k\rightarrow +\infty} \left(\frac{|\E[Y^k]|}{k!}\right)^{\frac{1}{k}}=0.
\end{equation}
\end{thm}
\begin{proof}
A reference for proof is given in \citep*{EntireChF}.
\end{proof}

When dealing with the Heston model, there is no closed-form expression for the moments of the underlying process $S(t)$, as well as for the moments of its transform $A(S)$. Nonetheless, a \emph{conditional} case can be studied and employed to provide a starting point for a convergence result. 

\begin{prop}[Conditional ChF $\phi_{A|\V}$ is entire]
\label{prop: EntireConditionalChF}
Let us define the $N$-dimensional random vector $\V$, with values in $\Omega_{\V}:=\R_+^N$, as:
\begin{equation*}
    \V:= [I_{v}(t_1),I_{v}(t_2),\dots,I_{v}(t_N)]^T,\qquad I_v(t_n):=\int_{t_0}^{t_n}v(\tau)\d\tau, \quad n=1,\dots,N.
\end{equation*}
Let the complex conditional characteristic function $\phi_{A|\V}(z):=\E[e^{izA}|\V]$, $z\in\C$, be the extended ChF of the conditional random variable $A|\V$, with $A\equiv A(S)$ as given in \Cref{eqn: ALambdaDefinitions}.

Then, $\phi_{A|\V}(z)$ is entire.
\end{prop}
\begin{proof}
    See \ref{ap: ErrorAnalysis}.
\end{proof}

From now on, using the notation of \Cref{prop: EntireConditionalChF}, we consider the following assumption on the tail behavior of the random vector $(\V,A)$ satisfied. Informally, we require that the density of the joint distribution of $(\V,A)$ has uniform (w.r.t. $\V$) exponential decay for $A$ going to $+\infty$.
\begin{assumption}
\label{ass: VADecay}
There exists a point $z^*\in \C$, $z^*=x^* - i y^*$, with $x^*,y^*\in \R$ and $y^*>0$, such that:
\begin{equation*}
    \int_{\Omega_{\V} \times \R^+} \e^{y^*a} \d F_{\V, A}(\v, a) < +\infty,
\end{equation*}
with $F_{\V, A}(\cdot,\cdot)$ the joint distribution of the random vector $\V$ and the random variable $A$.
\end{assumption}
Thanks to \Cref{ass: VADecay}, the ChF $\phi_A(z)$ is well defined for any $z\in \mathcal{S}_{y^*}\subset \C$, with the strip $\mathcal{S}_{y^*}:= \R +i \cdot [-y^*,y^*]$. Moreover, applying Fubini's Theorem, for any $z\in \mathcal{S}_{y^*}$, we have:
\begin{equation}
\label{eqn: ChFDecomposition}
    \phi_A(z)= \int_{\Omega_{\V}} \phi_{A|\V=\v} (z) \d F_{\V}(\v).
\end{equation}
Thus, we can show that the ChF $\phi_A(z)$ is analytic in the strip $\mathcal{S}_{y^*}$ (the details are given in \ref{ap: analyticPricing}).
\begin{prop}[ChF $\phi_A$ is analytic]
\label{prop: ChFAnalytic}
Let $\phi_A(z):=\E[\e^{izA}]$, $z\in \mathcal{S}_{y^*}$, with  $A\equiv A(S)$. Then, $\phi_A(z)$ is analytic in $\mathcal{S}_{y^*}$.
\end{prop}
\begin{proof}
A proof is given in \ref{ap: analyticPricing}.
\end{proof}

Thanks to \Cref{prop: ChFAnalytic}, and consistently with the previous discussion, we conclude that the map $g$ in (\ref{eqn: gDefinition}), is analytic on the domain $\Omega_M$. Therefore, we can apply \Cref{thm: BoundsAnalyticFunctions}, which yields the following error estimate:
\begin{equation*}
    ||g\1_{\Omega_{M}} - g_M||_{L^\infty (\Omega_{M})} \leq \overline{C}\varrho^{1-M},
\end{equation*}
for certain $\varrho > 1$ and $\overline{C} > 0$. As a consequence, the following bound for the $L^2$-error $\epsilon_M$ holds:
\begin{equation}
\label{eqn: L2bound}
    \epsilon_M = \E\big[(g\1_{\Omega_{M}} - g_M)^2 (\xi)\big] \leq \overline{C}^2 \varrho^{2-2M}.
\end{equation}
Furthermore, the exponential convergence is also confirmed numerically, as reported in \Cref{fig: ChebError}. In \Cref{fig: ChebError}a we can appreciate the improvement in the approximation of $\g\1_{\Omega_M}$ by means of $g_M$, when $M$ is increased, whereas \Cref{fig: ChebError}b reports the exponential decay of $\epsilon_M$.

Using (\ref{eqn: L2bound}), the $L^2$-norm of $g-\g$, $\epsilon_{SC}$ in (\ref{eqn: L2error}), is bounded by:
\begin{equation*}
\label{eqn: SCbound}
    \epsilon_{SC} \equiv \epsilon_{SC}(\xib, M) = \E\big[(g(\xi) - \g(\xi))^2\big] \leq \epsilon_{-}(\xib) + \overline{C}^2 \varrho^{2-2M} + \epsilon_{+}(\xib),
\end{equation*}
which goes to zero when $\xib\in \R^+$ and $M\in\mathbb{N}$ tend to $+\infty$. Therefore, for any $\epsilon>0$ there exist $\xib^*\in\R^+$ and $M^*\in \mathbb{N}$ such that:
\begin{equation}
\label{eqn: SpecificError}
    \epsilon_{SC}(\xib^*, M^*) < \epsilon^2,
\end{equation}
and because of the exponential decay, we expect $\xib^*$ and $M^*$ do not need to be taken too big.

Eventually, we can benefit from the bound in (\ref{eqn: SpecificError}) to control the pricing error, $\epsilon_P$ in (\ref{eqn: PricingError}). By employing the well-known inequality $\max(a+ b, 0)\leq \max(a,0)+\max(b, 0)$ and the Cauchy-Schwarz inequality, we can write:
\begin{equation*}
    \begin{aligned}
    \Tilde{V}_\omega(t_0)=\E\Big[ \Big(\omega\big(\g(\xi)-K\big)\Big)^+\Big] &\leq \sqrt{\E\Big[ \big(\g(\xi)-g(\xi)\big)^2\Big]} + \E\Big[ \Big(\omega\big(g(\xi)-K\big)\Big)^+\Big]\\
    &\leq \sqrt{\epsilon_{SC}(\xib,M)} + {V}_\omega(t_0),
\end{aligned}
\end{equation*}
and using the same argument twice (exchanging the roles of $g$ and $\g$), we end up with the following bound for the option price error:
\begin{equation*}
    \epsilon_P \leq \sqrt{\epsilon_{SC}(\xib^*,M^*)}\leq{\epsilon},
\end{equation*}
with $\xib^*$ and $M^*$ as in (\ref{eqn: SpecificError}).

\begin{figure}[t]%
    \centering
    \subfloat[\centering]{{\includegraphics[width=6.8cm]{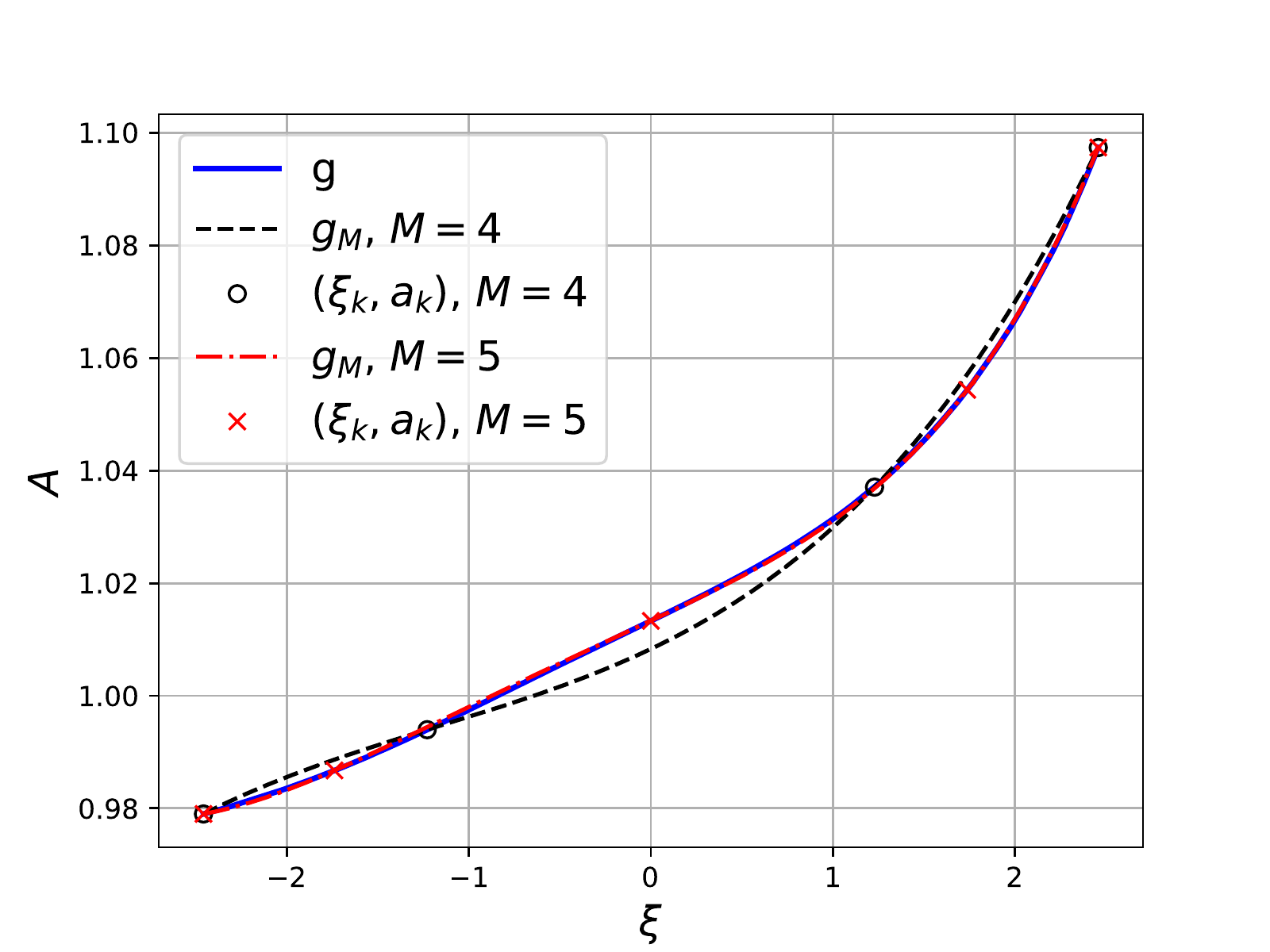} }}%
    ~
    \subfloat[\centering]{{\includegraphics[width=6.8cm]{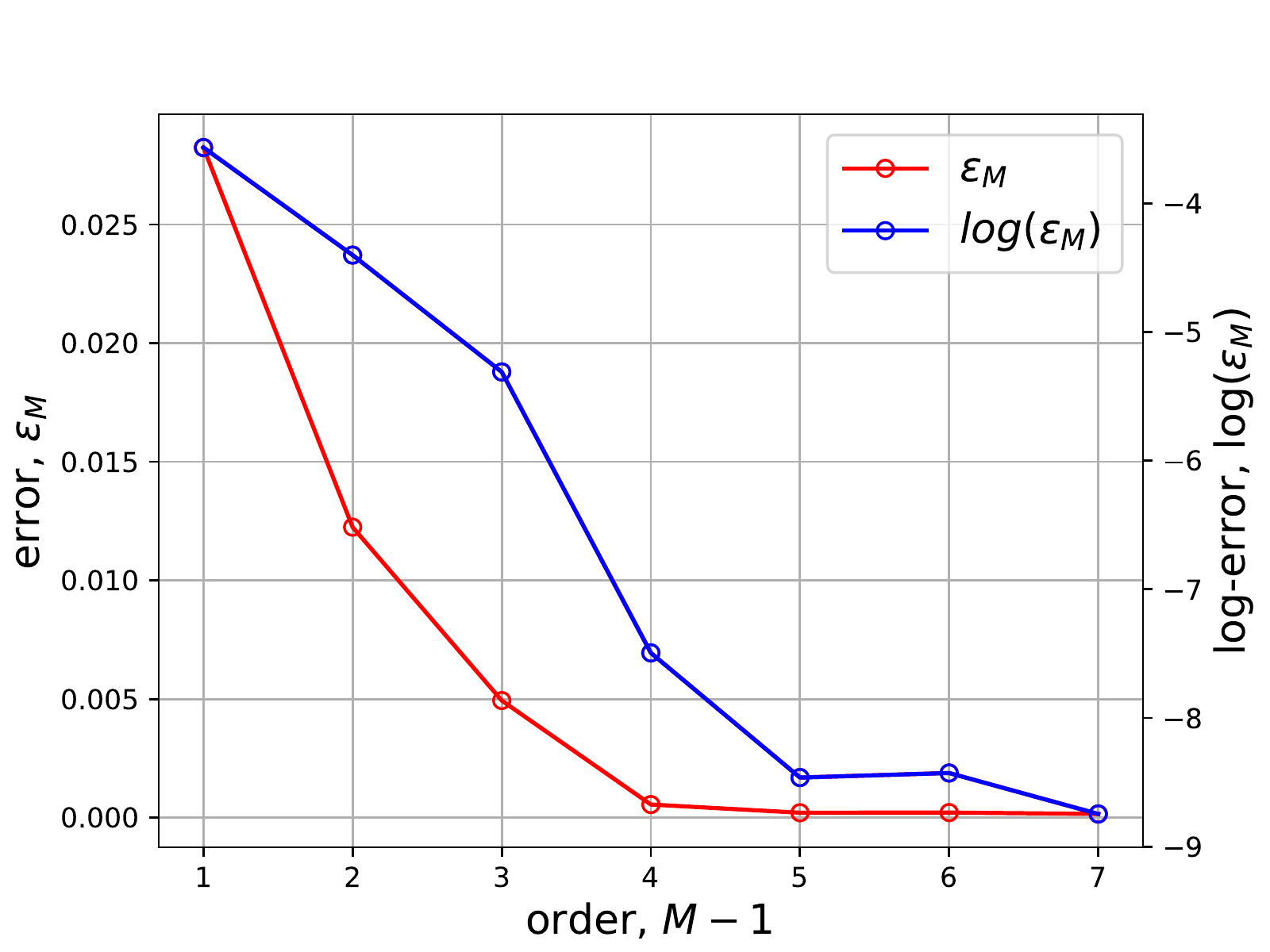} }}%
    \caption{\footnotesize Left: exact map $g\1_{\Omega_M}$ (blue) compared to the interpolation $g_M$ in the domain $\Omega_{M}$, for $M=4,5$. Right: $L^{2}$-error $\epsilon_M$ in (\ref{eqn: L2bound}) exponential decay in the order of the polynomial $g_M$.}%
    \label{fig: ChebError}%
\end{figure}

\subsection{Artificial neural network regression error}
As the final part of the error analysis, we investigate when ANNs are suitable approximating maps. In particular, we focus on ANNs with \texttt{ReLU}-architectures, namely ANNs whose activation units are all Rectified Linear Units defined as $\phi(x)=x \1_{x>0}(x)$. 

Consider the {Sobolev space} $\big(\mathcal{W}^{n,\infty}\big([0,1]^d\big), ||\cdot||_{n,d}^\infty\big)$, with $n,d\in\mathbb{N}\backslash\{0\}$, namely the space of functions $C^{n-1}\big([0,1]^d\big)$ whose derivatives up to the $(n-1)$-th order are all Lipschitz continuous, equipped with the norm $||\cdot||_{n,d}^\infty$ defined as:
\begin{equation}
    ||f||_{n,d}^\infty = \max_{|\mathbf{n}|\leq n} \esssup_{\mathbf{x}\in[0,1]^d} |D^{\mathbf{n}}f(\mathbf{x})|,
\end{equation}
with $\mathbf{n}:=(n_1,\dots,n_d)\in\mathbb{N}^d$, $|\mathbf{n}|=\sum_{i=1}^d n_i$, and $D^{\mathbf{n}}$ the weak derivative operator.
Furthermore, we define the unit ball $B_{n,d}:=\big\{f\in\mathcal{W}^{n,\infty}\big([0,1]^d\big): ||f||_{n,d}^\infty\leq 1\big\}$. Then, the following approximation result holds:
\begin{thm}[Convergence for \texttt{ReLU} ANN]
\label{thm: ConvergenceANN}
For any choice of $d,n\in\mathbb{N}\backslash\{0\}$ and $\epsilon\in(0,1)$, there exists an architecture $H(\x|\cdot)$ based on \texttt{\texttt{ReLU}} (Rectified Linear Unit) activation functions $\phi$, i.e. $\phi(x)=x \1_{x>0}(x)$, such that:
\begin{enumerate}
    \item $H(\x|\cdot)$ is able to approximate any function $f\in B_{n,d}$ with an error smaller than $\epsilon$, i.e., there exists a matrix of weights $\W$ such that $||f(\cdot)-H(\cdot|\W)||_{\textit{n},\textit{d}}^\infty< \epsilon$;
    \item $H$ has at most $c(\ln{1/\epsilon} + 1)$ layers and at most $c\epsilon^{-d/n}(\ln{1/\epsilon}+1)$ weights and neurons, with $c=c(d,n)$ an appropriate constant depending on $d$ and $n$.
\end{enumerate}
\end{thm}

\begin{proof}
A proof is available in \citep*{ErrorBoundANN}.
\end{proof}

Essentially, \Cref{thm: ConvergenceANN} states that there always exists a \texttt{ReLU}-architecture (with bounded number of layers and activation units) suitable to approximate at any desired precision functions with a certain level of regularity (determined by $(\mathcal{W}^{n,\infty}([0,1]^d), ||\cdot||_{n,d}^\infty)$).

\begin{rem}[Input scaling]
\label{rem: InputScaling}
We emphasize that although \Cref{thm: ConvergenceANN} applies to (a subclass of \emph{sufficiently} regular) functions with domain the $d$-dimensional hypercube $[0,1]^d$, this is not restrictive. Indeed, as long as the regularity conditions are fulfilled, \Cref{thm: ConvergenceANN} holds for any function defined on a $d$-dimensional hyperrectangle since it is always possible to linearly map its domain into the $d$-dimensional hypercube.
\end{rem}
Furthermore, we observe that all the results of convergence for ANN rely on the assumption that the training is performed successfully, and so the final error in the optimization process is negligible.
Under this assumption, \Cref{thm: ConvergenceANN} provides a robust theoretical justification for using \texttt{ReLU}-based ANNs as regressors. The goodness of the result can also be investigated empirically, as shown in the next section (see, for instance, \Cref{fig: AccuracyANN}).

\section{Numerical experiments}  
\label{sec: NumericalExperiments}

In this part of the paper, we detail some numerical experiments. We focus on applying the methodology given in \Cref{ssec: SpecialPricing} for the numerical pricing of fixed-strike discrete arithmetic Asian and Lookback options. We address the general case of discrete arithmetic Asian options described in \Cref{ssec: GeneralPricing}. For each pricing experiment errors and timing results are given. The ground truth benchmarks are computed via MC using the \emph{almost exact simulation} of the Heston models, detailed in \Cref{res: AESimulation}.

All the computations are implemented and run on a MacBook Air (M1, 2020) machine, with chip Apple M1 and 16 GB of RAM. The code is written in \texttt{Python}, and \texttt{torch} is the library used for the design and training of the ANN, as in \citep*{FastSamplingBridge}.

\subsection{A benchmark from the literature}

 \begin{table}[b!]
\vspace{-0.3cm}
\begin{center}
\caption{Training set for benchmark replication.}
\vspace{-0.2cm}
\label{tab: TrainingSetBenchmark}
\resizebox{0.58\textwidth}{!}{
\begin{tabular}{lllllll}
\hline\hline
\multicolumn{1}{c|}{$\p$} & \multicolumn{1}{c}{$N_{\texttt{val}}$} & \multicolumn{1}{c}{range} & \multicolumn{1}{c}{met.} & \multicolumn{1}{c}{} & \multicolumn{2}{c}{MC} \\\cline{1-4}\cline{6-7}
\multicolumn{1}{c|}{$r$} & \multicolumn{1}{c}{100} & \multicolumn{1}{c}{[0.04, 0.06]} & \multicolumn{1}{c}{LHS} & \multicolumn{1}{c}{} & \multicolumn{1}{c|}{$N_{\texttt{paths}}$} & \multicolumn{1}{c}{$10^6$} \\
\multicolumn{1}{c|}{$\kappa$} & \multicolumn{1}{c}{100} & \multicolumn{1}{c}{[2.90, 3.10]} & \multicolumn{1}{c}{LHS} & \multicolumn{1}{c}{} & \multicolumn{1}{c|}{$\Delta t$} & \multicolumn{1}{c}{$\frac{1}{800}$} \\
\multicolumn{1}{c|}{$\gamma$} & \multicolumn{1}{c}{100} & \multicolumn{1}{c}{[0.08, 0.12]} & \multicolumn{1}{c}{LHS}\\
\multicolumn{1}{c|}{$\rho$} & \multicolumn{1}{c}{100} & \multicolumn{1}{c}{[-0.11, -0.09]} & \multicolumn{1}{c}{LHS} & \multicolumn{1}{c}{}  & \multicolumn{2}{c}{SC}\\\cline{6-7}
\multicolumn{1}{c|}{$\Bar{v}$} & \multicolumn{1}{c}{100} & \multicolumn{1}{c}{[0.03, 0.05]} & \multicolumn{1}{c}{LHS} & \multicolumn{1}{c}{}  & \multicolumn{1}{c|}{$M$} & \multicolumn{1}{c}{$21$} \\
\multicolumn{1}{c|}{$v_0$} & \multicolumn{1}{c}{100} &\multicolumn{1}{c}{[0.03, 0.05]}   & \multicolumn{1}{c}{LHS} & \multicolumn{1}{c}{} & \multicolumn{1}{c|}{$\Bar{\xi}$} & \multicolumn{1}{c}{$2.46$} \\
\multicolumn{1}{c|}{$T$} & \multicolumn{1}{c}{25} &\multicolumn{1}{c}{[0.25, 0.28]} & \multicolumn{1}{c}{EQ-SP} & \multicolumn{1}{c}{} & \multicolumn{1}{c|}{$F_{\Bar{\xi}}(\Bar{\xi})$} & \multicolumn{1}{c}{$0.993$} \\\cline{1-4}
\multicolumn{1}{c|}{$N_{\texttt{pairs}}$}&\multicolumn{3}{c}{$2500 \quad(100 \times 25)$} &\multicolumn{1}{c}{}
\end{tabular}}
\end{center}
\vspace{-0.3cm}
\end{table}

\begin{table}[t!]
\begin{center}
\caption{\footnotesize Results benchmark (BM) replication (see, Table 5 from \cite{kirkby2020efficient}).}
\label{tab: BenchmarkReplio}
\vspace{-0.2cm}
\resizebox{0.75\textwidth}{!}{\begin{tabular}{lllllllll}
\hline\hline
\multicolumn{1}{c}{} & \multicolumn{7}{c}{Set BM} \\\cline{1-8}
\multicolumn{1}{c}{} & \multicolumn{1}{c}{$r$} & \multicolumn{1}{c}{$\kappa$} & \multicolumn{1}{c}{$\gamma$} & \multicolumn{1}{c}{$\rho$} & \multicolumn{1}{c}{$\Bar{v}$} & \multicolumn{1}{c}{$v_0$} & \multicolumn{1}{c}{$T$}\\
\multicolumn{1}{c}{} & \multicolumn{1}{c}{\multirow{1}{*}{0.05}} & \multicolumn{1}{c}{\multirow{1}{*}{3.0}} & \multicolumn{1}{c}{\multirow{1}{*}{0.1}} & \multicolumn{1}{c}{\multirow{1}{*}{-0.1}} & \multicolumn{1}{c}{\multirow{1}{*}{0.04}} & \multicolumn{1}{c}{\multirow{1}{*}{0.04}} & \multicolumn{1}{c}{\multirow{1}{*}{0.25}}\\\hline\hline
\multicolumn{1}{c}{}&\multicolumn{3}{c}{BM}&\multicolumn{3}{c}{SC}&\multicolumn{1}{c}{SA}\\\cline{1-8}
\multicolumn{1}{c|}{$K_2$} & \multicolumn{1}{c}{V} & \multicolumn{2}{c|}{95CI} & \multicolumn{1}{c}{V} & \multicolumn{2}{c|}{95CI} & \multicolumn{1}{c}{V}\\\cline{1-8}
\multicolumn{1}{c|}{90} & \multicolumn{1}{c}{10.5439} & \multicolumn{2}{c|}{[10.5329, 10.5550]} & \multicolumn{1}{c}{10.5459} & \multicolumn{2}{c|}{[10.5349, 10.5570]} & \multicolumn{1}{c}{10.5486}\\
\multicolumn{1}{c|}{95} & \multicolumn{1}{c}{6.0168} & \multicolumn{2}{c|}{[6.0069, 6.0267]} & \multicolumn{1}{c}{6.0173} & \multicolumn{2}{c|}{[6.0074, 6.0271]} & \multicolumn{1}{c}{6.0196}\\
\multicolumn{1}{c|}{100} & \multicolumn{1}{c}{2.6026} & \multicolumn{2}{c|}{[2.5953, 2.6098]} & \multicolumn{1}{c}{2.5992} & \multicolumn{2}{c|}{[2.5920, 2.6064]} & \multicolumn{1}{c}{2.5996}\\
\multicolumn{1}{c|}{105} & \multicolumn{1}{c}{0.7902} & \multicolumn{2}{c|}{[0.7862, 0.7943]} & \multicolumn{1}{c}{0.7867} & \multicolumn{2}{c|}{[0.7827, 0.7907]} & \multicolumn{1}{c}{0.7865}\\
\multicolumn{1}{c|}{110} & \multicolumn{1}{c}{0.1622} & \multicolumn{2}{c|}{[0.1604, 0.1639]} & \multicolumn{1}{c}{0.1612} & \multicolumn{2}{c|}{[0.1595, 0.1630]} & \multicolumn{1}{c}{0.1615}
\end{tabular}}
\end{center}
\end{table}

\textcolor{black}{To assess the quality of the proposed methodology, we compare the method against the benchmarks available in Table 5 from \cite{kirkby2020efficient}. In the experiment, we consider prices of 5 discrete Asian call options with $n=201$ equally spaced monitoring dates from time $t_0=0$ to $T=0.25$. The underlying initial value is $S(t_0)=100$, while the other Heston parameters (Set BM) as well as the target strikes are given in \Cref{tab: BenchmarkReplio}. To produce those results, a toy model has been trained based on the ranges provided in \Cref{tab: TrainingSetBenchmark}. The ANN employed consists of 2 hidden layers each one with 20 hidden units. The results are computed with $N_{\texttt{paths}}=10^6$ Monte Carlo paths (consistent with the benchmark \cite{kirkby2020efficient}) and are presented in \Cref{tab: BenchmarkReplio}. For both the benchmark and the SC technique are reported the absolute value (V) and the $95\%$ confidence interval (95CI) of the option price. For the SA technique, instead, only the absolute value is reported since no sampling is involved and so there is no information on the variance of the estimate. All the results are within the BM $95\%$ confidence interval, and hence confirm the high accuracy of the proposed method.}

\subsection{Experiments' specifications}
Among the three examples of applications presented, two of them rely on the technique given in \Cref{ssec: SpecialPricing}, while the third is based on the theory in \Cref{ssec: GeneralPricing}. The first experiment is the pricing of \emph{fixed-strike discrete arithmetic Asian options} (FxA) with an underlying stock price process following the Heston dynamics. The second example, instead, is connected to the ``interest rate world'', and is employed for the pricing of \emph{fixed-strike discrete Lookback swaptions} (FxL). We assume the underlying swap rate is driven by a \emph{displaced} Heston model with drift-less dynamics, typically used for interest rates. The last one is an application to the pricing of \emph{fixed- and floating-strikes discrete arithmetic Asian options} (FxFlA) on a stock price driven by the Heston dynamics.
In the first (FxA) and last experiment (FxFlA), $A(S)$ in (\ref{eqn: ALambdaDefinitions}) is specified as:
\begin{align}
\label{eqn: SpecsAsian}
    A(S)=\frac15 \sum_{n=1}^5 S(t_n), \qquad t_n:=T - (5 - n)\tau_\texttt{A}, \quad n=1,\dots,5,
\end{align}
with $\tau_\texttt{A}=\frac{1}{12}$ as monitoring time lag, and $T>4\tau_\texttt{A}$ as option maturity. Differently, in the second experiment (FxL) $A(S)$ is given by:
\begin{align}
\label{eqn: SpecsLook}
    A(S)=\min_{n} S(t_n), \qquad t_n:=T - (30 - n)\tau_\texttt{L}, \quad n=1,\dots,30,
\end{align}
with $\tau_\texttt{L}=\frac{1}{120}$ as monitoring time lag, and $T>29\tau_\texttt{L}$ as option maturity. Observe that, assuming the unit is 1 year, with 12 identical months and 360 days, the choices of $\tau_\texttt{A}$ and $\tau_\texttt{L}$ correspond respectively to 1 month and 3 days of time lag in the monitoring dates.

\subsection{Artificial neural network development}
In this section, we provide the details about the generation of the training set, for each experiment, and the consequent training of the ANN used in the pricing model.

\subsubsection{Training set generation}
The training sets are generated through MC simulations, using the {almost exact} sampling in \Cref{res: AESimulation}.
In the first two applications (FxA and FxL) the two training sets are defined as in (\ref{eqn: TrainingSet}), and particularly they read:
\begin{align*}
    \mathbf{T}_{\texttt{FxA}}&=\Big\{\Big(\{r,\kappa,\gamma,\rho,\bar{v},v_0,T\}_i,\{a_1,\dots,a_{21}\}_i\Big): i\in\{1,\dots,N_{\texttt{pairs}}^{\texttt{FxA}}\}\Big\},\\
    \mathbf{T}_{\texttt{FxL}}&=\Big\{\Big(\{\kappa,\gamma,\rho,\bar{v},v_0,T\}_i,\{a_1,\dots,a_{21}\}_i\Big): i\in\{1,\dots,N_{\texttt{pairs}}^{\texttt{FxL}}\}\Big\}.
\end{align*}
The Heston parameters, i.e. $\p\backslash\{T\}$\footnote{$\p\backslash\{T\}=\{r,\kappa,\gamma,\rho,\bar{v},v_0\}$ and $\p\backslash\{T\}=\{\kappa,\gamma,\rho,\bar{v},v_0\}$ for FxA and FxL, respectively.}, are sampled using \emph{Latin Hypercube Sampling} (LHS), to ensure the best filling of the parameter space \cite{SevenLeague,FastSamplingBridge}. For each set $\p\backslash\{T\}$, $N_{\texttt{paths}}$ paths are generated, with a time step of $\Dt$ and a time horizon up to $T_{\max}$. The underlying process $S$ is monitored at each time $T$ for which there are enough past observations to compute $A(S)$, i.e.:
\begin{equation*}
    \begin{aligned}
        T&\geq 4\tau_{\texttt{A}} + \Dt, \quad\text{for \texttt{FxA}},\\
        T&\geq 29\tau_{\texttt{L}} + \Dt, \quad\text{for \texttt{FxL}}.
    \end{aligned}
\end{equation*}
Consequently, the product between the number of Heston parameters' set and the number of available maturities determines the magnitude of the two training sets (i.e., $N_{\texttt{pairs}}^{\texttt{FxA}}$ and $N_{\texttt{pairs}}^{\texttt{FxL}}$). 

For each $\p$, the CVs $\a$ corresponding to $A(S)$ are computed as:
\begin{equation*}
    a_k:=F_A^{-1}\big(F_{\xi}(\xi_k)\big)\approx Q_A\big(F_{\xi}(\xi_k)\big),\qquad k\in\{1,\dots,21\},
\end{equation*}
where $Q_A$ is the empirical quantile function of $A(S)$, used as a numerical proxy of $F_A^{-1}$, and $\xi_k$ are the CPs computed as Chebyshev nodes:
\begin{equation*}
    \xi_k:=-\xib\cos{\frac{k - 1}{20}\pi},\qquad k\in\{1,\dots,21\},
\end{equation*}
with $\xib:=F_\xi^{-1}(0.993)\approx 2.46$. We note that the definition of $\xib$ avoids any CV $a_k$ to be ``deeply'' in the tails of $A(S)$, which are more sensitive to numerical instability in a MC simulation.

The information about the generation of the two training sets is reported in \Cref{tab: TrainingSetRangesUnconditional}. Observe that $\T_{\texttt{FxA}}$ is richer in elements than $\T_{\texttt{FxL}}$ because of computational constraints. Indeed, the higher number of monitoring dates of $A(S)$ in FxL makes the generation time of $\T_{\texttt{FxL}}$ more than twice the one of $\T_{\texttt{FxA}}$ (given the same number of pairs). 

\begin{table}[t]
\vspace{-0.3cm}
\begin{center}
\caption{Training sets $\T_{\texttt{FxA}}$ and $\T_{\texttt{FxL}}$ generation details.}
\label{tab: TrainingSetRangesUnconditional}
\vspace{-0.2cm}
\resizebox{0.92\textwidth}{!}{
\begin{tabular}{llllllllll}
\hline\hline
\multicolumn{1}{c}{} & \multicolumn{3}{c}{FxA} &  \multicolumn{3}{c}{FxL}\\
\multicolumn{1}{c|}{$\p$} & \multicolumn{1}{c}{$N_{\texttt{val}}$} & \multicolumn{1}{c}{range} & \multicolumn{1}{c|}{met.} & \multicolumn{1}{c}{$N_{\texttt{val}}$} & \multicolumn{1}{c}{range} & \multicolumn{1}{c}{met.} & \multicolumn{1}{c}{} & \multicolumn{2}{c}{MC} \\\cline{1-7}\cline{9-10}
\multicolumn{1}{c|}{$r$} & \multicolumn{1}{c}{700} & \multicolumn{1}{c}{[0.00, 0.05]} & \multicolumn{1}{c|}{LHS} & \multicolumn{4}{c}{} & \multicolumn{1}{c|}{$N_{\texttt{paths}}$} & \multicolumn{1}{c}{$10^6$} \\
\multicolumn{1}{c|}{$\kappa$} & \multicolumn{1}{c}{700} & \multicolumn{1}{c}{[0.20, 1.10]} & \multicolumn{1}{c|}{LHS} & \multicolumn{1}{c}{300} &
\multicolumn{1}{c}{[0.80, 1.60]}   & \multicolumn{1}{c}{LHS} & \multicolumn{1}{c}{} & \multicolumn{1}{c|}{$\Dt$} & \multicolumn{1}{c}{$\frac{1}{120}$} \\
\multicolumn{1}{c|}{$\gamma$} & \multicolumn{1}{c}{700} & \multicolumn{1}{c}{[0.80, 1.10]} & \multicolumn{1}{c|}{LHS} & \multicolumn{1}{c}{300} &
\multicolumn{1}{c}{[0.40, 1.00]} & \multicolumn{1}{c}{LHS} & \multicolumn{1}{c}{} & \multicolumn{1}{c}{} & \multicolumn{1}{c}{} \\
\multicolumn{1}{c|}{$\rho$} & \multicolumn{1}{c}{700} & \multicolumn{1}{c}{[-0.95, -0.20]} & \multicolumn{1}{c|}{LHS} & \multicolumn{1}{c}{300} &
\multicolumn{1}{c}{[-0.80, -0.30]} & \multicolumn{1}{c}{LHS} & \multicolumn{1}{c}{} & \multicolumn{2}{c}{SC}\\\cline{9-10}
\multicolumn{1}{c|}{$\Bar{v}$} & \multicolumn{1}{c}{700} & \multicolumn{1}{c}{[0.02, 0.15]} & \multicolumn{1}{c|}{LHS} & \multicolumn{1}{c}{300} & \multicolumn{1}{c}{[0.10, 0.20]} & \multicolumn{1}{c}{LHS} & \multicolumn{1}{c}{} & \multicolumn{1}{c|}{$M$} & \multicolumn{1}{c}{$21$} \\
\multicolumn{1}{c|}{$v_0$} & \multicolumn{1}{c}{700} & \multicolumn{1}{c}{[0.02, 0.15]} & \multicolumn{1}{c|}{LHS} & \multicolumn{1}{c}{300} &
\multicolumn{1}{c}{[0.10, 0.20]}   & \multicolumn{1}{c}{LHS} & \multicolumn{1}{c}{} & \multicolumn{1}{c|}{$\xib$} & \multicolumn{1}{c}{$2.46$} \\
\multicolumn{1}{c|}{$T$} & \multicolumn{1}{c}{160} & \multicolumn{1}{c}{[0.34, 1.67]} & \multicolumn{1}{c|}{EQ-SP} & \multicolumn{1}{c}{170} &
\multicolumn{1}{c}{[0.26, 1.67]} & \multicolumn{1}{c}{EQ-SP} & \multicolumn{1}{c}{} & \multicolumn{1}{c|}{$F_{\xi}(\xib)$} & \multicolumn{1}{c}{$0.993$}\\\cline{1-7}
\multicolumn{1}{c|}{$N_{\texttt{pairs}}$}&\multicolumn{3}{c|}{$112000 \quad(700 \times 160)$}& \multicolumn{3}{c}{$51000 \quad(300 \times 170)$}
\end{tabular}}
\end{center}
\vspace{-0.3cm}
\end{table}

Since in the general procedure (see \Cref{ssec: GeneralPricing}) ANNs are used to learn the \emph{conditional distribution} $A(S)|S(T)$ (not just $A(S)$!), the third experiment requires a training set which contains also information about the conditioning value, $S(T)$. We define $\T_{\texttt{FxFlA}}$ as:
\begin{align*}
    \mathbf{T}_{\texttt{FxFlA}}&=\Big\{\Big(\{r,\kappa,\gamma,\rho,\bar{v},v_0,T, S^q,p^q\}_i,\{a_1,\dots,a_{14}\}_i\Big): i\in\{1,\dots,N_{\texttt{pairs}}\}\Big\},
\end{align*}
where $p^q$ is the probability corresponding to the quantile $S^q$, given as in (\ref{eqn: ReferenceQuantiles}), i.e.:
\begin{equation*}
    S^q=S_{\min}+\frac{q-1}{14}(S_{\max}-S_{\min}), \qquad q\in\{1,\dots,15\},
\end{equation*}
with the $S_{\min}=F_{S(T)}^{-1}(p_{\min})$ and $S_{\max}=F_{S(T)}^{-1}(p_{\max})$. Heuristic arguments drove the choice of adding in the input set $\p$ the probability $p^q:=F_{S(T)}(S^q)$, i.e. the probability \emph{implied} by the final value $p^q$. Indeed, the ANN training process results more accurate when both $S^q$ and $p^q$ are included in $\p$.

\begin{table}[t]
\vspace{-0.3cm}
\begin{center}
\caption{Training set $\T_{\texttt{FxFlA}}$ generation details.}
\vspace{-0.2cm}
\label{tab: TrainingSetRangesConditional}
\resizebox{0.65\textwidth}{!}{
\begin{tabular}{lllllllll}
\hline\hline
\multicolumn{1}{c|}{$\p$} & \multicolumn{1}{c}{$N_{\texttt{val}}$} & \multicolumn{1}{c}{range} & \multicolumn{1}{c}{met.} & \multicolumn{1}{c}{} & \multicolumn{2}{c}{MC} \\\cline{1-4}\cline{6-7}
\multicolumn{1}{c|}{$r$} & \multicolumn{1}{c}{180} & \multicolumn{1}{c}{[0.00, 0.05]} & \multicolumn{1}{c}{LHS} & \multicolumn{1}{c}{} & \multicolumn{1}{c|}{$N_{\texttt{tot}}$} & \multicolumn{1}{c}{$3\cdot10^6$} \\
\multicolumn{1}{c|}{$\kappa$} & \multicolumn{1}{c}{180} & \multicolumn{1}{c}{[0.20, 1.10]} & \multicolumn{1}{c}{LHS} & \multicolumn{1}{c}{} & \multicolumn{1}{c|}{$N_{\texttt{paths}}$} & \multicolumn{1}{c}{$2 \cdot 10^5$} \\
\multicolumn{1}{c|}{$\gamma$} & \multicolumn{1}{c}{180} & \multicolumn{1}{c}{[0.80, 1.10]} & \multicolumn{1}{c}{LHS} & \multicolumn{1}{c}{} & \multicolumn{1}{c|}{$\Dt$} & \multicolumn{1}{c}{$\frac{1}{120}$} \\
\multicolumn{1}{c|}{$\rho$} & \multicolumn{1}{c}{180} & \multicolumn{1}{c}{[-0.92, -0.28]} & \multicolumn{1}{c}{LHS} & \multicolumn{1}{c}{} & \multicolumn{1}{c}{} & \multicolumn{1}{c}{}\\
\multicolumn{1}{c|}{$\Bar{v}$} & \multicolumn{1}{c}{180} & \multicolumn{1}{c}{[0.03, 0.10]} & \multicolumn{1}{c}{LHS} & \multicolumn{1}{c}{} & \multicolumn{2}{c}{SC}\\\cline{6-7}
\multicolumn{1}{c|}{$v_0$} & \multicolumn{1}{c}{180} &
\multicolumn{1}{c}{[0.03, 0.10]}   & \multicolumn{1}{c}{LHS} & \multicolumn{1}{c}{} & \multicolumn{1}{c|}{$Q$} & \multicolumn{1}{c}{$15$} \\
\multicolumn{1}{c|}{$T$} & \multicolumn{1}{c}{160} &
\multicolumn{1}{c}{[0.34, 1.67]} & \multicolumn{1}{c}{EQ-SP} & \multicolumn{1}{c}{} & \multicolumn{1}{c|}{$M$} & \multicolumn{1}{c}{$14$} \\
\multicolumn{1}{c|}{$S^q$} & \multicolumn{1}{c}{15} & \multicolumn{1}{c}{[0.35, 2.01]} & \multicolumn{1}{c}{EQ-SP} & \multicolumn{1}{c}{} & \multicolumn{1}{c|}{$\xib$} & \multicolumn{1}{c}{$2.46$} \\
\multicolumn{1}{c|}{$p^q$} & \multicolumn{1}{c}{15} & \multicolumn{1}{c}{[0.05, 0.85]} & \multicolumn{1}{c}{IMPL} & \multicolumn{1}{c}{} & \multicolumn{1}{c|}{$F_{\xi}(\xib)$} & \multicolumn{1}{c}{$0.993$}\\\cline{1-4}
\multicolumn{1}{c|}{$N_{\texttt{pairs}}$}&\multicolumn{3}{c}{$432000 \quad(180 \times 160 \times 15)$} &\multicolumn{1}{c}{}
\end{tabular}}
\end{center}
\vspace{-0.3cm}
\end{table}

Similarly as before, the sets of Heston parameters are sampled using LHS. For each set, $N_{\texttt{tot}}$ paths are generated, with a time step of $\Dt$ and a time horizon up to $T_{\max}$.
The underlying process $S$ is monitored at each time $T$ for which there are enough past observations to compute $A(S)$, i.e.:
\begin{equation*}
        T\geq 4\tau_{\texttt{A}} + \Dt.
\end{equation*}
For any maturity $T$ and any realization $S^q$, the inverse CDF of the conditional random variable $A(S)|S(T)=S^q$ is approximated with the empirical quantile function, $Q_{A|S^q}$. The quantile function $Q_{A|S^q}$ is built on the $N_{\texttt{paths}}$ ``closest'' paths to $S^q$, i.e. those $N_{\texttt{paths}}$ paths whose final values $S(T)$ are the closest to $S^q$.

Eventually, for any input set $\p=\{r,\kappa,\gamma,\rho,\bar{v},v_0,T, S^q,p^q\}$, the CVs $\a$ corresponding to $A(S)|S(T)=S^q$ are computed as:
\begin{equation*}
    a_k:=F_{A|S^q}^{-1}\big(F_{\xi}(\xi_k)\big)\approx Q_{A|S^q}\big(F_{\xi}(\xi_k)\big),\qquad k\in\{1,\dots,14\},
\end{equation*}
with $\xi_k$ the Chebyshev nodes:
\begin{equation*}
    \xi_k:=-\xib\cos{\frac{k - 1}{13}\pi},\qquad k\in\{1,\dots,14\},
\end{equation*}
and $\xib:=F_\xi^{-1}(0.993)\approx 2.46$.

The information about the generation of the training set $\T_{\texttt{FxFlA}}$ are reported in \Cref{tab: TrainingSetRangesConditional}.

\subsubsection{Artificial neural network training}

\begin{figure}[b]
    \centering
    \makebox[\textwidth][c]{\includegraphics[width=.8\textwidth,height=.4\textwidth]{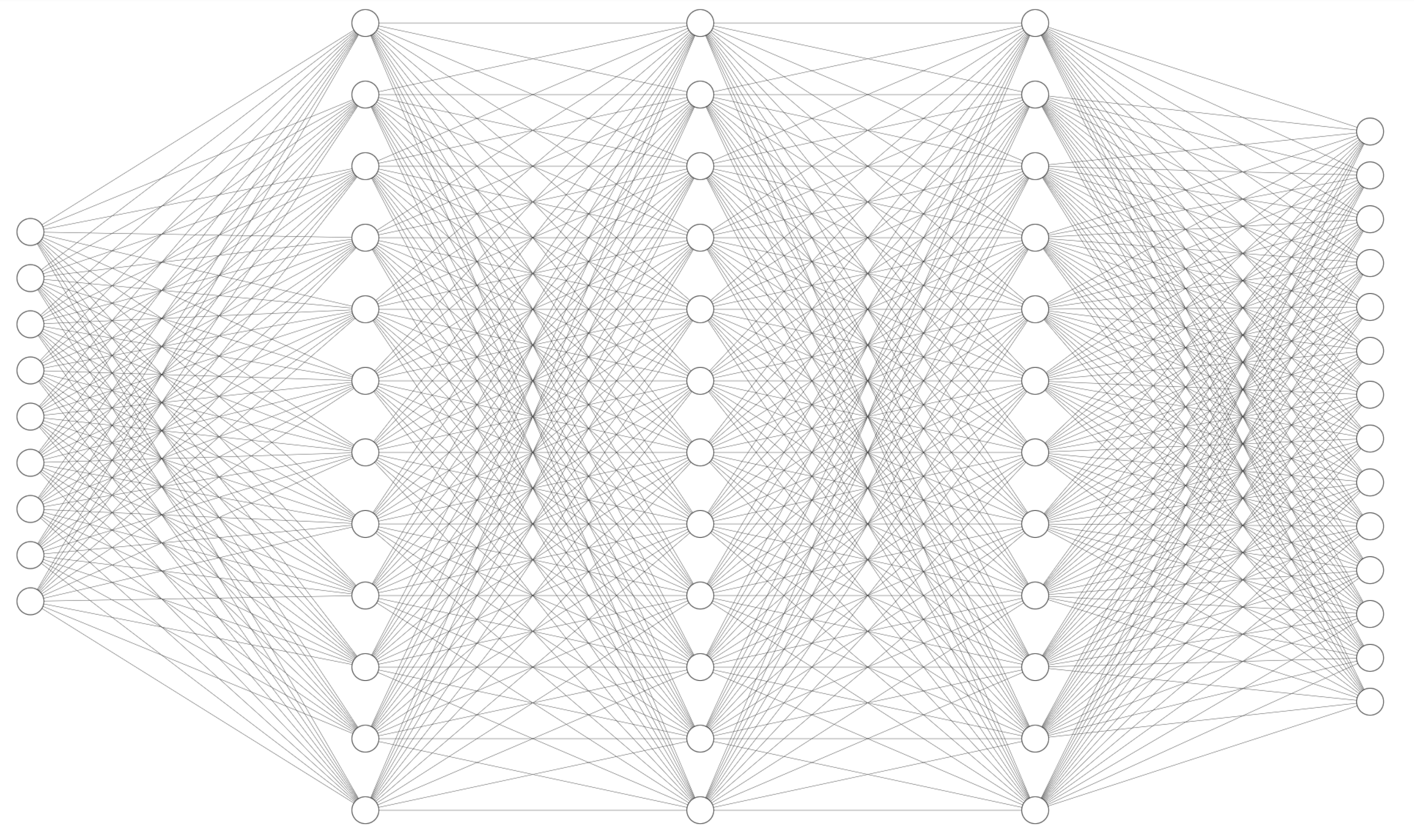}}
    \caption{\footnotesize Illustration of a dense ANN with (from the left) one input layer, three hidden layers, and one output layer. Each white node represents an activation unit.}
    \label{fig: ANN architecture}
\end{figure}

Each training set $\T_{(\cdot)}$ store a finite amount of pairs $(\p,\a)$, in which each $\p$ and each corresponding $\a$ are connected by the mapping $H$. The artificial neural network $\H$ is used to approximate and generalize $H$ for inputs $\p$ not in $\T_{(\cdot)}$. The architecture of $\H$ was initially chosen accordingly to \citep*{SevenLeague,FastSamplingBridge}, then suitably adjusted by heuristic arguments. 

$\H$ is a \emph{fully connected (or dense)} ANN with five \emph{layers} -- one \emph{input}, one \emph{output}, and three \emph{hidden} (HidL), as the one illustrated in \Cref{fig: ANN architecture}. Input and output layers have a number of \emph{units (neurons)} -- input size (InS) and output size (OutS) -- coherent with the targeted problem (FxA, FxL, or FxFlA). Each hidden layer has the same hidden size (HidS) of 200 neurons, selected as the optimal one among different settings. \texttt{ReLU}
(\texttt{Re}ctified \texttt{L}inear \texttt{U}nit) is the \emph{non-linear activation unit} (ActUn) for each neuron, and it is defined as $\phi(x):=\max(x,0)$ \citep*{ActivationFunctions}. The loss function (LossFunc) is the \emph{Mean Squared Error} (MSE) between the actual outputs, $\a$ (available in $\T$), and the ones predicted by the ANN, $\H(\p)$. 

\begin{table}[t]
\begin{center}
\caption{\footnotesize Artificial neural network and optimization details.}
\label{tab: ANNandOptimizationDetails}
\vspace{-0.2cm}
\resizebox{0.95\textwidth}{!}{
\begin{tabular}{lllllllllllll}
\hline\hline
\multicolumn{6}{c}{\multirow{1}{*}{ANN architecture}} & \multicolumn{1}{c}{} & \multicolumn{6}{c}{\multirow{1}{*}{ANN optimization}}\\\cline{1-6}\cline{8-13}
\multicolumn{1}{c|}{} & \multicolumn{1}{c}{FxA} & \multicolumn{1}{c}{FxL} & \multicolumn{1}{c|}{FxFlA} & \multicolumn{1}{c}{HidL} & \multicolumn{1}{c}{3} & \multicolumn{1}{c}{} & \multicolumn{1}{c}{E} & \multicolumn{1}{c|}{3000} & \multicolumn{1}{c}{InitLR} & \multicolumn{1}{c|}{$10^{-3}$} & \multicolumn{1}{c}{Opt} & \multicolumn{1}{c}{\texttt{Adam}} \\
\multicolumn{1}{c|}{InS} & \multicolumn{1}{c}{7} & \multicolumn{1}{c}{6} & \multicolumn{1}{c|}{9} & \multicolumn{1}{c}{HidS} & \multicolumn{1}{c}{200} & \multicolumn{1}{c}{} & \multicolumn{1}{c}{B} & \multicolumn{1}{c|}{$1024$} & \multicolumn{1}{c}{DecR} & \multicolumn{1}{c|}{0.1} & \multicolumn{1}{c}{LossFunc} & \multicolumn{1}{c}{MSE} \\
\multicolumn{1}{c|}{OutS} & \multicolumn{1}{c}{21} & \multicolumn{1}{c}{21} & \multicolumn{1}{c|}{14} & \multicolumn{1}{c}{ActUn} & \multicolumn{1}{c}{\texttt{ReLU}} & \multicolumn{1}{c}{} & \multicolumn{1}{c}{} & \multicolumn{1}{c|}{} & \multicolumn{1}{c}{DecS} & \multicolumn{1}{c|}{1000} & \multicolumn{1}{c}{} & \multicolumn{1}{c}{} \\
\end{tabular}}
\end{center}
\vspace{-0.3cm}
\end{table}

\begin{figure}[b]%
\vspace{-.3cm}
    \centering
    \subfloat[\centering]{{\includegraphics[width=6.8cm]{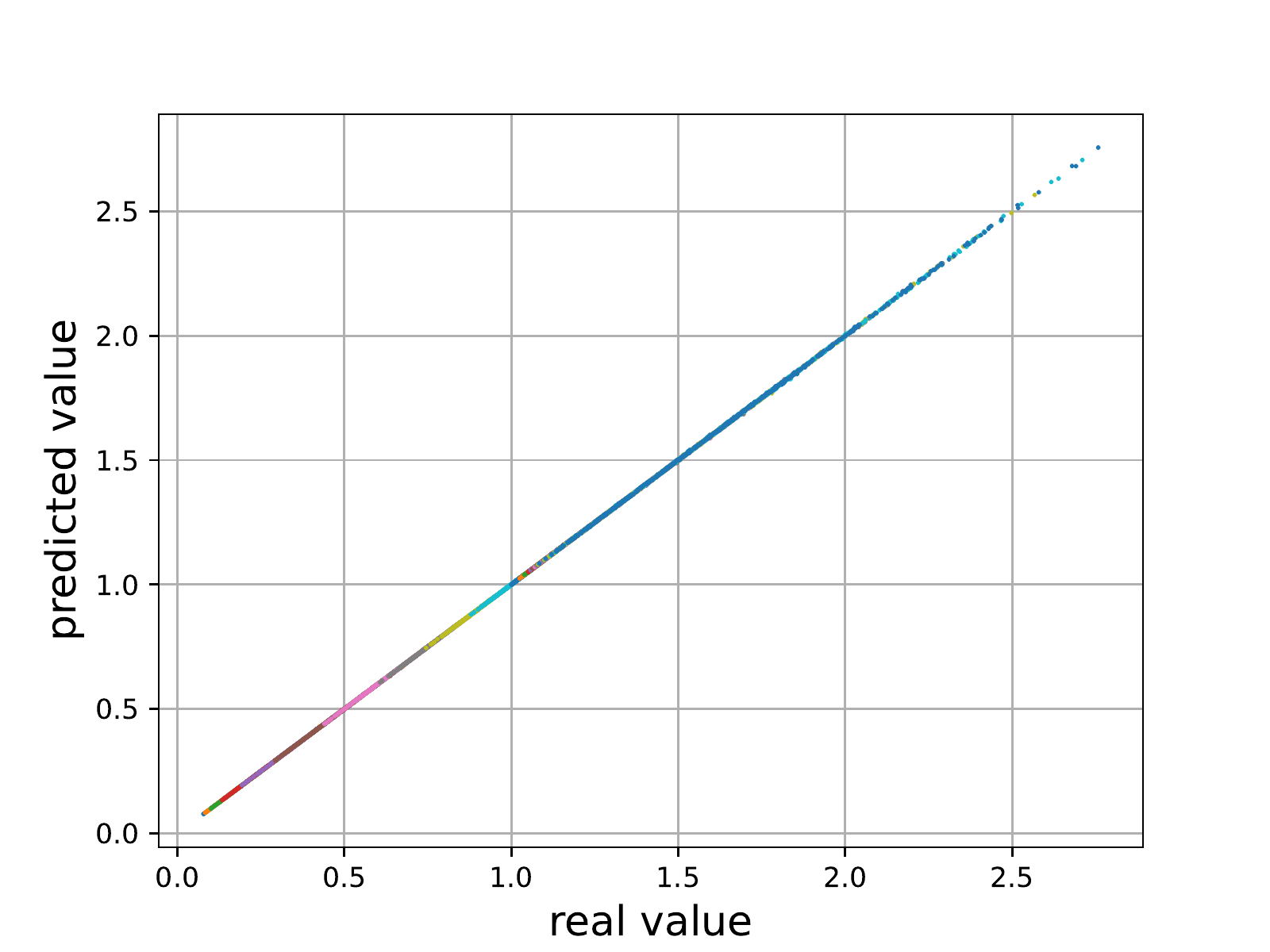} }}%
    ~
    \subfloat[\centering]{{\includegraphics[width=6.8cm]{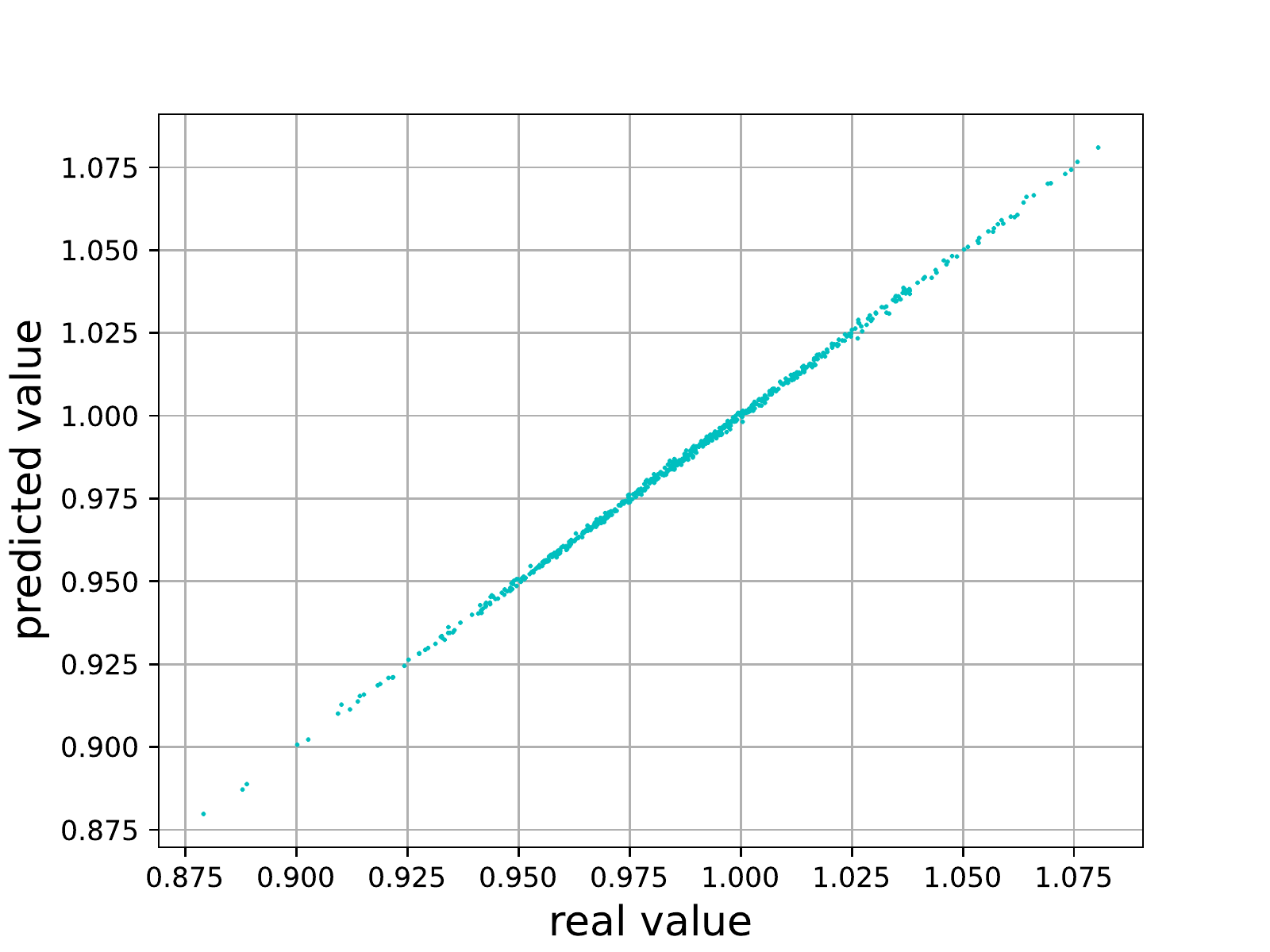} }}%
    \vspace{-0.1cm}
    \caption{\footnotesize First experiment: FxA. Left: scatter plot of the real CVs ($a_k$, $k=1,\dots,21$, with different colors) against the predicted ones. Right: zoom on the ``worst'' case, namely $a_{10}$.}%
    \label{fig: AccuracyANN}%
    \vspace{-0.3cm}
\end{figure}

The optimization process is composed of 3000 epochs (E). During each epoch, the major fraction (70\%) of the $\T$ (the actual \emph{training set}) is ``back-propagated'' through the ANN in batches of size 1024 (B). The stochastic gradient-based optimizer (Opt) \texttt{Adam} \citep*{Adam} is employed in the optimization. Particularly, the optimizer updates the ANN weights based on the gradient computed on each random batch (during each epoch). The initial learning rate (InitLR) is $10^{-3}$, with a decaying rate (DecR) of 0.1 and a decaying step (DecS) of 1000 epochs. The details are reported in \Cref{tab: ANNandOptimizationDetails}.

Furthermore, during the optimization routine, the 20\% of $\T$ is used to \emph{validate} the result (namely, to avoid the overfitting of the training set). Eventually, the remaining 10\% of $\T$ is used for \emph{testing} the quality of the ANN. \Cref{fig: AccuracyANN} provides a visual insight into the high accuracy the ANN reaches at the end of the training process. \Cref{fig: AccuracyANN}a shows the scatter plot of the real CVs $a_k$, $k=1,\dots,21$, against the ones predicted using the ANN, for the experiment FxA; in \Cref{fig: AccuracyANN}b a zoom is applied to the ``worst'' case, namely the CV $a_{10}$, for which anyway is reached the extremely high $R^2$ score of 0.9994.

\subsection{Sampling and pricing}

Given the trained model from the previous section, we can now focus on the actual sampling and/or pricing of options. In particular, for the first two experiments, we consider the following payoffs:
\begin{align}
\label{eqn: FxAPayoff}
    \text{FxA:}&\quad \max\Big(\omega\big( A(S)-K_2\big), 0\Big), \quad \frac{K_2}{S(t_0)}\in[0.8,1.2],\\
    \label{eqn: FxLPayoff}
    \text{FxL:}&\quad \max\Big(\omega\big( A(S)-K_2\big), 0\Big), \quad \frac{K_2}{S(t_0)}\in[0.8,1.2],
\end{align}
whereas for the third, FxFlA, we have:
\begin{equation}
\label{eqn: FxFlAPayoff}
    \max\Big(\omega\big( A(S)- K_1S(T)-K_2\big), 0\Big), \quad\bigg(\frac{K_1}{S(t_0)},\frac{K_2}{S(t_0)}\bigg)\in[0.4,0.6]\times[0.4,0.6],
\end{equation}
with $A(S)$ defined as in (\ref{eqn: SpecsAsian}) for FxA and FxFlA, and as in (\ref{eqn: SpecsLook}) for FxL.

All the results in the following sections are compared to a MC benchmark obtained using the almost exact simulation described in \Cref{ap: ExactSimulationHeston}.

\subsubsection{Numerical results for FxA}
\label{sssec: FxA}

The procedure described in \Cref{ssec: SpecialPricing} is employed to solve the problem of pricing fixed-strike discrete Asian options with payoffs as in (\ref{eqn: FxAPayoff}), with underlying stock price initial value $S(t_0)=1$. In this experiment, the ANN is trained on Heston model parameters' ranges, which include the examples proposed in \cite{SimulationHestonAndersen} representing some real applications. Furthermore, we note the following aspect.
\begin{rem}[Scaled underlying process and (positive) homogeneity of $A$]
\label{rem: ScaledProcessAndHomogeneity}
The unit initial value is not restrictive. Indeed, the underlying stock price dynamics in \Cref{eqn: StockDynamics,eqn: VolatilityDynamics} are independent of $S(t_0)$, with the initial value only accounting as a multiplicative constant (this can be easily proved by means of It\^o's lemma). Moreover, since $A(S)$ is (positive) homogeneous in $S$ also $A(S)$ can be easily ``scaled'' according to the desired initial value. Particularly, given the constant $c>0$, $c A(S)\overset{\d}{=} A(c S)$.
\end{rem}

\begin{table}[t]
\begin{center}
\caption{\footnotesize Tested Heston parameter sets.}
\label{tab: HestonSets}
\vspace{-0.2cm}
\resizebox{0.48\textwidth}{!}{\begin{tabular}{llllllll}
\hline\hline
\multicolumn{1}{c|}{Set} &\multicolumn{1}{c}{$r$} & \multicolumn{1}{c}{$\kappa$} & \multicolumn{1}{c}{$\gamma$} & \multicolumn{1}{c}{$\rho$} & \multicolumn{1}{c}{$\Bar{v}$} & \multicolumn{1}{c}{$v_0$} & \multicolumn{1}{c}{$T$}\\
\multicolumn{1}{c|}{I} &
\multicolumn{1}{c}{\multirow{1}{*}{0.04}} & \multicolumn{1}{c}{\multirow{1}{*}{0.5}} & \multicolumn{1}{c}{\multirow{1}{*}{1.0}} & \multicolumn{1}{c}{\multirow{1}{*}{-0.8}} & \multicolumn{1}{c}{\multirow{1}{*}{0.08}} & \multicolumn{1}{c}{\multirow{1}{*}{0.05}} & \multicolumn{1}{c}{\multirow{1}{*}{1.0}}\\\hline
\multicolumn{1}{c|}{Set} &\multicolumn{1}{c}{$r$} & \multicolumn{1}{c}{$\kappa$} & \multicolumn{1}{c}{$\gamma$} & \multicolumn{1}{c}{$\rho$} & \multicolumn{1}{c}{$\Bar{v}$} & \multicolumn{1}{c}{$v_0$} & \multicolumn{1}{c}{$T$}\\
\multicolumn{1}{c|}{II} &
\multicolumn{1}{c}{\multirow{1}{*}{0.02}} & \multicolumn{1}{c}{\multirow{1}{*}{1.0}} & \multicolumn{1}{c}{\multirow{1}{*}{0.9}} & \multicolumn{1}{c}{\multirow{1}{*}{-0.6}} & \multicolumn{1}{c}{\multirow{1}{*}{0.10}} & \multicolumn{1}{c}{\multirow{1}{*}{0.13}} & \multicolumn{1}{c}{\multirow{1}{*}{1.5}}\\\hline
\multicolumn{1}{c|}{Set} &\multicolumn{1}{c}{$r$} & \multicolumn{1}{c}{$\kappa$} & \multicolumn{1}{c}{$\gamma$} & \multicolumn{1}{c}{$\rho$} & \multicolumn{1}{c}{$\Bar{v}$} & \multicolumn{1}{c}{$v_0$} & \multicolumn{1}{c}{$T$}\\
\multicolumn{1}{c|}{III} &\multicolumn{1}{c}{\multirow{1}{*}{0.01}} & \multicolumn{1}{c}{\multirow{1}{*}{0.46}} & \multicolumn{1}{c}{\multirow{1}{*}{0.99}} & \multicolumn{1}{c}{\multirow{1}{*}{-0.79}} & \multicolumn{1}{c}{\multirow{1}{*}{0.09}} & \multicolumn{1}{c}{\multirow{1}{*}{0.11}} & \multicolumn{1}{c}{\multirow{1}{*}{1.0}}\\
\end{tabular}}
\end{center}
\end{table}

\begin{figure}[b]
\vspace{-0.5cm}
    \centering
    \subfloat[\centering]{{\includegraphics[height=4.3cm,width=4.68cm]{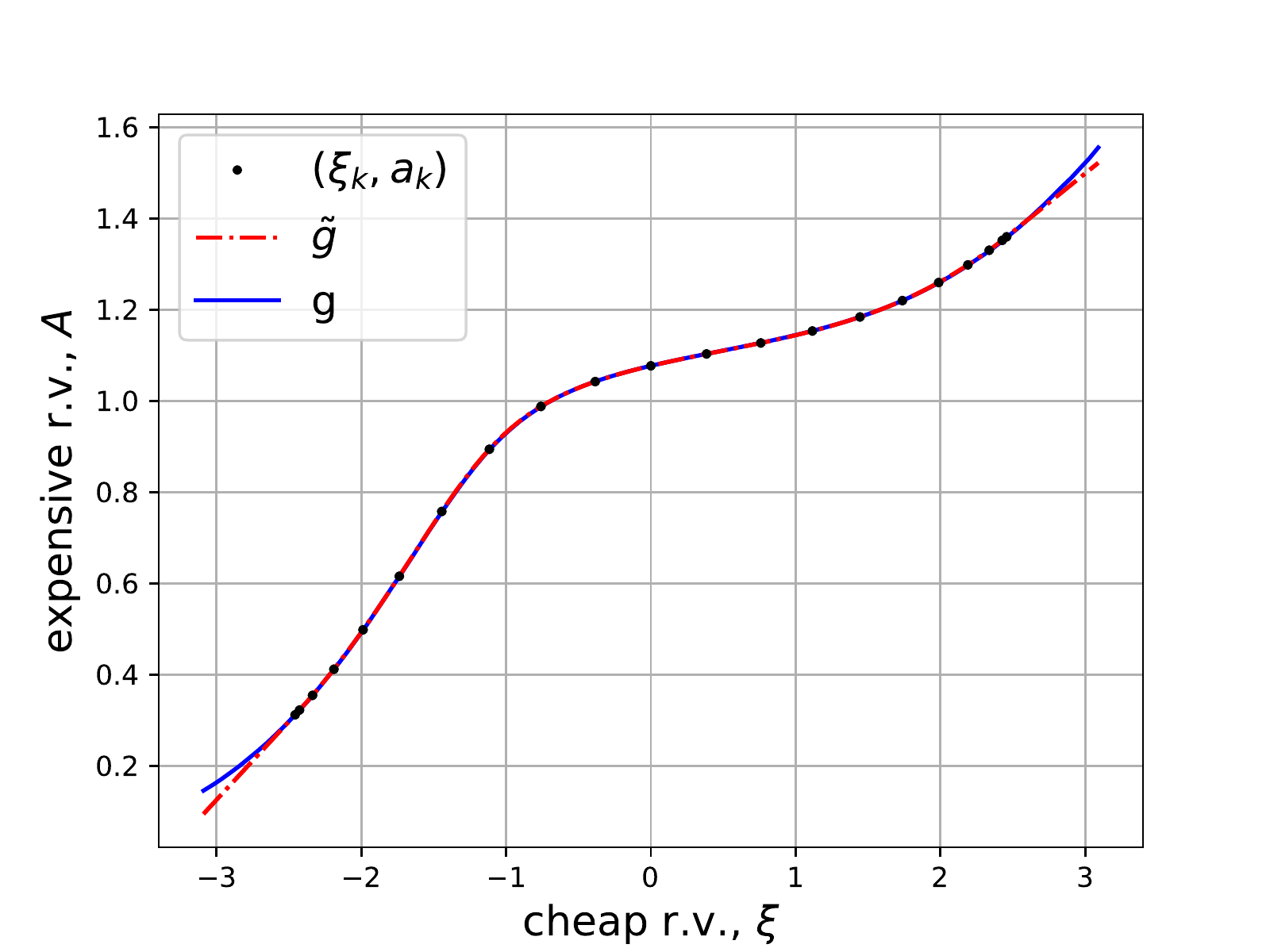} }}%
    \subfloat[\centering]{{\includegraphics[height=4.3cm,width=4.68cm]{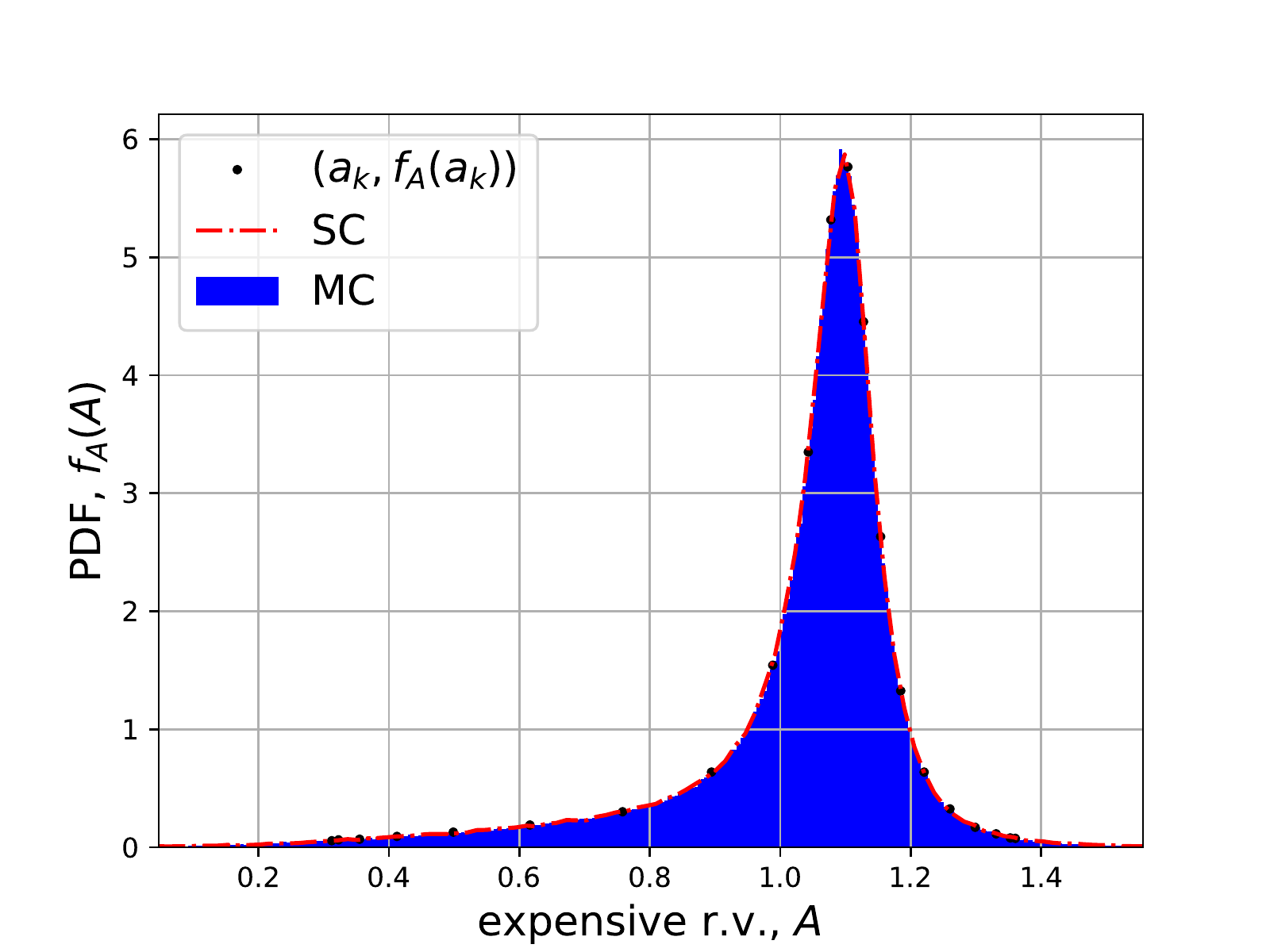} }}%
    \subfloat[\centering]{{\includegraphics[height=4.3cm,width=4.68cm]{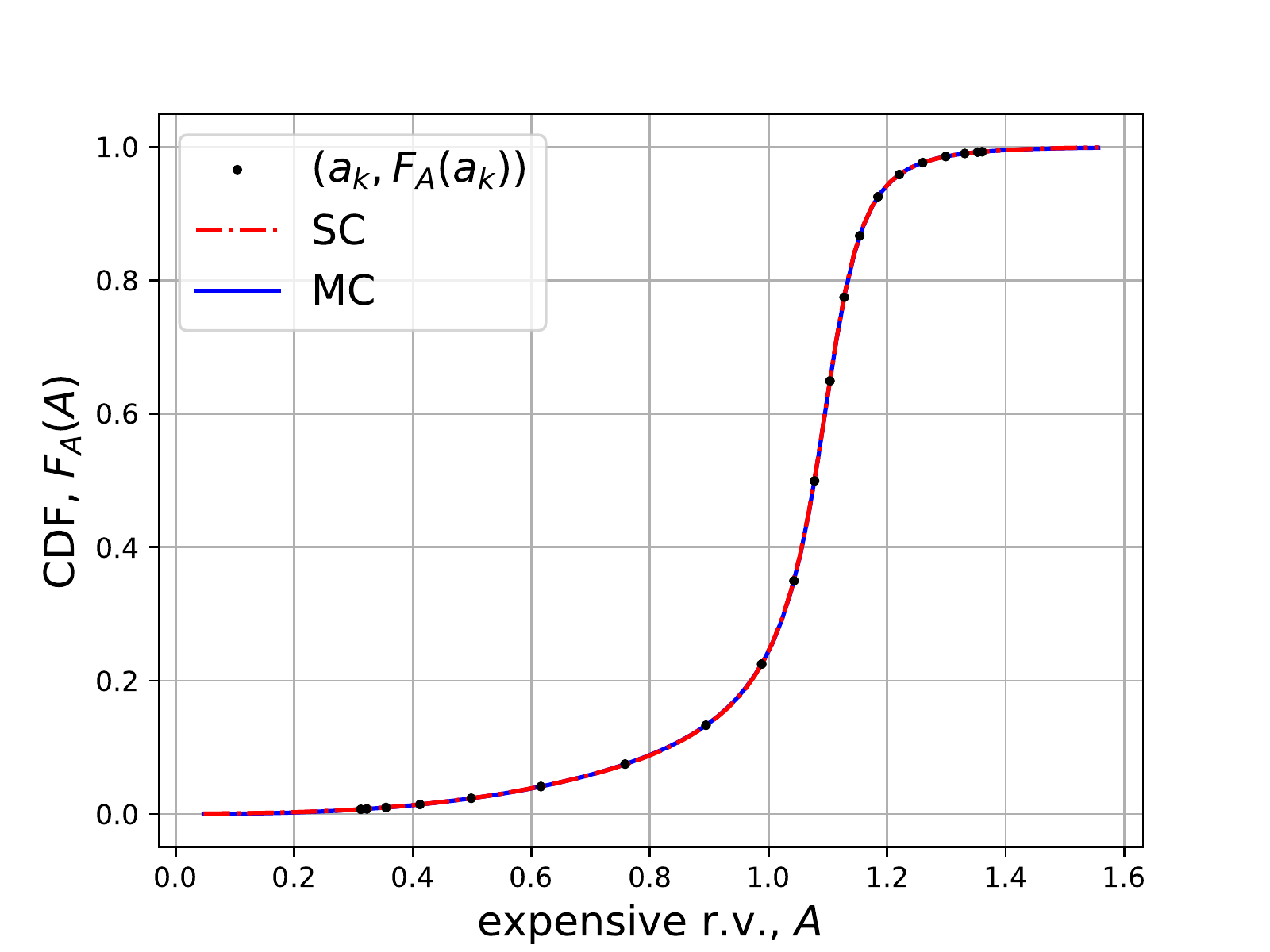} }}%
    \vspace{-0.1cm}
    \caption{\footnotesize Left: comparison between maps $g$ from MC and $\g$ from SC (with linear extrapolation). Center: comparison between MC histogram of $A(S)$ and the numerical PDF from SC. Right: comparison between MC numerical CDF of $A(S)$ and the numerical CDF from SC.}%
    \label{fig: AsianDistribution}%
    \vspace{-0.3cm}
\end{figure}

The methodology is tested on different randomly chosen sets of Heston parameters. We report the details for two specific sets, Set I and Set II, available in \Cref{tab: HestonSets}. For Set I, in \Cref{fig: AsianDistribution}, we compare the population from $A(S)$ obtained employing SC with the MC benchmark (both with $N_{\texttt{paths}}=10^5$ paths each). \Cref{fig: AsianDistribution}a shows the highly accurate approximation of the exact map $g=F_A^{-1}\circ F_{\xi}$ by means of the piecewise polynomial approximation $\g$. As a consequence, both the PDF (see \Cref{fig: AsianDistribution}b) and the CDF (see \Cref{fig: AsianDistribution}c) perfectly match. Moreover, the methodology is employed to value fixed-strike arithmetic Asian options (calls and puts) for two sets of parameters (Set I and Set II) and 50 different strikes $K_2$. The resulting prices are reported in Figures \ref{fig: AsianCallPut}a and \ref{fig: AsianCallPut}c. \textcolor{black}{Figures \ref{fig: AsianCallPut}b and \ref{fig: AsianCallPut}d display the standard errors for MC and SC on the left y-axis, and the pricing error $\epsilon_P$ for SC and SA (in units of the corresponding MC standard error) on the right y-axis. The pricing error tends to be more significant the more out of money the option is, due to the smaller SE.}

\textcolor{black}{The timing results are reported (in milliseconds) in \Cref{tab: TimingResults}, for different choices of $N_{\texttt{paths}}$. Both SC and SA times only refer to the online pricing computational time, namely the time required for the pricing procedure, excluding the training sets generation and the ANNs training, which are performed -- only once -- offline.}
The semi-analytic formula requires a constant evaluation time, as well as the SC technique (if $N_{\texttt{paths}}$ is fixed), whereas the MC simulation is dependent on the parameter $T$ (since we decided to keep the same MC step in every simulation). Therefore, the methodology becomes more convenient the longer the maturity of the option $T$. The option pricing computational time is reduced by tens of times when using SC to generate the population from $A$, while it is reduced by hundreds of times if the semi-analytic (SA) formula is employed.

\begin{figure}[t]
    \centering
    \subfloat[\centering]{{\includegraphics[width=6.8cm]{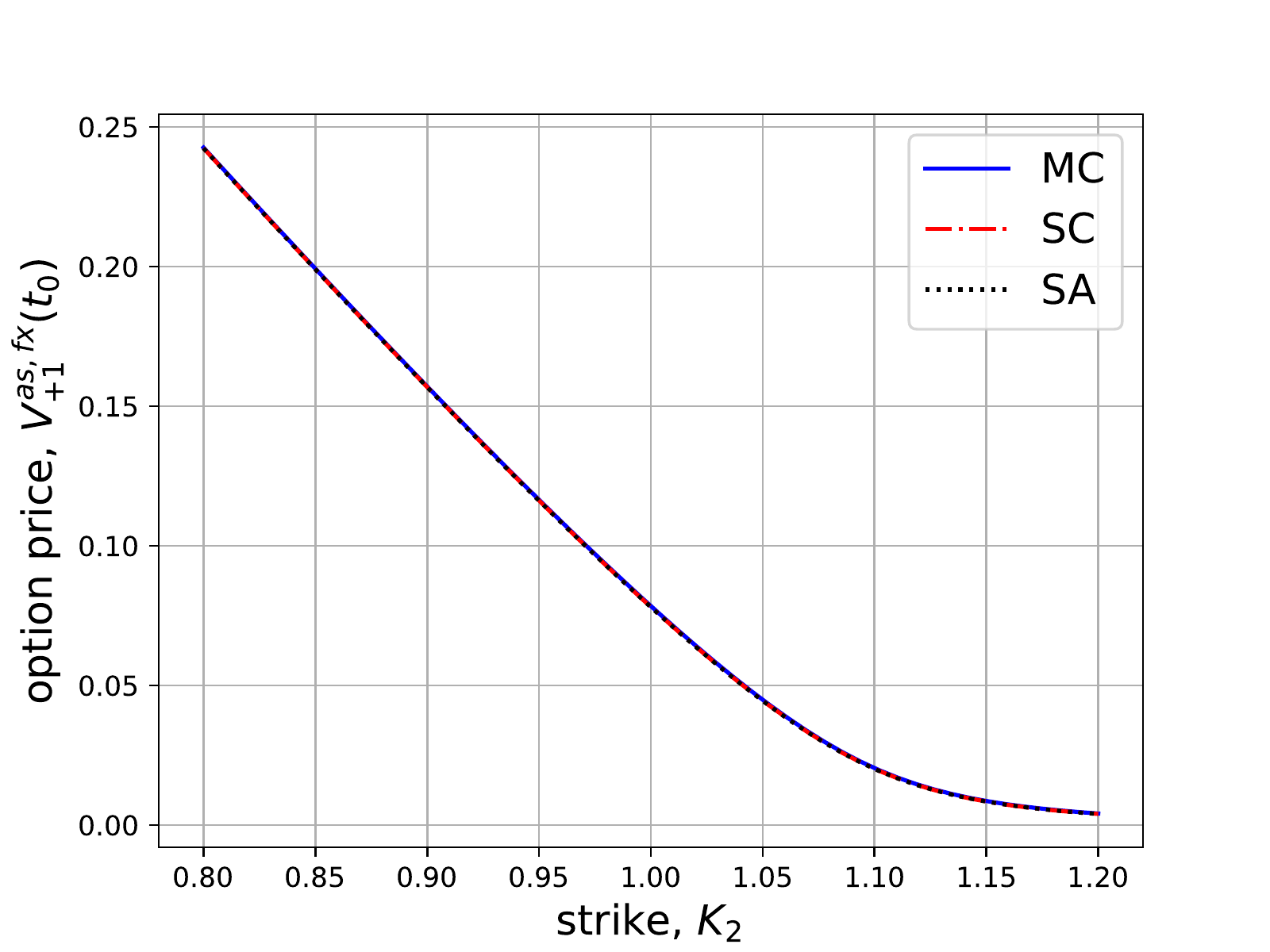} }}%
    ~
    \subfloat[\centering]{{\includegraphics[width=6.8cm]{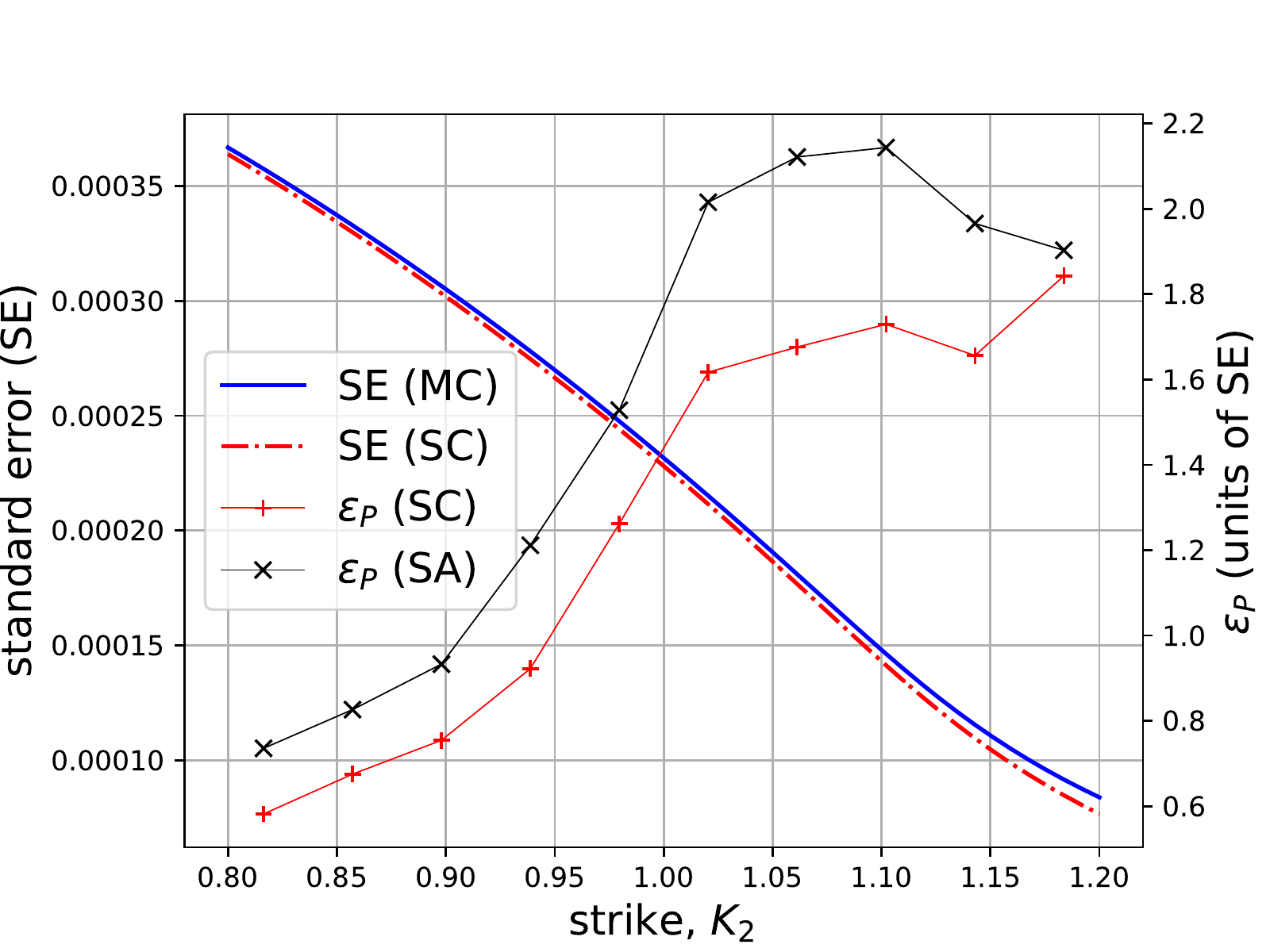} }}\\
    \subfloat[\centering]{{\includegraphics[width=6.8cm]{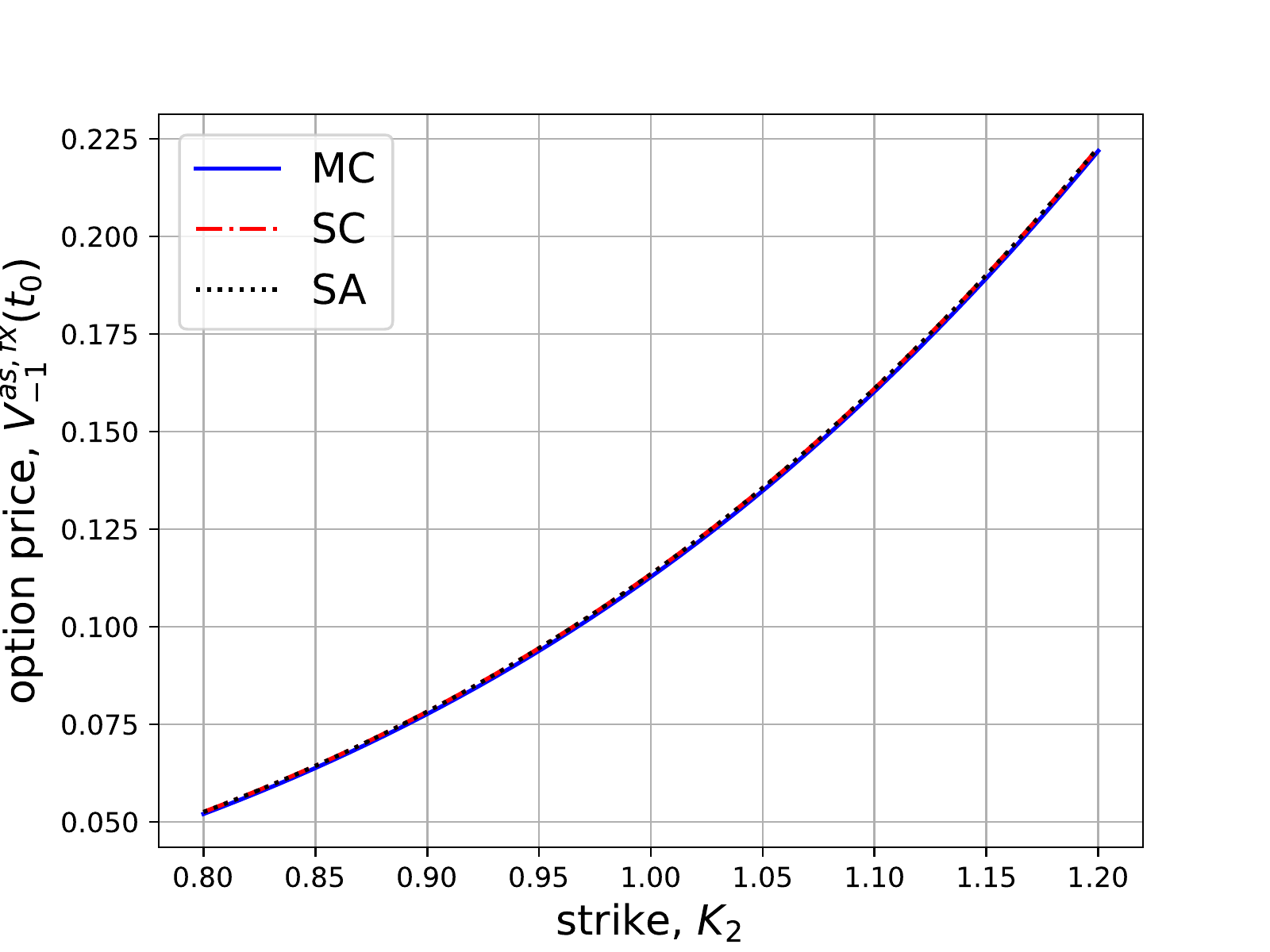} }}%
    ~
    \subfloat[\centering]{{\includegraphics[width=6.8cm]{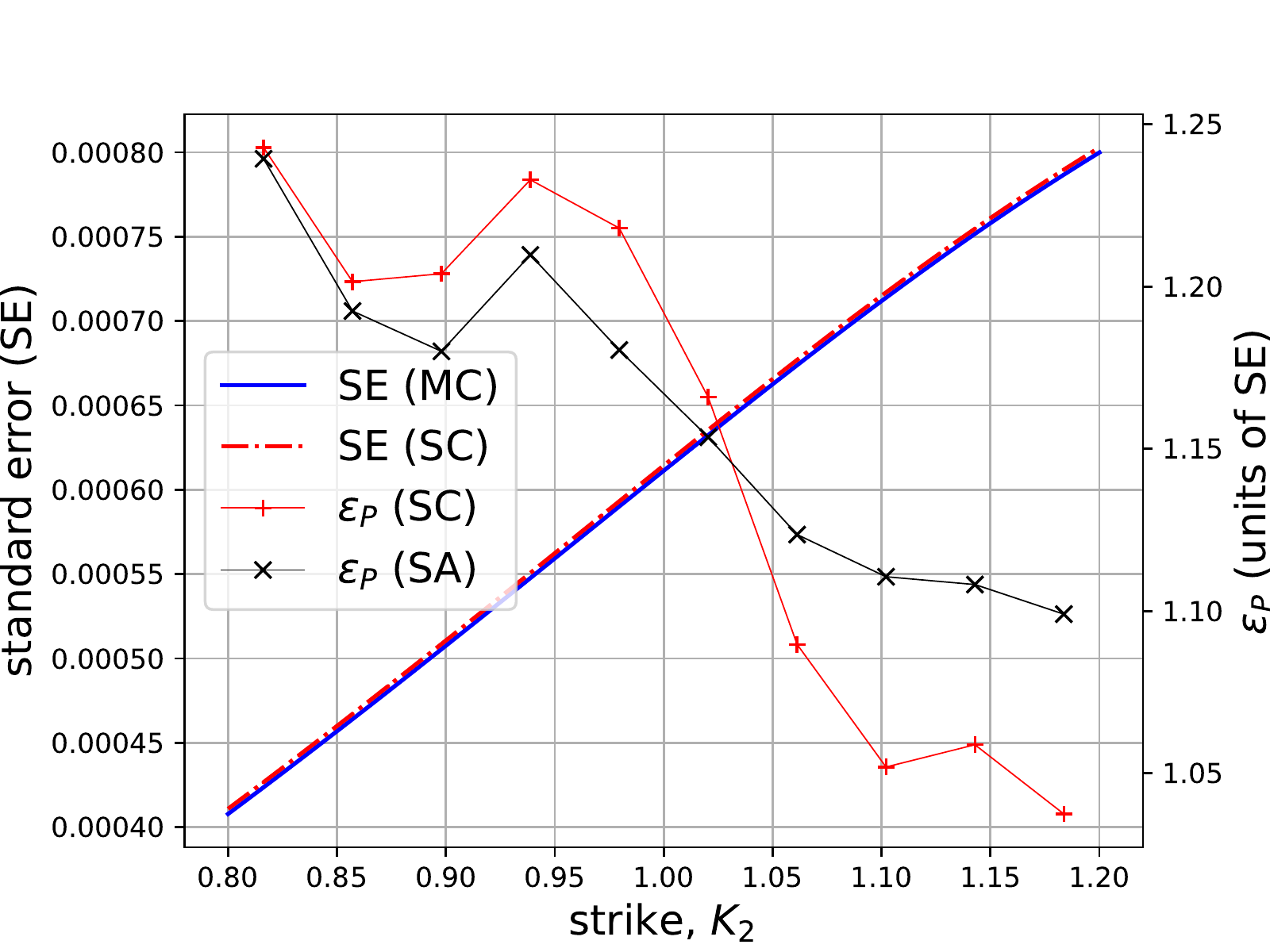} }}%
    \caption{\footnotesize Left: fixed-strike discrete Asian option prices for 50 different strikes under the Heston model dynamics. Right: MC and SC standard errors (left y-axis) and pricing errors in units of the corresponding MC standard errors SE, obtained with $N_{\texttt{paths}}=10^5$ for both MC and SC (right y-axis). Up: call with Set I. Down: put with Set II.}%
    \label{fig: AsianCallPut}%
\end{figure}

\begin{table}[b]
\begin{center}
\caption{\footnotesize Timing results for option pricing.}
\label{tab: TimingResults}
\vspace{-0.2cm}
\resizebox{0.95\textwidth}{!}{\begin{tabular}{llllllllllll}
\hline\hline
\multicolumn{1}{c}{} & \multicolumn{3}{c}{Fx Asian call (Set I)} & \multicolumn{1}{c}{} & \multicolumn{3}{c}{Fx Asian put (Set II)}& \multicolumn{1}{c}{} & \multicolumn{3}{c}{Fx-Fl Asian call (Set III)}\\\cline{1-12}
\multicolumn{1}{c}{} &\multicolumn{1}{c}{$N_{\texttt{paths}}$} & \multicolumn{1}{c}{time} & \multicolumn{1}{c}{speed-up} & \multicolumn{1}{c}{} & \multicolumn{1}{c}{$N_{\texttt{paths}}$} & \multicolumn{1}{c}{time} & \multicolumn{1}{c}{speed-up} & \multicolumn{1}{c}{} & \multicolumn{1}{c}{$N_{\texttt{paths}}$} & \multicolumn{1}{c}{time} & \multicolumn{1}{c}{speed-up}\\\cline{2-4}\cline{6-8}\cline{10-12}
\multicolumn{1}{c|}{MC} & \multicolumn{1}{c}{$10^5$} & \multicolumn{1}{c}{1951 ms} & \multicolumn{1}{c|}{-} & \multicolumn{1}{c|}{MC} & \multicolumn{1}{c}{$10^5$} & \multicolumn{1}{c}{3134 ms} & \multicolumn{1}{c|}{-} & \multicolumn{1}{c|}{MC}  & \multicolumn{1}{c}{$10^5$} & \multicolumn{1}{c}{1996 ms} & \multicolumn{1}{c}{-}\\
\multicolumn{1}{c|}{SC} & \multicolumn{1}{c}{$10^5$} & \multicolumn{1}{c}{24 ms} & \multicolumn{1}{c|}{81} & \multicolumn{1}{c|}{SC} & \multicolumn{1}{c}{$10^5$} & \multicolumn{1}{c}{29 ms} & \multicolumn{1}{c|}{107} & \multicolumn{1}{c|}{SC}  & \multicolumn{1}{c}{$10^5$} & \multicolumn{1}{c}{69 ms} & \multicolumn{1}{c}{29}\\
\multicolumn{1}{c|}{SA} & \multicolumn{1}{c}{-} & \multicolumn{1}{c}{15 ms} & \multicolumn{1}{c|}{130} & \multicolumn{1}{c|}{SA} & \multicolumn{1}{c}{-} & \multicolumn{1}{c}{16 ms} & \multicolumn{1}{c|}{196} & \multicolumn{1}{c|}{-} & \multicolumn{1}{c}{-} & \multicolumn{1}{c}{-} & \multicolumn{1}{c}{-}\\\\\cline{2-4}\cline{6-8}\cline{10-12}
\multicolumn{1}{c|}{MC} & \multicolumn{1}{c}{$2\times 10^5$} & \multicolumn{1}{c}{3905 ms} & \multicolumn{1}{c|}{-} & \multicolumn{1}{c|}{MC} & \multicolumn{1}{c}{$2\times 10^5$} & \multicolumn{1}{c}{6386 ms} & \multicolumn{1}{c|}{-} & \multicolumn{1}{c|}{MC}  & \multicolumn{1}{c}{$5\times 10^4$} & \multicolumn{1}{c}{953 ms} & \multicolumn{1}{c}{-}\\
\multicolumn{1}{c|}{SC} & \multicolumn{1}{c}{$2\times 10^5$} & \multicolumn{1}{c}{58 ms} & \multicolumn{1}{c|}{68} & \multicolumn{1}{c|}{SC} & \multicolumn{1}{c}{$2\times 10^5$} & \multicolumn{1}{c}{55 ms} & \multicolumn{1}{c|}{116} & \multicolumn{1}{c|}{SC}  & \multicolumn{1}{c}{$5\times 10^4$} & \multicolumn{1}{c}{43 ms} & \multicolumn{1}{c}{22}\\
\multicolumn{1}{c|}{SA} & \multicolumn{1}{c}{-} & \multicolumn{1}{c}{15 ms} & \multicolumn{1}{c|}{263} & \multicolumn{1}{c|}{SA} & \multicolumn{1}{c}{-} & \multicolumn{1}{c}{15 ms} & \multicolumn{1}{c|}{440} & \multicolumn{1}{c|}{-} & \multicolumn{1}{c}{-} & \multicolumn{1}{c}{-} & \multicolumn{1}{c}{-}\\
\end{tabular}}
\end{center}
\vspace{-0.5cm}
\end{table}

Eventually, the error distribution of 10000 different option prices (one call and one put with 100 values for $K_2$ each for 50 randomly chosen Heston parameters' sets and maturities) is given in \Cref{fig: ErrorDistribution}a. The SA prices (assuming a linear extrapolation) are compared with MC benchmarks. The outcome is satisfactory and shows the robustness of the methodology proposed. \textcolor{black}{The error is smaller than three times the MC standard error in more than 90\% of the cases when $N_{\texttt{paths}}=10^5$ (red histogram), and in more than 80\% of the cases when $N_{\texttt{paths}}=2\times 10^5$ (blue histogram).}

\begin{figure}[t]
\vspace{-0.5cm}
    \centering
    \subfloat[\centering]{{\includegraphics[height=4.3cm,width=4.68cm]{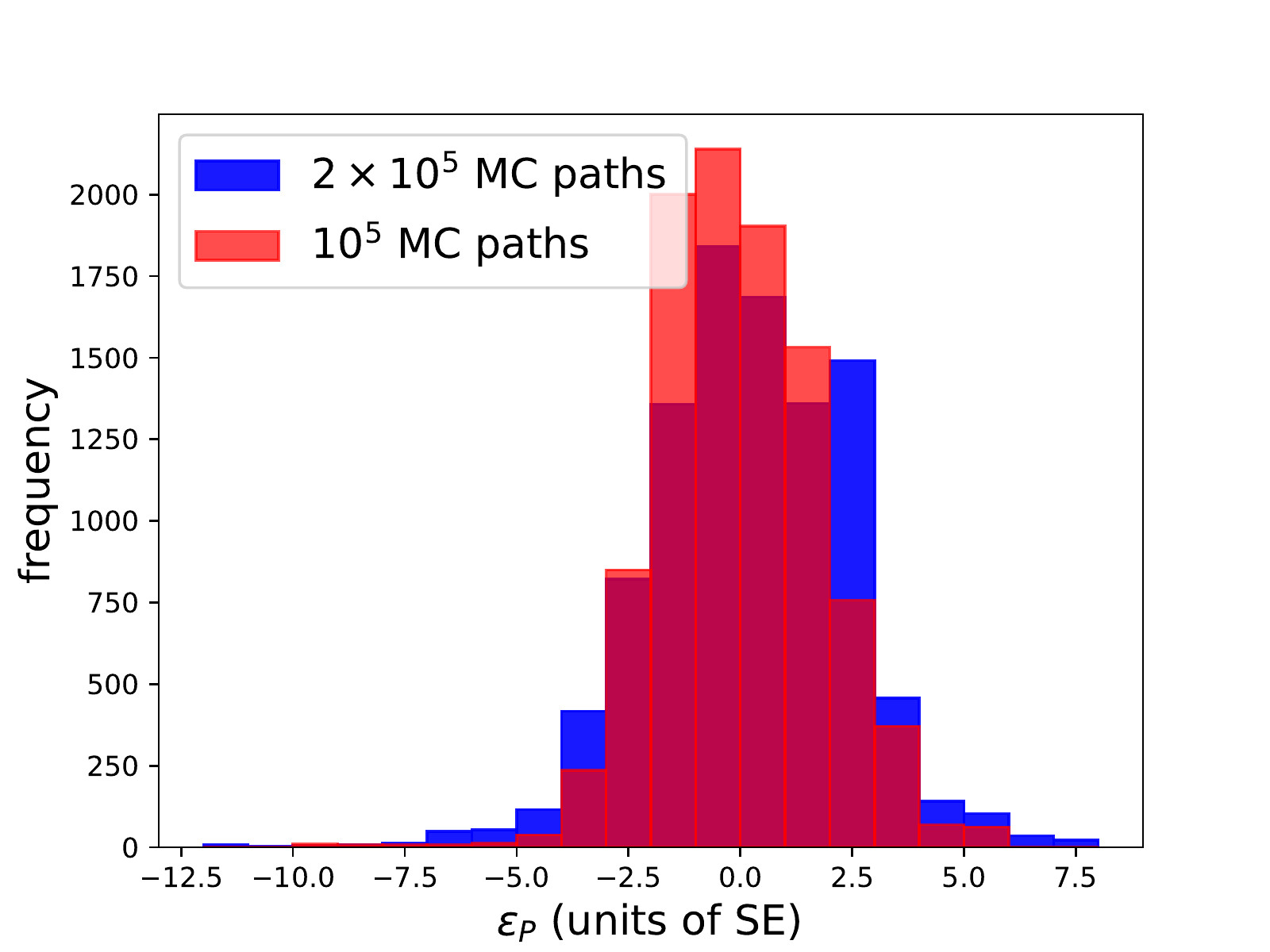} }}%
    \subfloat[\centering]{{\includegraphics[height=4.3cm,width=4.68cm]{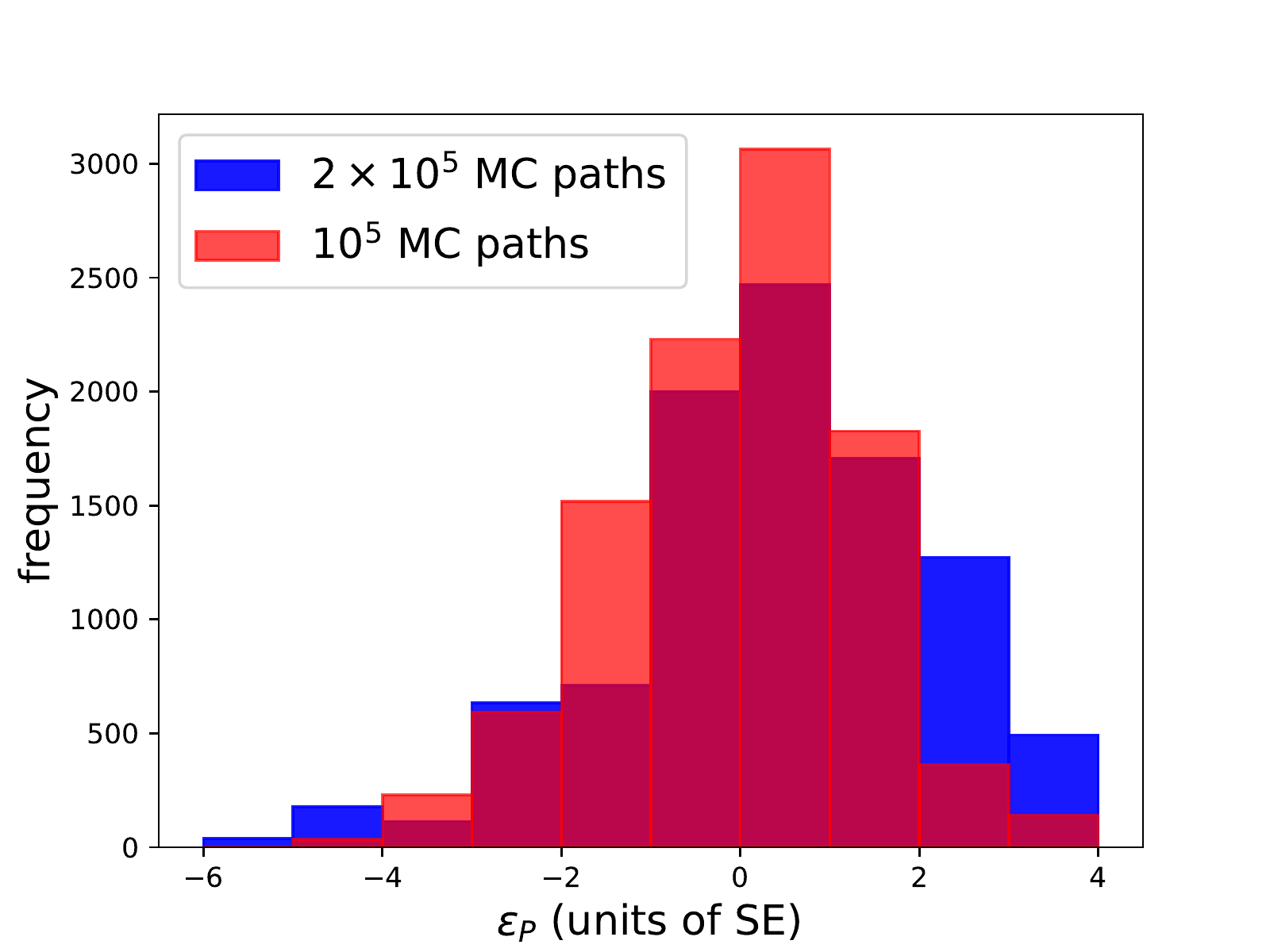} }}%
    \subfloat[\centering]{{\includegraphics[height=4.3cm,width=4.68cm]{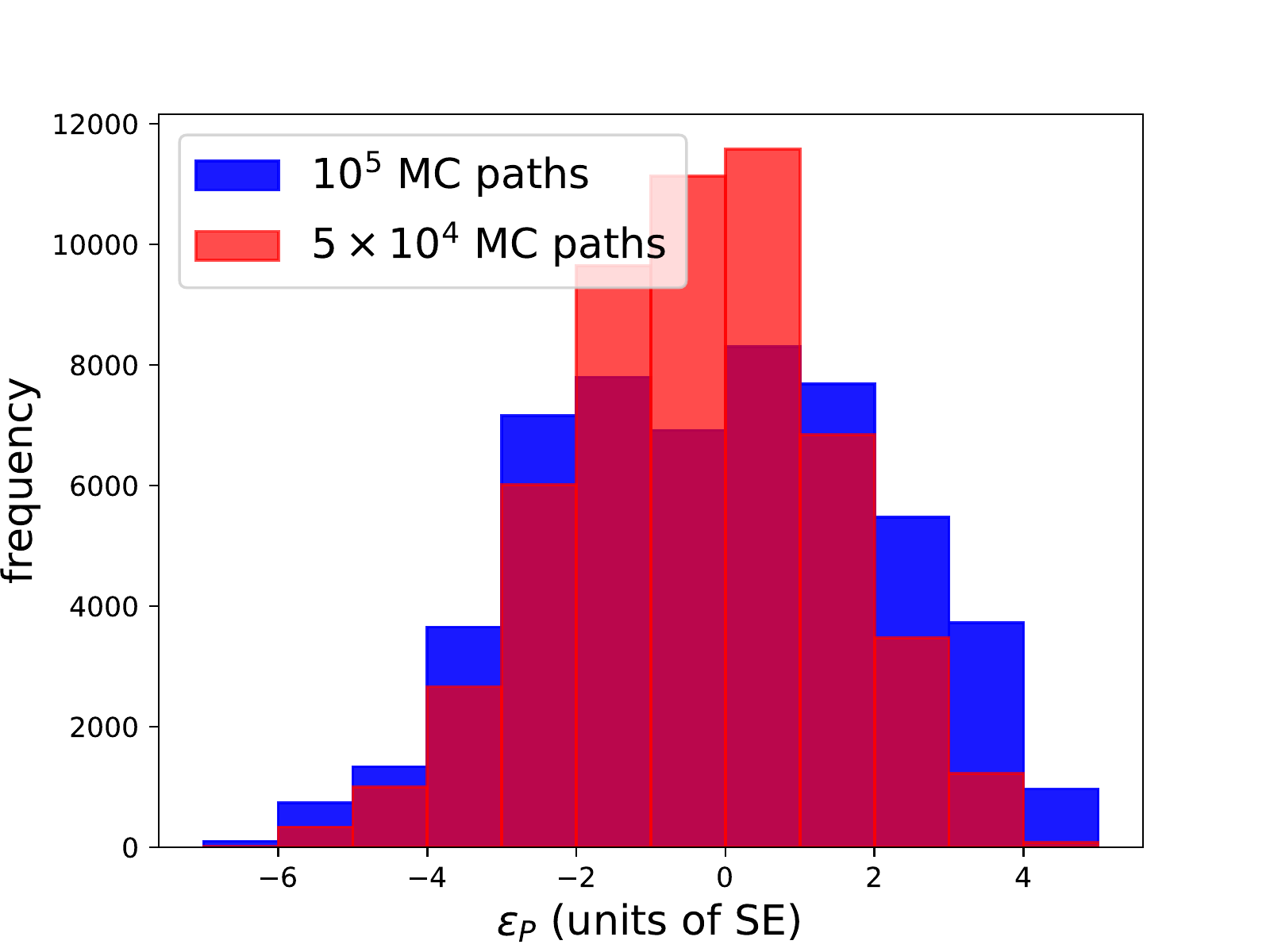} }}%
    \vspace{-0.1cm}
    \caption{\footnotesize  Pricing error $\epsilon_P$ distribution. The error is expressed in units of the corresponding Monte Carlo benchmark standard error (SE), and reported for two different numbers of MC paths. Left: FxA (10000 values). Center: FxL (10000 values). Right: FxFlA (54000 values).}%
    \label{fig: ErrorDistribution}%
\end{figure}

\subsubsection{Numerical results for FxL}

\label{sssec: LookbackSwaption}

In this section, we use the procedure to efficiently value the \emph{pipeline risk} typically embedded in mortgages. The pipeline risk (in mortgages) is one of the risks a financial institution is exposed to any time a client buys a mortgage. Indeed, when a client decides to buy a mortgage there is a \emph{grace period} (from one to three months in The Netherlands), during which (s)he is allowed to pick the most convenient rate, namely the minimum. 

Observe now that a suitable Lookback option on a payer swap, namely a Lookback payer swaption, perfectly replicates the optionality offered to the client. In other words, the ``cost'' of the pipeline risk is assessed by evaluating a proper Lookback swaption. In particular, we price fixed-strike discrete Lookback swaptions with a monitoring period of 3 months and 3-day frequency (see \ref{eqn: SpecsLook}).

We assume the underlying swap rate $S(t)$, $0\leq t\leq T$, is driven by the dynamics given in \Cref{eqn: StockDynamics,eqn: VolatilityDynamics} with $S(t_0)=0.05$ and parallel shifted of $\theta=0.03$. By introducing a shift, we handle also the possible situation of \emph{negative rates}, which otherwise would require a different model specification. 
\begin{rem}[Parallel shift of $S(t)$ and $A(S)$]
A parallel shift $\theta$ does not affect the training set generation. Indeed, since $A(S)=\min_n S(t_n)$, it holds $A(S-\theta)\overset{\d}{=}A(S)-\theta$. Then, it is enough to sample from $A(S)$ (built from the paths of $S(t)$ without shift) and perform the shift afterward, to get the desired distribution.
\end{rem}

The timing results from the application of the procedure are comparable to the ones in \Cref{sssec: FxA} (see \Cref{tab: TimingResults}). Furthermore, in \Cref{fig: ErrorDistribution}b, we report the pricing error distribution obtained by pricing call and put options for 50 randomly chosen Heston parameters' sets and 100 values for $K_2$. \textcolor{black}{In this experiment, we observe that over 95\% of the errors are within three MC SE when $N_{\texttt{paths}}=10^5$, and the percentage is about 90\% when $N_{\texttt{paths}}=2\times 10^5$.}

\subsubsection{Numerical results for FxFlA}
\label{ssec: GeneralExperiment}

\begin{figure}[t]%
\vspace{-1.cm}
    \centering
    \subfloat[\centering]{{\includegraphics[width=6.8cm]{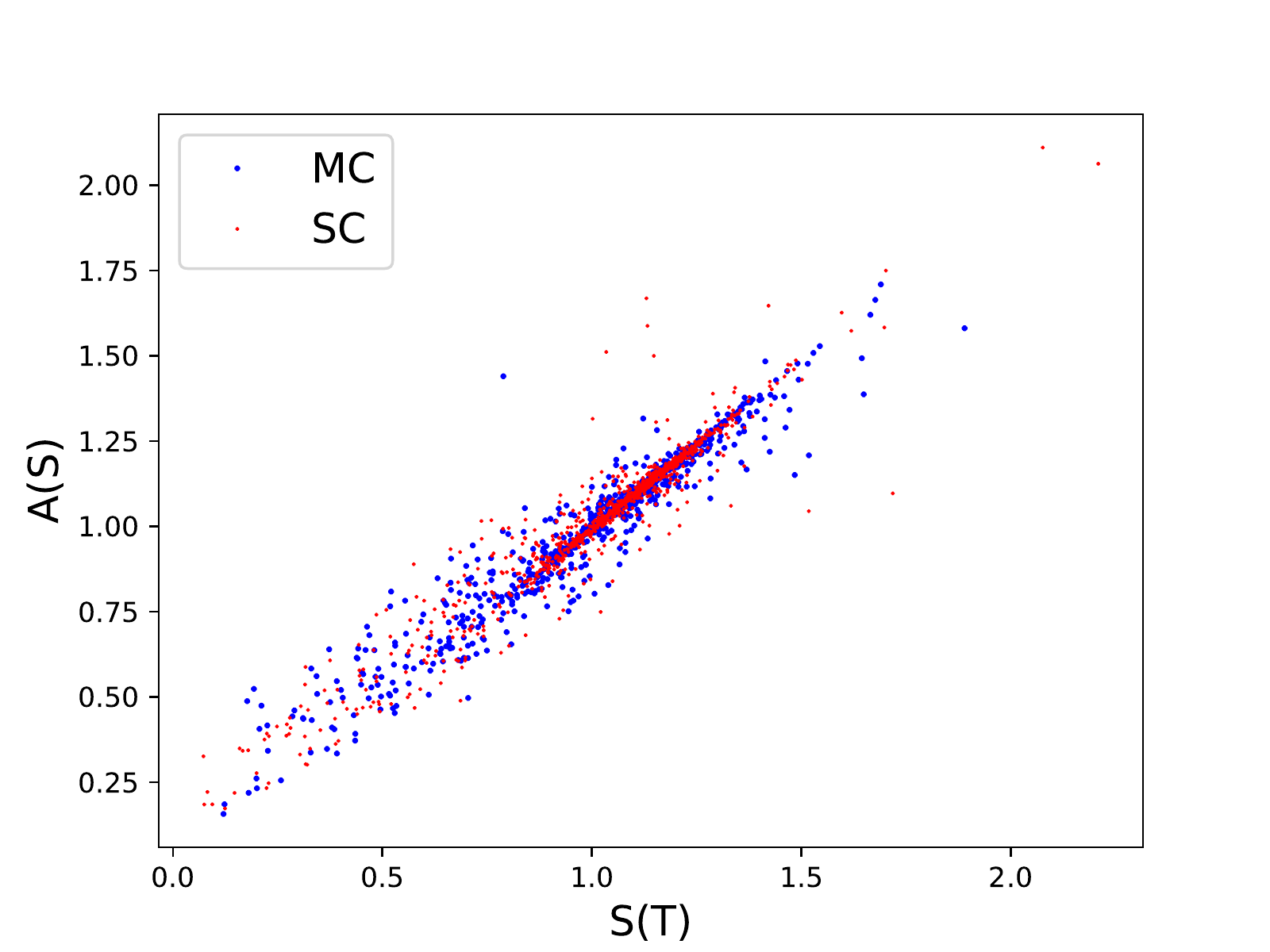} }}%
    ~
    \subfloat[\centering]{{\includegraphics[width=6.8cm]{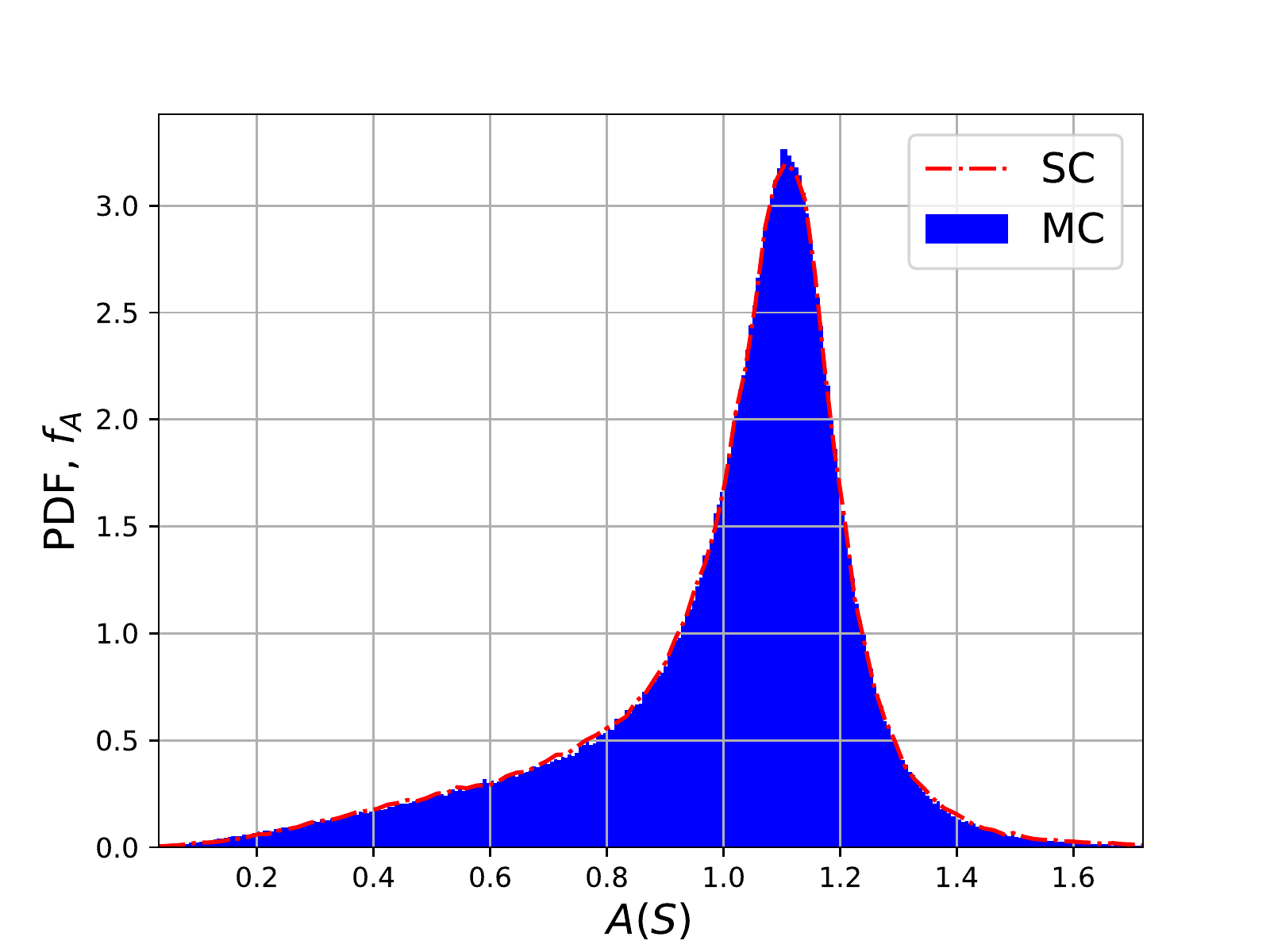} }}%
    \caption{\footnotesize Left: joint distribution of $(S(T),A(S))$, for the Heston set of parameters in Set III (see \Cref{tab: HestonSets}). Right: marginal distribution of $A(S)$, for the same set of Heston parameters.}%
    \vspace{-0.1cm}
    \label{fig: ConditionalDistribution}%
\vspace{-0.3cm}
\end{figure}

\begin{figure}[b!]
    \centering
    \subfloat[\centering]{{\includegraphics[width=6.8cm]{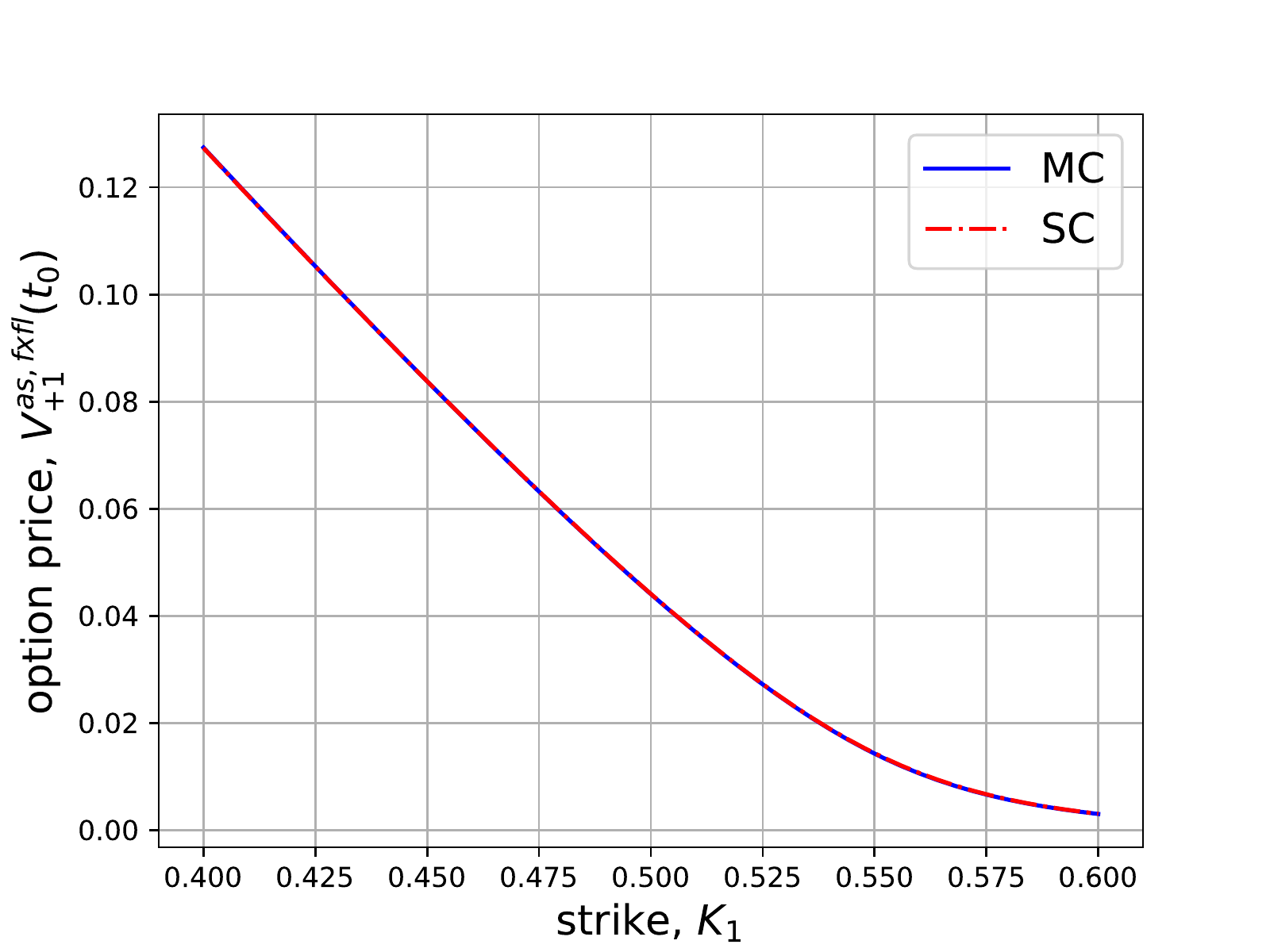} }}%
    ~
    \subfloat[\centering]{{\includegraphics[width=6.8cm]{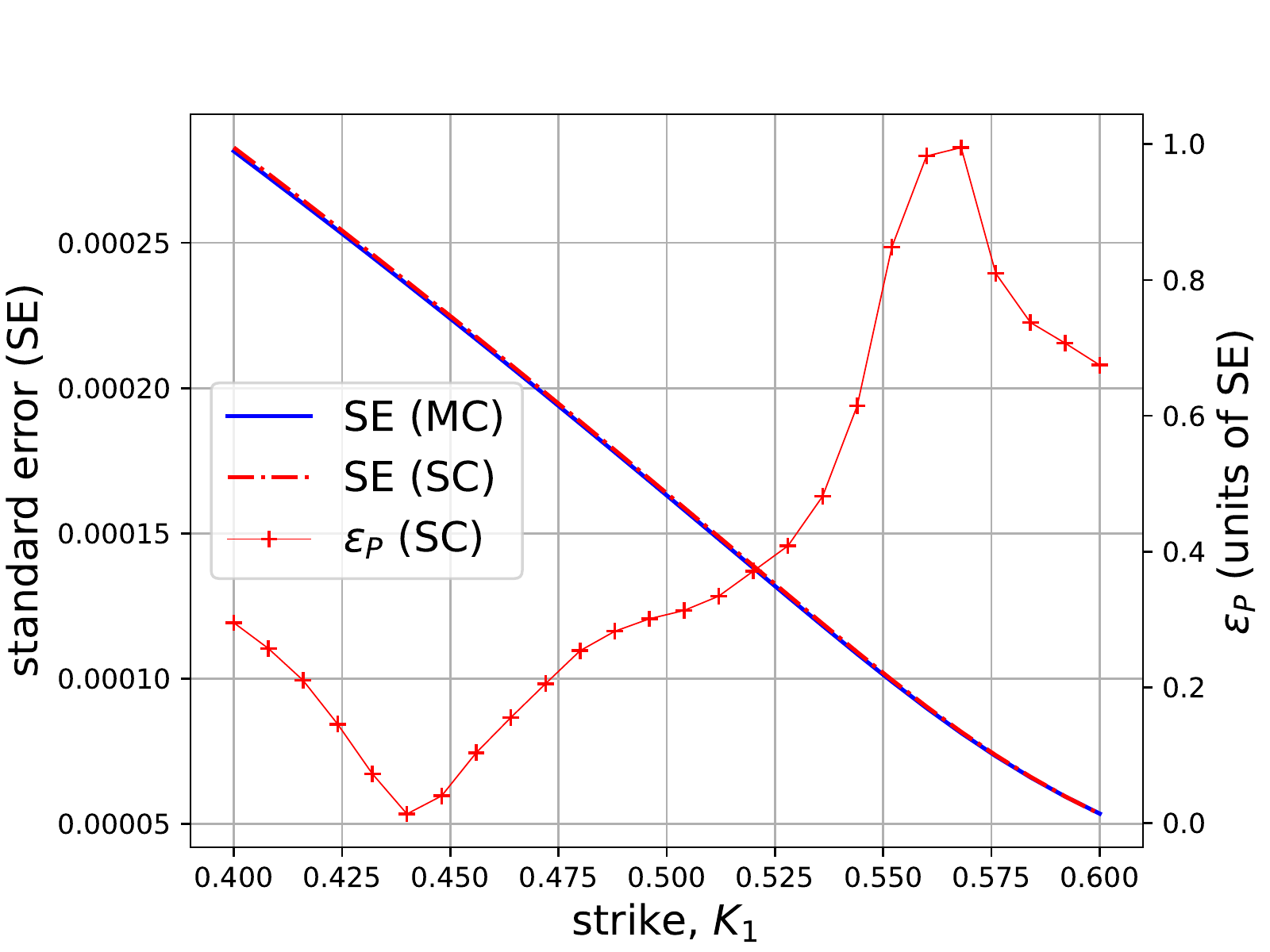} }}%
    \caption{\footnotesize Left: fixed-float-strike discrete Asian call option prices for 30 different $K_1$, $K_2=0.5$, and Heston parameter Set III. Right: MC and SC standard errors (left y-axis) and pricing errors in units of the corresponding MC standard errors SE, obtained with $N_{\texttt{paths}}=10^5$ for both MC and SC (right y-axis).}%
    \label{fig: CondAsian}%
\end{figure}

The third and last experiment consists in the conditional sampling of $A(S)|S(T)$. The samples are then used, together with $S(T)$, for pricing of fixed- and floating-strikes discrete Asian options.

The procedure is tested on 30 randomly chosen Heston parameter sets. Both the MC benchmark and the SC procedure are based on populations with $N_{\texttt{paths}}=10^5$ paths.
In the SC pricing, the process $S(T)$ is sampled using the COS method \cite{TheBook} combined with SC (COS-SC) to avoid a huge number of CDF numerical inversions \cite{StochasticCollocation}, and so increase efficiency. Then, we apply the grid-based algorithm of \Cref{ssec: GeneralPricing}. We evaluate the ANN at a reduced number of reference quantiles, and we compute the CVs corresponding to each sample of $S(T)$ by means of linear interpolation. The CVs identify the map $\g$, which is employed for the conditional sampling. \Cref{fig: ConditionalDistribution}a shows the cloud of points (for parameters' Set III in \Cref{tab: HestonSets}) of the bivariate distribution $(S(T), A(S))$ generated using the procedure against the MC benchmark, while \Cref{fig: ConditionalDistribution}b only focuses on the marginal distribution of $A(S)$. We can appreciate a good matching between the two distributions.

For each set, we price $30\times 30$ call and put options with equally-spaced strikes $K_1$ and $K_2$ in the ranges of (\ref{eqn: FxFlAPayoff}). The results for the particular case of call options with Set III in \Cref{tab: HestonSets} and $K_2=0.5$ \textcolor{black}{are illustrated in \Cref{fig: CondAsian}. The SC option prices and the corresponding MC benchmarks are plotted on the left. On the right, the standard errors for MC and SC are reported (left y-axis), and the absolute pricing error $\epsilon_P$ is shown in units of the corresponding MC standard error (right y-axis).} The timing results are reported in \Cref{tab: TimingResults} for $N_{\texttt{paths}}=10^5$ and $N_{\texttt{paths}}=5\times 10^4$ and they keep into account the computational time for the pricing of all the 900 different call options (according to each combination of $K_1$ and $K_2$). 
\textcolor{black}{\Cref{fig: ErrorDistribution}c displays the pricing error $\epsilon_P$ distribution for the 30 randomly chosen Heston parameter sets (for each set 900 call and 900 put options are priced for every combination of $K_1$ and $K_2$, so the overall number of data is 54000). About 92\% of the errors are within three MC SE when $N_{\texttt{paths}}=5\times 10^4$. The percentage is around 80\% when $N_{\texttt{paths}}=10^5$ are used.}

It might look surprising that the performance of the general procedure is better than the special one, but actually, it is not. Indeed, an important aspect needs to be accounted for. The high correlation between $S(T)$ and $A(S)$ makes the task of the ANN easier, in the sense that the distribution of $A(S)|S(T)$ typically has a low variance around $S(T)$. In other words, the ANN has to ``learn'' only a small correction to get $A(S)|S(T)$ from $S(T)$ ($S(T)$ is an input for the ANN!), whereas the ANN in the special procedure learns the unconditional distribution of $A(S)$ with no information on the final value $S(T)$, and so only on the Heston parameters. The result is that a small lack in accuracy due to a not-perfect training process, or most likely to a not-perfect training set, is less significant in the conditional case rather than in the unconditional.

\section{Conclusion}  
\label{sec: Conclusions}

In this work, we presented a robust, data-driven procedure for the pricing of fixed- and floating-strike discrete Asian and Lookback options, under the stochastic volatility model of Heston. The usage of Stochastic Collocation techniques combined with deep artificial neural networks allows the methodology to reach a high level of accuracy while reducing the computational time by tens of times when compared to Monte Carlo benchmarks. Furthermore, we provide a semi-analytic pricing formula for European-type options with a payoff given by piecewise polynomial mapping of a standard normal random variable. Such a result allows to even increase the speed-up up to hundreds of times, without deterioration on the accuracy.
An analysis of the error provides theoretical justification for the proposed scheme, and the problem of sampling from both unconditional and conditional distributions is further investigated from a numerical perspective. Finally, the numerical results provide clear evidence of the quality of the method.

\section*{Acknowledgement}
We would like to thank the two anonymous reviewers whose insightful comments and constructive feedback greatly contributed to the enhancement of the quality of this article. Additionally, we extend our appreciation to Rabobank (the Netherlands) for funding this project.

\bibliography{Paper.bib}

\appendix
\section{Proofs and lemmas}
\subsection{Underlying process measure for floating-strike options}
\label{ap: StockMeasureForfloatstrike}
\begin{proof}[Proof of Proposition \ref{prop: floatAsianStockMeasure}]
Under the risk-neutral measure $\Q$ the value at time $t_0 \geq 0$ of a floating-strike Asian Option, with maturity $T>t_0$, underlying process $S(t)$, and future monitoring dates $t_n$, $n\in\{1,\dots,N\}$, is given by:
\begin{eqnarray*}
V^{\fl}_\omega(t_0)=\E_{t_0}^{\Q}\left[\frac{M(t_0)}{M(T)} \max\Big(\omega\big({A}(S)-K_1{S}(T)\big), 0\Big)\right].
\end{eqnarray*}
We define a Radon-Nikodym derivative to change the measure from the risk-neutral measure $\Q$ to the stock measure $\Q^S$, namely the measure associated with the num\'eraire $S$:
\[\frac{\d\Q^{S}}{\d\Q}=\frac{ S(T)}{ S(t_0)}\frac{M(t_0)}{M(T)},\]
which yields the following present value, expressed as an expectation under the measure $\Q^S$:
\begin{align*}
V^{\fl}_\omega(t_0)&=\E_{t_0}^{S}\left[\frac{M(t_0)}{M(T)} \max\Big(\omega\big({A}(S)-K_1{S}(T)\big), 0\Big)\frac{ S(t_0)}{ S(T)}\frac{M(T)}{M(t_0)}\right]\\
&=S(t_0)\E_{t_0}^{S}\bigg[ \max\bigg(\omega\bigg(\frac{A(S)}{S(T)}-K_1\bigg), 0\bigg)\bigg].
\end{align*}
\end{proof}

\begin{prop}[The Heston model under the underlying process measure]
\label{prop: HestonUnderStockMeasure}
Using the same notation as in (\ref{eqn: StockDynamics}) and (\ref{eqn: VolatilityDynamics}), under the stock $S(t)$ measure, $\Q^S$, the Heston framework yields the following dynamics for the process $S(t)$:
\begin{align*}
\d S(t)&=\left(r+v(t)\right)S(t)\dt + \sqrt{v(t)}{S}(t)\dW_x^S(t), & S(t_0)&=S_0,\\
\label{eqn: VolatilityDynamicsStockMeasure}
\d v(t)&=\kappa^*(\bar{v}^*-v(t))\dt + \gamma\sqrt{v(t)}\dW_v^S(t), & v(t_0)&=v_0,
\end{align*}
with $\kappa^*=\kappa-\gamma\rho$, $\bar{v}^*=\kappa\bar{v}/\kappa^*$, and $W_x^S(t)$ and $W_v^S(t)$ are BMs under the underlying process measure $\Q^S$ such that $\d W_s^S(t)\d W_v^S(t)=\rho \dt$.
\end{prop}

\begin{proof}
Under the stock measure $\Q^S$, implied by the stock $S(t)$ as num\'eraire, all the assets discounted with $S$ must be martingales. Particularly, this entails that $M(t)/S(t)$ must be a martingale, where $M(t)$ is the money-savings account defined as $\d M(t)=rM(t)\dt$. 

From (\ref{eqn: StockDynamics}) and (\ref{eqn: VolatilityDynamics}), using Cholesky decomposition, the Heston model can be expressed in terms of independent Brownian motions, $\widetilde{W}_x(t)$ and $\widetilde{W}_v(t)$, through the following system of SDEs:
\begin{align*}
    	\d S(t)&=rS(t)\dt+\sqrt{v(t)}S(t)\d\widetilde{W}_x(t),\\
\d
v(t)&=\kappa\left(\bar{v}-v(t)\right)\dt+\gamma\sqrt{v(t)}\left[\rho\d\widetilde{W}_x(t)+\sqrt{1-\rho^2}\d\widetilde{W}_v(t)\right].
\end{align*}
After application of It\^o's Lemma, we find:
\begin{align*}
\d \frac{M(t)}{S(t)}&=\frac{1}{S(t)}rM(t)\dt -\frac{M(t)}{S^2(t)}\left(rS(t)\dt +\sqrt{v(t)}S(t)\d\widetilde{W}_x(t)\right)+\frac{M(t)}{S^3(t)}v(t)S^2(t) \dt,
\end{align*}
which implies the following measure transformation:
\[\d\widetilde{W}_x(t)=\d\widetilde{W}_x^S(t)+\sqrt{v(t)}\dt.\]
Thus, under the stock measure $\Q^S$, the dynamics of $S(t)$ reads:
\begin{align*}
\d S(t)&=rS(t)\dt + \sqrt{v(t)}{S}(t)\left(\d\widetilde{W}_x^S(t)+\sqrt{v(t)}\dt\right)\\
&=\big(r+v(t)\big)S(t)\dt + \sqrt{v(t)}{S}(t)\d\widetilde{W}_x^S(t),
\end{align*}
while for the dynamics of $v(t)$ we find:
\begin{align*}
\d
v(t)&=\kappa\left(\bar{v}-v(t)\right)\dt+\gamma\sqrt{v(t)}\left[\rho\left(\d\widetilde{W}_x^S(t)+\sqrt{v(t)}\dt\right)+\sqrt{1-\rho^2}\d\widetilde{W}_v(t)\right]\\
&=\left[\kappa\left(\bar{v}-v(t)\right)+\gamma \rho v(t)\right]\dt+\gamma\sqrt{v(t)}\left[\rho\d\widetilde{W}_x^S(t)+\sqrt{1-\rho^2}\d\widetilde{W}_v(t)\right].
\end{align*}
Setting $\kappa^*:=\kappa-\gamma\rho$, $\bar{v}^*:=\kappa\bar{v}/\kappa^*$, ${W}_x^S(t):=\widetilde{W}_x^S(t)$, and $W_v^S(t):=\rho\widetilde{W}_x^S(t)+\sqrt{1-\rho^2}\widetilde{W}_v(t)$ the proof is complete.
\end{proof}

\subsection{Semi-analytic pricing formula}
\label{ap: analyticPricing}
\begin{result}[Moments of truncated standard normal distribution]
\label{res: MomentsTruncatedStandardNormal}
Let $\xi\sim \N(0,1)$ and $a,b\in [-\infty,+\infty]$, $a < b$. Then, the recursive expression for:
\begin{equation*}
    m_i(a,b):=\E[\xi^i|a \leq \xi \leq b],
\end{equation*}
the $i$-th moment of the truncated standard normal distribution $\xi|a \leq \xi \leq b$, reads:
\begin{equation*}
    m_i(a,b)=(i-1)m_{i-2}(a,b)-\frac{b^{i-1}f_{\xi}(b)-a^{i-1}f_{\xi}(a)}{F_{\xi}(b)-F_{\xi}(a)}, \qquad i\in\mathbb{N}\backslash \{0\},
\end{equation*}
where $m_{-1}(a,b):=0$, $m_0(a,b):=1$, and $f_{\xi}$ and $F_{\xi}$ are the PDF and the CDF of $\xi$, respectively.
\end{result}
\begin{result}[Expectation of polynomial of truncated normal distribution]
\label{res: IntegralPolyNormal}
Let $p(x)=\sum_{i=0}^{M-1} \alpha_i x^i$ be a $(M-1)$-degree polynomial and let $\xi\sim\N(0,1)$, with $f_\xi$, $F_\xi$ its PDF and CDF, respectively. Then, for any $a,b\in[-\infty,+\infty]$ with $a < b$, the following holds:
\begin{align}
    \int_a^b p(x)f_\xi(x)\d x = \sum_{i=0}^{M-1} \alpha_i m_i(a,b) \left(F_\xi(b) - F_\xi(a)\right),
\end{align}
with $m_i(a,b)$ as defined in \Cref{res: MomentsTruncatedStandardNormal}.
\end{result}
\begin{proof}
The proof immediately follows thanks to the following equalities:
\begin{equation}
    \begin{aligned}
    \int_a^b p(x)f_\xi(x)\d x &= \sum_{i=0}^{M-1}\alpha_i \int_a^b x^i f_\xi(x) \d x = \sum_{i=0}^{M-1}\alpha_i \E[\xi^i\1_{[a, b]}(\xi)]\\
    &=\sum_{i=0}^{M-1}\alpha_i \E[\xi^i|a \leq \xi \leq b] \cdot \P[{a\leq \xi\leq b}].
\end{aligned}
\end{equation}
\end{proof}
\begin{proof}[Proof of Proposition \ref{prop: semianalyticPrice}]
The approximation $\g$ is strictly increasing in the domain of interest. Then, setting $c_K= \g^{-1}(K)$, we have:
\begin{equation}
\begin{aligned}
    \frac{\Tilde{V}_\omega(t_0)}{C} &= \int_{-\infty}^{+\infty} \max\Big(\omega\big(\g(x)-K\big),0\Big) f_{\xi}(x) \d x\\
    &=\int_{\omega c_K}^{+\infty}  \omega\big( \g(\omega y) - K\big) f_{\xi}(\omega y) \d y\\
    &=\omega \left(\int_{\omega c_K}^{+\infty} \g(\omega y)f_{\xi}(y) \d y - K\P[\xi>\omega c_K]\right),
\end{aligned}    
\end{equation}
where the first equality holds by definition of expectation, the second one relies on a suitable change of variable ($y=-x$) and the last one holds thanks to the even symmetry of $f_{\xi}$.
We define the integral $I_\omega(c_K)$ as:
\begin{equation*}
    I_\omega(c_K):=\int_{\omega c_K}^{+\infty} \g(\omega x)f_{\xi}(x) \d x,
\end{equation*}
and using the definition of $\g$ as a piecewise polynomial, we get:
\begin{equation}
\label{eqn: DefinitionIntegral}
    I_\omega(c_K) = \int_{\omega c_K}^{-\xib \vee c_K} g_{-\omega}(\omega x)f_\xi(x)\d x + \int_{-\xib \vee c_K}^{\xib \vee c_K} g_{M}(\omega x)f_\xi(x)\d x + \int_{\xib \vee c_K}^{+\infty} g_{\omega}(\omega x)f_\xi(x)\d x.
\end{equation}
The thesis follows by applying \Cref{res: IntegralPolyNormal} at each term in (\ref{eqn: DefinitionIntegral}) and exploiting the definition of $F_\xi$.
\end{proof}

\subsection{Almost Exact Simulation from the Heston Model}
\label{ap: ExactSimulationHeston}

In a MC framework, the most common scheme employed in the industry is the Euler-Maruyama discretization of the system of SDEs which describes the underlying process dynamics. For the stochastic volatility model of Heston, such a scheme can be improved, allowing for an \emph{exact} simulation of the variance process $v(t)$ (see (\ref{eqn: VolatilityDynamics})), as shown in \citep*{ExactSimulationHeston}. This results in increased accuracy, and avoids numerical issues due to the theoretical non-negativity of the process $v(t)$, leading to the so-called \emph{almost exact simulation} of the Heston model \cite{TheBook}.

\begin{result}[Almost exact simulation from the Heston Model]
\label{res: AESimulation}
Given $X(t):=\log S(t)$, its dynamics\footnote{Under the risk-neutral measure, $\Q$. However, the same scheme applies under the underlying process measure $\Q^S$, with only a minor difference, i.e., $k_1:=\left(\rho\kappa/{\gamma}+1/2\right)\Delta t-{\rho}/{\gamma}$.} between the consequent times $t_i$ and $t_{i+1}$ is discretized with the following scheme:
\begin{equation}
\label{eqn: ExactSimulationHeston}
    \begin{aligned}
    x_{i+1}&\approx x_{i}+k_0+k_1v_{i}+k_2v_{i+1}+\sqrt{k_3v_{i}}\xi,\\
    v_{i+1}&= \bar c \chi^2(\delta,\bar \kappa v_{i})
    \end{aligned}
\end{equation}
with the quantities:
\begin{align*}
    \Delta t&:=t_{i+1}-t_1,
	&\delta&:=\frac{4\kappa\bar{v}}{\gamma^2},
	&\bar c&:=\frac{\bar{v}}{\delta}(1-\e^{-\kappa\Delta t}),
	&\bar \kappa&:=\bar c ^{-1}\e^{-\kappa \Delta t},
\end{align*}
the noncentral chi-squared
random variable $\chi^2(\delta,\eta)$ with \index{degrees of freedom parameter}$\delta$ degrees of freedom and non-centrality parameter $\eta$, and $\xi\sim\N(0,1)$. The remaining constants are defined as:
\begin{align*}
	k_0&:=\bigg(r-\frac{\rho}{\gamma}\kappa\bar{v}\bigg)\Delta t, & k_1&:=\left(\frac{\rho\kappa}{\gamma}-\frac12\right)\Delta t-\frac{\rho}{\gamma},
	& k_2&:=\frac{\rho}{\gamma}, & k_3&:=(1-\rho^2)\Delta t.
\end{align*}
\end{result}

\begin{proof}[Derivation]
Given $X(t)=\log S(t)$, by applying It\^o's Lemma and Cholesky decomposition on the dynamics in (\ref{eqn: StockDynamics}) and (\ref{eqn: VolatilityDynamics}) , we get:
\begin{align}
    \label{eqn: LogStockDynamicsIndepGeneric}
	\d X(t)&=\left(r-\frac12v(t)\right)\dt+\sqrt{v(t)}\left[\rho\d\widetilde{W}_v(t)+\sqrt{1-\rho^2}\d\widetilde{W}_x(t)\right],\\
    \label{eqn: VolatilityDynamicsIndepGeneric}
    \d v(t)&=\kappa\left(\bar{v}-v(t)\right)\dt+\gamma\sqrt{v(t)}\d\widetilde{W}_v(t),
\end{align}
where $\widetilde{W}_x(t)$ and $\widetilde{W}_v(t)$ are independent BMs.

By integrating (\ref{eqn: LogStockDynamicsIndepGeneric}) and (\ref{eqn: VolatilityDynamicsIndepGeneric}) in a
the time interval $[t_{i},t_{i+1}]$, the following
discretization scheme is obtained:
\begin{align}
\label{eqn: AEDiscr1}
	x_{i+1}&=x_{i}+\int_{t_{i}}^{t_{i+1}}\bigg(r-\frac12v(t)\bigg)\dt+\rho\; {\int_{t_{i}}^{t_{i+1}}\sqrt{v(t)}\d\widetilde{W}_v(t)}+\sqrt{1-\rho^2}\int_{t_{i}}^{t_{i+1}}\sqrt{v(t)}\d\widetilde{W}_x(t),\\
	\label{eqn: AEDiscr2}
	v_{i+1}&=v_{i}+\kappa\int_{t_{i}}^{t_{i+1}}\left(\hat{v}-v(t)\right)\dt+\gamma\; {\int_{t_{i}}^{t_{i+1}}\sqrt{v(t)}\d\widetilde{W}_v(t)},
\end{align}
where $x_i:=X(t_i)$, $x_{i+1}:=X(t_{i+1})$, $v_i:=v(t_i)$, $v_{i+1}:=v(t_{i+1})$.

Given $v_i$, the variance $v_{i+1}$ is distributed as a suitable scaled noncentral chi-squared distribution \cite{TheBook}. Therefore, we substitute ${\int_{t_{i}}^{t_{i+1}}\sqrt{v(t)}\d\widetilde{W}^\lambda_v(t)}$ in (\ref{eqn: AEDiscr1}) using (\ref{eqn: AEDiscr2}), ending up with:
\begin{equation*}
\begin{aligned}
    x_{i+1}&=x_{i}+\int_{t_{i}}^{t_{i+1}}\left(r-\frac12v(t)\right)\dt+\frac{\rho}{\gamma}\left(v_{i+1}-v_{i}-\kappa\int_{t_{i}}^{t_{i+1}}\left(\bar{v}-v(t)\right)\dt\right)\nonumber\\
	&+\sqrt{1-\rho^2}\int_{t_{i}}^{t_{i+1}}\sqrt{v(t)}\d\widetilde{W}_x(t).
\end{aligned}
\end{equation*}
We approximate the integrals in the expression above employing the left integration boundary values of the integrand, as in the Euler-Maruyama discretization scheme.
The scheme (\ref{eqn: ExactSimulationHeston}) follows collecting the terms and employing the property  $\widetilde{W}_x(t_{i+1})-\widetilde{W}_x(t_{i})\overset{\d}{=}\sqrt{\Delta
t}\xi$, with $\xi\sim\mathcal{N}(0,1)$ and $\Delta t:=t_{i+1} - t_i$.
\end{proof}

\subsection{SC error analysis for Chebyshev interpolation}
\label{ap: ErrorAnalysis}

The two following lemmas are useful to show that the conditional complex ChF $\phi_{A|\V}(z)=\E[e^{izA}|\V]$ is an analytic function of $z\in\C$. The first one provides the law of the conditional stock-price distribution, whereas the second one is meant to give algebraic bounds for the target function $A(S)$.

\begin{lem}[Conditional distribution under Heston]
\label{lem: ConditionalLognormalStock}
Let $S(t)$ be the solution at time $t$ of \Cref{eqn: StockDynamics} and $I_v(t):=\int_{t_0}^{t}v(\tau)\d\tau$, with $v$ driven by the dynamics in \Cref{eqn: VolatilityDynamics}. Then, the following equality in distribution holds:
\begin{equation*}
\label{eqn: ConditionalStockDistribution}
    S(t)\Big|I_v(t_0,t)\overset{\d}{=} \exp\Big(\mu(I_v(t_0,t))+\sigma(I_v(t_0,t))\xi\Big),
\end{equation*}
with $\xi\sim\N(0,1)$, $\mu$ and $\sigma$ defined as $\mu(y):=\log S(t_0) + r(t - t_0)-y/2$ and $\sigma(y):=\sqrt{y}$.
Furthermore, for any $k=0,1,\dots$, the following holds:
\begin{equation}
\label{eqn: MomentsLognormal}
    \E[S(t)^k|I_v(t)]=\exp{\left(k\mu(I_v(t))+\frac12 k^2 \sigma^2(I_v(t))\right)}
\end{equation}
In other words, the stock price given the time-integral of the variance process $I_v$ is log-normally distributed, with parameters dependent on the time-integral $I_v$, and its moments up to any order are given as in \Cref{eqn: MomentsLognormal}.
\end{lem}
\begin{proof}
Writing (\ref{eqn: StockDynamics}) in integral form we get:
\begin{equation*}
    S(t)=S(t_0)\exp \left(r(t - t_0)-\frac{1}{2}\int_{t_0}^{t}v(\tau)\d\tau+\int_{t_0}^{t}\sqrt{v(\tau)}\dW_x(\tau)\right).
\end{equation*}
By considering the conditional distribution $S(t)|I_v(t)$ (instead of $S(t)$) the only source of randomness is given by the It\^o's integral (and it is due to the presence of the Brownian motion $W_x(t)$). The thesis follows since the It\^o's integral of a deterministic argument is normally distributed with zero mean and variance given by the time integral of the argument squared (in the same interval). Therefore, $S(t)|I_v(t)$ is log-normally distributed, with moments given as in (\ref{eqn: MomentsLognormal}).
\end{proof}

\begin{lem}[Algebraic bounds]
\label{lem: AlgebraicBounds}
Let us consider $\{s_1,\dots,s_N\}$, with $s_n>0$ for each $n=1,\dots, N$. Then, for any $k=1,2,\dots$, we have:
\begin{enumerate}
    \item $\big(\sum_n s_n\big)^k \leq 2^{(N-1)(k-1)}\sum_n s_n^k$.
    \item $\big(\min_n s_n\big)^k \leq s_{n^*}^k$ for any $n^*=1,\dots,N$.
\end{enumerate}
\end{lem}
\begin{proof}
The second thesis is obvious. We prove here the first one.
We recall that in general, given $a,b>0$ and any $k = 1,2,\dots$, the following inequality holds: 
    \begin{equation}
    \label{eqn: IneqConvexPower}
        (a+b)^k\leq 2^{k-1}(a^k + b^k).
    \end{equation}
    Then, applying (\ref{eqn: IneqConvexPower}) $N-1$ times we get:
    \begin{equation*}
        \left(\sum_{n=1}^N s_n\right)^k \leq 2^{k-1}\left(s_1^k + \left(\sum_{n=2}^N s_n\right)^k\right)\leq \cdots\leq 2^{(N-1)(k-1)}s_N^k + \sum_{n=1}^{N-1}2^{n(k-1)}s_n^k, 
    \end{equation*}
    which can be further bounded by:
    \begin{equation*}
        \left(\sum_{n=1}^N s_n\right)^k \leq 2^{(N-1)(k-1)}s_N^k + \sum_{n=1}^{N-1}2^{(N-1)(k-1)}s_n^k = \sum_{n=1}^{N}2^{(N-1)(k-1)}s_n^k.
    \end{equation*}
\end{proof}

We have all the ingredients to prove \Cref{prop: EntireConditionalChF}.

\begin{proof}[Proof of \Cref{prop: EntireConditionalChF}]
To exploit the characterization for entire ChFs in \Cref{thm: CharacterizationEntireChF} we need to show the finiteness of each absolute moment as well as that \Cref{eqn: LimitMomentsChF} is satisfied. Both the conditions can be proved using \Cref{lem: ConditionalLognormalStock} and \Cref{lem: AlgebraicBounds}. For $k=0,1,\dots$, we consider the two cases:
\begin{enumerate}
    \item If $A=\frac{1}{N}\sum_n S(t_n)$, then thanks to \Cref{lem: AlgebraicBounds} we have:
    \begin{align*}
        \E[|A|^k|\V]= \frac{1}{N^k}\E\Big[\Big(\sum_n S(t_n)\Big)^k\Big|\V\Big] & \leq \frac{2^{(N-1)(k-1)}}{N^k}\E\Big[\sum_n S(t_n)^k\Big|\V\Big]\\
        & = \frac{2^{(N-1)(k-1)}}{N^k}\sum_n\E[S(t_n)^k|\V],
    \end{align*}
    whereas from \Cref{lem: ConditionalLognormalStock} follows:
    \begin{align}
        \E[|A|^k|\V]&\leq \frac{2^{(N-1)(k-1)}}{N^k}\sum_n\E[S(t_n)^k|\V]\\
        \label{eqn: BoundAsianMoments}
        &= \frac{2^{(N-1)(k-1)}}{N^k}\sum_n \exp\left(k\mu_n(\V)+\frac12 k^2 \sigma^2_n(\V)\right),
    \end{align}
    where $\mu_n(\V):=\mu(I_v(t_n))$ and $\sigma_n(\V):=\sigma(I_v(t_n))$, $n=1,\dots,N$.
    \item If $A=\min_n S(t_n)$, then we immediately have:
    \begin{align}
        \E[|A|^k|\V]&\leq \E[S(t_{n^*})^k|\V]\\
        \label{eqn: BoundLookbackMoments}
        &=  \exp\left(k\mu_{n^*}(\V)+\frac12 k^2 \sigma^2_{n^*}(\V)\right),
    \end{align}
    for an arbitrary $n^*=1,\dots,N$.
\end{enumerate}
The finiteness of the absolute moments up to any order follows directly from \Cref{eqn: BoundAsianMoments} and \Cref{eqn: BoundLookbackMoments} respectively, since $I_v(t_n)$ are finite (indeed, they are time-integrals on compact intervals of continuous paths).

Eventually, thanks to Jensen's inequality we have $|\E[A|\V]|^k\leq \E[|A|^k|\V]$. This, together with the at-most exponential growth (in $k$) of the absolute moments of $A|\V$, ensures that the limit in \Cref{eqn: LimitMomentsChF} holds. Then, by \Cref{thm: CharacterizationEntireChF}, $\phi_{A|\V}(z)$ is an entire function of the complex variable $z\in\C$.
\end{proof}

\begin{proof}[Proof of \Cref{prop: ChFAnalytic}]
\label{pr: ChFAnalytic}
    The goal here is to apply Morera's theorem. Hence, let $\gamma\in \mathcal{S}_{y^*}$ be any piecewise $C^1$ closed curve in the strip $\mathcal{S}_{y^*}$. Then:
\begin{align*}
    \int_\gamma \phi_A(z)\d z &\overset{(\ref{eqn: ChFDecomposition})}{=}\int_\gamma \int_{\Omega_\V} \phi_{A|\V=\v}(z)\d F_\V(\v) \d z\\
    &\overset{\text{Fubini}}{=}\int_{\Omega_\V} \int_\gamma \phi_{A|\V=\v}(z)\d z\d F_\V(\v)\overset{\text{Cauchy}}{=} 0,
\end{align*}
where in the first equality we exploited the representation of the unconditional ChF $\phi_A$ in terms of conditional ChFs $\phi_{A|\V}$, in the second equality we use Fubini's theorem to exchange the order of integration, and eventually in the last equation we employ the Cauchy's integral theorem on $\int_\gamma \phi_{A|\V=\v}(z)\d z$.
\end{proof}

\end{document}